\documentclass[11pt]{article}
\usepackage{fullpage}
\usepackage[margin=1in]{geometry}
\usepackage{subfig}
\usepackage[usenames,dvipsnames]{xcolor}
\usepackage[linktocpage=true,
	pagebackref=true,colorlinks,
	linkcolor=BrickRed,citecolor=blue]
	{hyperref}
	\usepackage{latexsym,graphicx,epsfig,color}
\usepackage{amsfonts,amssymb,amsmath,amsthm,amstext}
\usepackage{enumitem}
\setitemize{itemsep=2pt,topsep=0pt,parsep=0pt}
\usepackage{url,setspace}
\usepackage{multirow}
\usepackage{rotating}
\usepackage{comment}
\usepackage{makeidx}
\usepackage{tikz}
\usepackage{xspace}
\usepackage{algorithm,algpseudocode}
\usepackage{bm}
\usepackage{changepage} 	%\begin{adjustwidth}{0.5cm}{} \end{adjustwidth}
\usepackage{thmtools,thm-restate}
\usepackage{booktabs}

\usepackage{mdframed}

\usepackage{cleveref}

\newtheorem{thm}{Theorem}[section]
\newtheorem{claim}[thm]{Claim}

\newtheorem{lemma}[thm]{Lemma}
\newtheorem{cor}[thm]{Corollary}
\newtheorem{conjecture}[thm]{Conjecture}

\newtheorem{definition}[thm]{Definition}

\newtheorem{obs}[thm]{Observation}

\newtheorem{proc}[thm]{Procedure}
\newtheorem{assumption}[thm]{Assumption}

\newcommand{\E}[0]{\ensuremath \mathbb{E}}
\newcommand{\R}[0]{\ensuremath \mathbb{R}}
\newcommand{\e}{\varepsilon}
\newcommand{\ip}[2]{\langle #1, #2 \rangle}

\newcommand{\F}{\mathcal{F}}
\newcommand{\G}{\mathcal{G}}

\newcommand{\ones}{\bm{1}}
\newcommand{\blue}[1]{{\color{blue} #1}}

\newcommand{\OPT}{\mathrm{OPT}}
\newcommand{\ALG}{\mathrm{ALG}}
\newcommand{\poly}{\mathrm{poly}}

\newcommand{\cA}{\mathcal{A}}

\newcommand{\numberofsteps}{{\ensuremath{{T}}}}
\newcommand{\numberofdimensions}{\ensuremath{{n}}}

\newcommand{\monotone}{\ensuremath{\mathsf{Supermodular}}\xspace}
\newcommand{\monotonicity}{\ensuremath{\mathsf{Supermodularity}}\xspace}

\newcommand{\final}{\ensuremath{\mathrm{final}}}

\renewcommand{\d}{\textrm{d}}
\newcommand{\cover}{\textsc{OnlineCover}\xspace}
\newcommand{\pack}{\textsc{OnlinePacking}\xspace}
\newcommand{\cost}{\textrm{cost}}
\newcommand{\supp}{supp}
\newcommand{\OLO}{\textsc{OLO}\xspace}
\newcommand{\stoch}{\textsc{StochProbing}\xspace}
\newcommand{\ballopt}{\textsf{Ball-Optimization}\xspace}

\newcommand{\adapt}{\textsc{Adapt}\xspace}
\newcommand{\NA}{\textsc{NonAdapt}\xspace}

\newcommand{\topk}{\ensuremath{\textrm{Top-k}}\xspace}
\newcommand{\topx}[1]{\ensuremath{\textrm{Top-}#1}}

\DeclareMathOperator{\argmax}{argmax}
\DeclareMathOperator{\argmin}{argmin}

\newcommand{\vertm}[1]{{\left\vert\kern-0.25ex\left\vert\kern-0.25ex\left\vert #1     \right\vert\kern-0.25ex\right\vert\kern-0.25ex\right\vert}}

\newcommand{\red}[1]{{\color{red} #1}}

\newcommand{\tkc}[1]{\textcolor{Plum}{Thomas: #1}}
\newcommand{\snote}[1]{\textcolor{red}{Sahil: #1}}
\newcommand{\mnote}[1]{\textcolor{blue}{Marco: #1}}

\renewcommand{\red}[1]{#1}
\renewcommand{\tkc}[1]{}
\renewcommand{\snote}[1]{}
\renewcommand{\mnote}[1]{}

\title{Supermodular Approximation of Norms and Applications}

	\author{Thomas Kesselheim\thanks{
	(thomas.kesselheim@uni-bonn.de)
 Institute of Computer Science,	    University of Bonn. 
    }
	\and Marco Molinaro\thanks{
         (mmolinaro@microsoft.com)
         Microsoft Research and PUC-Rio.   
Supported in part by the Coordenação de Aperfeiçoamento de Pessoal de Nível Superior - Brasil (CAPES) - Finance Code 001, and by Bolsa de Produtividade em Pesquisa $\#3$12751/2021-4 from CNPq.
    }
	\and Sahil Singla\thanks{
        (ssingla@gatech.edu)
        School of Computer Science,
        Georgia Tech.
        Supported in part by NSF award CCF-2327010.
        }
}
\begin{document}
\maketitle

%\thispagestyle{empty}
%###########################################################################
%###########################################################################
%###########################################################################
%###########################################################################

\begin{abstract}
 \bigskip

Many classical problems in theoretical computer science involve norm, even if implicitly; for example, both XOS functions and downward-closed sets are equivalent to some norms. The last decade has seen a lot of interest in designing algorithms beyond the standard $\ell_p$ norms $\|\cdot \|_p$. Despite notable advancements, many existing methods remain tailored to specific problems, leaving a broader applicability to general norms less understood. This paper investigates the intrinsic properties of $\ell_p$ norms that facilitate their widespread use and seeks to abstract these qualities to a more general setting. 

We identify \emph{supermodularity}—often reserved for combinatorial set functions and characterized by monotone gradients—as a defining feature beneficial for $ \|\cdot\|_p^p$. We introduce the notion of $p$-supermodularity for norms, asserting that a norm is $p$-supermodular if its $p^{th}$ power function exhibits supermodularity. The association of supermodularity with norms offers a new lens through which to view and construct algorithms.
   
    Our work demonstrates  that for a large class of problems $p$-supermodularity is a sufficient criterion for developing good algorithms.     
    This is either by reframing existing algorithms for problems like Online Load-Balancing and Bandits with Knapsacks through a supermodular lens, or by introducing novel analyses for problems such as Online Covering, Online Packing, and Stochastic Probing.  Moreover, we prove that every symmetric norm can be approximated by a $p$-supermodular norm. Together, these recover and extend several results from the literature, and support $p$-supermodularity as a unified theoretical framework for optimization challenges centered around norm-related problems.
    
 %   demonstrate that $p$-supermodular approximations of norms are  possible for large classes of norms, especially for all symmetric norms,  thereby  proposing a unified theoretical framework for optimization challenges centered around norm objectives.

\end{abstract}

%##########################################################################
%##########################################################################
%##########################################################################
%##########################################################################
\newpage

\setcounter{tocdepth}{2}
 {\small
 %\begin{spacing}{0}
    \tableofcontents
 %\end{spacing}
 }

\clearpage
%\pagenumbering{arabic} 

\section{Introduction}
Many classical problems in theoretical computer science are framed in terms of optimizing norm objectives. For instance, Load-Balancing involves minimizing the maximum machine load, which is an $\ell_\infty$ objective, while Set Cover aims at minimizing the $\ell_1$ objective, or the number of selected sets.   
However,  contemporary applications, such as 
energy-efficient scheduling \cite{Albers-CACM10}, network routing \cite{GKP-WAOA12}, paging \cite{MS-SPAA15}, and budget allocation \cite{AD-SODA15},  demand algorithms that are capable of handling more complex objectives. Norms also underline other seemingly unrelated concepts in computer science, such as XOS functions from algorithmic game theory (both are max of linear functions) and downward-closed constraints from combinatorial optimization (the downward-closed set corresponds to the unit ball of the norm); these connections are further discussed in \Cref{sec:newApplications}. %\blue{[Add cite? Point to sec where we define them]}

 Hence,  ongoing efforts have focused on designing    good algorithms for general norm objectives.  
 Notably,  the last decade has seen a lot of progress in this direction for the class of \emph{symmetric norms}---those invariant to coordinate permutations.  Examples include   $\ell_p$ norms,  \topk norm,  and Orlicz norms.   
  They offer rich possibilities,  e.g., enabling the simultaneous capture of multiple symmetric norm objectives, as their maximum is also a symmetric norm. 
  We have seen the fruit of this in algorithms for a range of applications like 
  Load-Balancing \cite{CS-STOC19, ChakrabartyS19b},  
   Stochastic Probing \cite{PRS-APPROX23}, Bandits with Knapsacks \cite{KMS-SODA23},    clustering \cite{CS-STOC19,ChakrabartyS19b},  
nearest-neighbor search \cite{ANNNorms,ANNHolder},  and linear regression \cite{orliczRegression1,orliczRegression2}.

 Despite the above progress, our understanding of applying algorithms beyond $\ell_p$ norms remains incomplete.  For instance,  while  \cite{AzarBCCCG0KNNP16}  (where 3 independent papers were merged)  provide an algorithm for Online Cover with $\ell_p$ norms,  which was extended to  sum of $\ell_p$ norms  in \cite{NS-ICALP17},   the extension to general symmetric norms is unresolved.   Indeed,  \cite{NS-ICALP17} posed as an open question whether good Online Cover algorithms exist for more general norms.   
 Other less understood applications with  norms include Online Packing \cite{BN-MOR09} and Stochastic Probing \cite{GNS-SODA17}. 

A notable limitation  of current techniques  extending beyond $\ell_p$ norms is that they are often ad-hoc. Our aim is to create a unified framework that provides a better understanding of norms in this context, simplifies proofs, and enhances generalizability. 
\begin{quote}
\emph{What properties of $\ell_p$ norms make them  amenable to various applications? Can we reduce the problem of designing good algorithms for general norms to $\ell_p$ norms?}
\end{quote}
A common approach taken when working with $\ell_p$ norms  is to instead work with the function $\|x\|_p^p = \sum_i x_i^p$.  This function has several nice properties,  e.g.,  it  is separable and convex.    We want to understand its  fundamental properties that suffice for many applications,  hoping that this would allow us to define similar nice functions beyond $\ell_p$ norms.  

We identify \monotonicity,  characterized by   monotone gradients,  as a particularly valuable property  of $ \|x\|_p^p$.   This may sound intriguing because \monotonicity is typically associated with combinatorial set functions and not a priori norms.   
This is perhaps because all norms, except for scalings of $\ell_1$, are \emph{not}  \monotone.  
We therefore propose that a norm $\|\cdot\|$ is $p$-\monotone if $\|\cdot\|^p$  exhibits \monotonicity.

We show that for a large class of problems involving norms or equivalent objects,  $p$-\monotonicity suffices to design good algorithms. This is either by 
reframing existing algorithms for problems like Online Load-Balancing \cite{KMS-SODA23} and Bandits with Knapsacks \cite{ISSS-JACM22,KS-COLT20}  through a \monotone lens or by introducing novel analyses for problems such as Online Covering~\cite{AzarBCCCG0KNNP16}, Online Packing \cite{BN-MOR09}, and Stochastic Probing \cite{GNS-SODA17,PRS-APPROX23}.  

Moreover, we demonstrate that $p$-\monotone approximations of norms are  possible for large classes of norms, especially for all symmetric norms. Our approach paves the path for a unified approach to algorithm design involving norms and for 
obtaining guarantees that only depend polylogarithmically on the number of dimensions $n$. In particular, it can bypass the limitations of ubiquitous approaches like the use of ``concentration + union bound'' or Multiplicative Weights Update, that typically cannot give bounds depending only on the ambient dimension (they usually depend on the number of linear inequalities/constraints that define the norm/set); we expand on this a bit later. 

%, instead of polynomial in $n$ bounds that we typically get via  approaches such as John's ellipsoid, Multiplicative Weights, and a Union Bound.

%,  thereby  paving the way for a unified approach to algorithm design across a spectrum of optimization problems with norm objectives.

%\mnote{Maybe say something here like: In particular, our approach possible path for guarantees that depend on an intrinsic dimension of the problem, without the dependence on the number of inequalities used to define the norm/downward-closed set/XOS function/etc. which is present in the typical approaches for these problems (e.g. using MWU, using union bound).}

\subsection{$p$-\monotonicity and a Quick Application} \label{sec:intoSuper}

    Throughout the paper, we only deal with non-negative vectors, i.e.,  $x \in \R^n_+$, and monotone norms, namely those where $\|x\| \ge \|y\|$ if $x \ge y$. 
    
    We now reach the central definition of the paper, \underline{$p$-\monotonicity}: a monotone norm $\|\cdot\|$ is $p$-\monotone if its $p$-th power $\|\cdot\|^p$ has increasing marginal gains (a.k.a. supermodularity).
        
        \begin{definition}[$p$-\monotonicity]
		A monotone norm $\|\cdot\|$ is $p$-\monotone for  $p \geq 1$ if   for all $u,v,w \in \R^n_+$, 
        \[\|u+v +w\|^p - \|u+v\|^p \ge \|u + w \|^p - \|u\|^p .
        \]
	\end{definition}
As an example, $\ell_p$ norms are $p$-\monotone (follows from convexity of $x^p$). It may not be immediately clear, but the larger the $p$, the weaker this condition is and easier to satisfy (but the guarantees of the algorithm also become weaker as $p$ grows). In \Cref{sec:pSub} we present an in-depth discussion of $p$-\monotonicity, including this and other properties, equivalent characterizations, how to create new $p$-\monotone norms from old ones, etc.

%We will later show that large classes of norms, e.g. all symmetric norms, can be approximated by $O(\log n)$-\monotone norms.

    But to give a quick illustration of why $p$-\monotonicity is useful, we consider the classic \emph{Online Load-Balancing} problem \cite{AzarNR95,AspnesAFPW-JACM97}. In this problem,  there are  $T$ jobs arriving one-by-one that are to be scheduled on $n$ machines.  On arrival,  job $t \in [T]$ reveals how much size $p_{ti} \in \R_+$ it takes if executed on machine $i \in [n]$.  Given an $n$-dimensional norm $\|\cdot\|$,  the goal is to find an online assignment to minimize the norm of the load vector,  i.e., $\|\Lambda_T\|$ where the $i$-th coordinate of $\Lambda_T$ is the sum of  sizes of the jobs assigned to the $i$-th machine.  The following simple argument shows why $p$-\monotonicity implies a good algorithm for Online Load-Balancing.

	\begin{thm} \label{thm:loadBalancing}
    For Online Load-Balancing problem with a $p$-\monotone norm objective,  there is an  $O(p)$-competitive algorithm. 
	\end{thm}
	
	\begin{proof}
		The algorithm is simple: be greedy with respect to $\|\cdot\|$,  i.e.,  allocate job $t$ to a machine such that the increase in the norm of load vector is the smallest, breaking ties arbitrarily.
		
		For the analysis, 
		%notice that $p$-\monotonicity implies that $\|\cdot\|^p$ is supermodular, namely $\|x+y+z\|^p - \|x+z\|^p \ge \|x+y\|^p - \|x\|^p$. Also, 
		let $v_t \in \R^n_+$ be the load vector that the algorithm incurs at time $t$ and $\Lambda_t := v_1 + \ldots + v_t$, and let $v^*_t$ and $\Lambda^*_t$ be defined analogously for the hindsight optimal solution. Then the cost of the algorithm to the power of $p$ is 
		\begin{align*}
		\|\Lambda_T\|^p = \sum_t \bigg( \|\Lambda_t\|^p - \|\Lambda_{t-1}\|^p\bigg) 
			& \le \sum_t \bigg( \|\Lambda_{t-1} + v^*_t\|^p - \|\Lambda_{t-1}\|^p\bigg)\\
			& \le \sum_t \bigg( \|\Lambda_T + \Lambda^*_{t-1} + v^*_t \|^p - \|\Lambda_T + \Lambda^*_{t-1} \|^p\bigg) \\
			& = \|\Lambda_T + \Lambda^*_T\|^p - \|\Lambda_T\|^p,
		\end{align*} 
		where the first inequality follows from the greedyness of the algorithm and the second inequality from $p$-\monotonicity.  Rearranging and taking $p$-th root gives
		\begin{align*}
			2^{1/p} \|\Lambda_T\| \le \|\Lambda_T + \Lambda^*_T\| \le \|\Lambda_T\| + \|\Lambda_T^*\|.
		\end{align*}
		Thus,  $\|\Lambda_T\| \le \frac{1}{2^{1/p}-1}\|\Lambda_T^*\| = O(p) \cdot \|\Lambda_T^*\|$ as desired. 
	\end{proof}

 Since $\ell_p$ norms are $p$-\monotone, we obtain $O(p)$-competitive algorithms for Online Load-Balancing with these norms, implying the results of \cite{AzarNR95,AspnesAFPW-JACM97}.
%

%--------------------------------------------------------
\subsection{$p$-\monotone Approximation and our Technique via Orlicz Norms} 

%Before discussing further applications of $p$-\monotonicity,   we describe  how one can approximate many norms by $p$-\monotone  norms. This is  useful because  many norms  (e.g.,   $\ell_\infty$)   are not $p$-\monotone.  Indeed,  the  greedy algorithm for online load balancing is known to be $\Omega(n)$-competitive for $\ell_\infty$ \cite{AspnesAFPW-JACM97}.  However,  in such cases one would like to  \emph{approximate} the original norm by a    $p$-\monotone norm before running the algorithm; e.g., approximate $\ell_\infty$ by $\ell_{\log n}$ before running the greedy online load balancing algorithm. We show that such an approximation exists for  large classes of norms,  like for all symmetric norms.

One difficulty is that many norms  (e.g.,   $\ell_\infty$) are not $p$-\monotone for a reasonable $p$ (e.g., polylogarithmic in the number of dimensions $n$).  Indeed,  the  greedy algorithm for online load balancing is known to be $\Omega(n)$-competitive for $\ell_\infty$ \cite{AspnesAFPW-JACM97}.  However,  in such cases one would like to  \emph{approximate} the original norm by a    $p$-\monotone norm before running the algorithm; e.g., approximate $\ell_\infty$ by $\ell_{\log n}$. %before running the greedy online load balancing algorithm. 

One of our main contributions is showing that such an approximation exists for  large classes of norms. Formally, we say that a norm $\vertm{\cdot}$ $\alpha$-approximates another norm $\|\cdot \|$ if 
\[ \|x\| \leq \vertm{x} \leq \alpha\cdot \|x\| \qquad \text{for all $x \in \R_+^n$} .
\]
As our first main result (in \Cref{sec:approximation}), we show that all symmetric norms can be approximated by an  $O(\log n)$-\monotone norm.

\begin{restatable}{thm}{thmSymmetric} \label{thm:sym}
    %\begin{thm} 
        Every monotone symmetric norm $\|\cdot\|$ in $n$ dimensions can be $O(\log n)$-approximated by an $O(\log n)$-\monotone norm.  
    \end{restatable}

{Moreover, this approximation can be done efficiently given \ballopt oracle\footnote{We use the definition in \cite{CS-STOC19}, whereby \ballopt oracle allows us to compute  $\max_{v: \|v\|\leq 1} \langle x, v \rangle$ for any vector $x \in \R^n$ with a single oracle call.} access to the norm $\|\cdot\|$.}
This result plays a crucial role not only in allowing us to rederive many existing results for symmetric norms in a unified way, but also to obtain new results where previously general symmetric norms could not be handled. 

    We now give a high-level idea of the different steps in the proof of \Cref{thm:sym}.
    
%This result is useful since it allows us to rederive many existing results for symmetric norms in a unified way.  It also plays a crucial role in our new applications with symmetric norms,  such as online coverage and online packing,   where previously we could not handle symmetric norms.   

%\paragraph{Existing Techniques via \topk norms. }
\paragraph{Reduction to \topk norms.}
The reason why general norms are often difficult to work with is that they cannot be easily described. 
%That is,  although by duality any norm can be written as a max of linear functions,   general norms often require exponentially many linear functions.   This prevents the usage of techniques such as the softmax function or concentration bounds as they lose logarithmically in the number of linear functions,  which is a polynomial in $n$ (whereas only we want  $poly\log n$).  The hope therefore is to write general norms as a composition/max of simpler objects  that are more complicated than a single linear function.   
An approach that has been widely successful when dealing with symmetric norms is to instead work with \topk norms--- sum of the largest $k$ coordinates of a non-negative vector.   Besides giving a natural way to interpolate between $\ell_1$ and $\ell_\infty$,  they actually form a ``basis'' for all symmetric norms. In particular, it is known that any symmetric norm can be $O(\log n)$-approximated by the max of  polynomially many (weighted) \topk norms (see \Cref{lem:strucSymm}). Leveraging this property, we reduce our problem in that of finding $p$-\monotone approximations of \topk norms. 

%\red{[M: Will remove this par, ok?]
%Thus,  often solving a problem for symmetric norms reduces to the problem of solving it for \topk norms,  with some poly-log $n$ factor losses.  Although $\topk$ norms have a lot of structure,  existing methods remain tailored to specific problems, leaving a broader applicability to general norms less understood.  We show how to obtain $p$-\monotone approximation for \topk norms,  which allows us to unify many existing approaches and obtain new applications.}

\paragraph{Our Approach via Orlicz Norms. }    
Even though \topk norms have a very simple structure, it is still not clear how to design $p$-\monotone approximations for them. Not only thinking about $p$-th power of functions in high dimensional setting is not easy, but there is no constant or ``wiggle room'' in the definition of $p$-\monotonicity to absorb errors.
%This is  because thinking about  \monotonicity of the  $p^{th}$ power function,  especially in high $n$-dimensional setting,  is difficult.  
Our main idea to overcome this is to instead work with \emph{Orlicz norms} (defined in \Cref{sec:Orlicz}). These norms are fundamental objects in functional analysis (e.g., see book \cite{harjulehto2019generalized}) and have also found use in statistics and computer science; see for example~\cite{orliczRegression1,orliczRegression2} for their application in regression. Orlicz functions are much easier to work with because they are defined via a $1$-dimensional function $\R_+ \rightarrow \R_+$.  

    So our next step is showing that any \topk norm can be $O(1)$-approximated by an Orlicz norm.
 %   (given by the Orlicz function $G(t) := \max\{0, t - \frac{1}{k}\}$).  
This effectively reduce our task of designing a $p$-\monotone approximation from an $n$-dimensional situation to a $1$-dimensional situation.

\paragraph{Approximating  Orlicz Norms.} The last step is showing that every Orlicz norm can be approximated by a $p$-\monotone one. 

    \begin{restatable}{thm}{thmorlicz} \label{thm:orlicz}
        Every Orlicz norm $\|\cdot\|_G$ in $n$-dimensions can be  $O(1)$-approximated pointwise by a (twice differentiable) $O(\log n)$-\monotone norm.  
    \end{restatable}
    
As an example, an immediate corollary of this result along with \Cref{thm:loadBalancing} is an $O(\log n)$-competitive algorithm for Online Load-Balancing with an Orlicz norm objective.

%Given the sufficient condition for $p$-\monotonicity via the growth rate of the Orlicz function,  we approximate any Orlicz norm by  a $p$-\monotone norm.

    Our key handle for approaching \Cref{thm:orlicz} is the proof of a sufficient guarantee for an Orlicz norm to be $p$-\monotone: the 1-dimensional function $G$ generating it should grow ``at most like a polynomial of power $p$'' (\Cref{lemma:orliczpMono}).
    %we provide a growth rate condition on the Orlicz function that is sufficient for the corresponding Orlicz norm to be $p$-\monotone.  The proof of this lemma requires us to provide a lower bound on the Hessian of an Orlicz norm.  
    Then the construction of the approximation in the theorem proceeds in three steps. First, we simplify the structure of the Orlicz function $G$ by approximating it with a sum of (increasing) ``hinge'' functions $\tilde{G}(t) := \sum_i \tilde{g}_i(t)$. These hinge function, by definition, have a sharp ``kink'', hence do not satisfy the requisite growth condition. Thus, the next step is to approximate them by smoother functions $f_i(t)$ that grow at most like power $p$.   The standard smooth approximations of hinge functions (e.g.,  Hubber loss) do not give the desired approximation properties, so we design an approximation that depends on the relation between the slope and the location of the kink of the hinge function.   Finally, we  show  that the Orlicz norm $\|\cdot\|_F$, generated by the the function $F(t)= \sum_i f_i(t)$, both approximates $\|\cdot\|_G$ and is $O(\log n)$-\monotone.   

    \medskip
    Putting these ideas together,  gives the desired approximation of every symmetric norm by an $O(\log n)$-\monotone one.  

%--------------------------------------------------------
\subsection{Direct Applications of $p$-\monotonicity}

    Next, we detail a variety of applications for $p$-\monotone functions. Our discussion includes both reinterpretations of existing algorithms through the lens of \monotonicity and the introduction of novel  techniques that leverage \monotonicity to address previously intractable problems.
    In this section, we discuss applications that immediately follow from prior works due to $p$-\monotonicity.

    \subsubsection{Online Covering with a Norm Objective} \label{sec:introCoverSimple}
    
    The \cover problem is defined as follows: a norm $f : \R^n \rightarrow \R$ is given upfront, and at each round $r$ a new constraint $\ip{A^r}{x} \ge 1$ arrives (for some non-negative vector $A^r \in\R^n$). The algorithm needs to maintain a non-negative solution $x \in \R^n_+$ that satisfies the constraints $\ip{A_1}{y} \ge 1$, \ldots, $\ip{A_r}{y} \ge 1$ seen thus far, and is only allowed to increase the values of the variables $x$ over the rounds. The goal is to minimize the cost $f(x)$ of the final solution $x$.

    When the cost function $f$ is linear (i.e., the $\ell_1$ norm), this corresponds to the classical problem of Online Covering LPs \cite{AlonAA-STOC03,BuchbinderNaor-Book09}, where $O(\log s)$-competitive algorithms are known ($s$ is the maximum row sparsity)~\cite{BN-MOR09,GN-MOR14}. 
    This was first extended to $O(p \log s)$-competitive algorithms when $f$ is the $\ell_p$ norm \cite{AzarBCCCG0KNNP16}, and was later extended to sums of $\ell_p$ norms \cite{NS-ICALP17}. \cite{NS-ICALP17} posed as an open question whether good online coverage algorithms exist outside of $\ell_p$-based norms. 
    The following result, which  follows directly by applying the algorithm of~\cite{AzarBCCCG0KNNP16} to the $p$-\monotone approximations of Orlicz and symmetric norms provided by \Cref{thm:orlicz} and \Cref{thm:sym}, shows that this is indeed the case. 
    
    \begin{cor} \label{cor:cover}
        In the \cover problem, if the objective  can be $\alpha$-approximated by a $p$-\monotone norm then there exists an $O(\alpha p \log s)$-competitive algorithm, where $s$ is the maximum row sparsity. Hence, if the objective  is an Orlicz norm then this yields  $O(\log n \log s)$ competitive ratio, and if the objective is a symmetric norm then this yields $O(\log^2 n \log s)$ competitive ratio.  
    \end{cor}

    Since  $\ell_p$-norms are $p$-\monotone, this  extends  the result of \cite{AzarBCCCG0KNNP16}.

    \subsubsection{Applications via  Gradient Stability: Bandits with Knapsacks or Vector Costs }  
     
    Recently, \cite{KMS-SODA23} introduced the notion of gradient stability of norms and showed that  it implies good  algorithms for  online problems such as  Online Load-Balancing, Bandits with Vector Costs, and Bandits with Knapsacks.     (Gradient stability, however, does not suffice for other applications in this paper, like for Online Covering, Online Packing, Stochastic Probing, and robust algorithms.)
    In \Cref{sec:gradStability}, we show     that gradient stability is (strictly) weaker than  $p$-\monotonicity, and hence we can recover all of the results in \cite{KMS-SODA23}. Due to \Cref{thm:orlicz} for Orlicz norms, this   also 
    improves the approximation factors in all these applications from $O(\log^2 n)$  to $O(\log n)$ for Orlicz norms. See \Cref{sec:gradStability} for more details.

\subsubsection{Robust Algorithms}

    \monotonicity also has implications for online problem in stochastic, and even better, \emph{robust} input models. Concretely, consider the Online Load-Balancing problem from \Cref{sec:intoSuper}, but in the \textsc{Mixed} model where the time steps are partitioned (unbeknownst to the algorithm) into an \emph{adversarial} part and a \emph{stochastic} part, where in the latter jobs are generated i.i.d. from an unknown distribution. Such models that interpolate between the pessimism and optimism of the pure worst-case and stochastic models, respectively, have received significant attention in both online algorithms~\cite{meyerson,adSim,welfareSim,KKN15,molinaro,mixedModelSpikes,kesselheimMolinaro,guptaSahilRobust,ArgueGMS22} and online learning (see \cite{guptaCOLT19} and references within).

    Consider the (Generalized)\footnote{This is the generalization where there are $k$ ``options'' for processing each job, and each option incurs possible different loads on multiple machines.}  Online Load-Balancing in this model, with processing times normalized to be in $[0,1]$. For the $\ell_p$-norm objective,~\cite{Molinaro21} designed an algorithm with cost most $O(1)\cdot\OPT_{Stoch} + O(\min\{p, \log n\})\cdot \OPT_{Adv} + O(\min\{p, \log m\}\,n^{1/p})$, where $\OPT_{Adv}$ and $\OPT_{Stoch}$ are the hindsight optimal solutions for the items on each part of the input. That is, the algorithm has strong performance on the ``easy'' part of the instance, while being robust to ``unpredictable'' jobs. We extend this result beyond $\ell_p$-norm objectives,  by applying Theorem 1 of~\cite{Molinaro21} and our $p$-\monotone approximation for Orlicz norms from \Cref{thm:orlicz}.

    \begin{cor} \label{thm:robust}
        Consider the (Generalized) Online Load-Balancing problem in the \textsc{Mixed} model with processing times in $[0,1]$. If the objective function can be $\alpha$-approximated by a $p$-\monotone norm $\|\cdot\|$, then there is an algorithm with cost at most $O(\alpha)\,\OPT_{Stoch} + O(\alpha p^2) \, \OPT_{Adv} + O(\alpha p \|\ones\|).$ For Orlicz norm objective, this becomes $O(1)\,\OPT_{Stoch} + O(\log^2 n) \, \OPT_{Adv} + O(\log n \cdot \|\ones\|).$
    \end{cor}

%    Generalizes the $O(p)..$ for $\ell_p$-norms from~\cite{xx}.  
    
%--------------------------------------------------------
 \subsection{New Applications using $p$-\monotonicity} \label{sec:newApplications}

We discuss applications that require additional work but crucially rely on $p$-\monotonicity.
    
 \subsubsection{Online Covering with Composition of Norms}   \label{sec:introCovering}

 To illustrate the general applicability of our ideas, in particular going beyond symmetric norms, let us reconsider the \cover problem but now with ``composition of norms'' objective.  
 This version of the problem is surprisingly general: {its offline version captures the fractional setting of other fundamental problems such as Generalized Load-Balancing \cite{DLR-SODA23} and Facility Location.} %\cite{meyerson} %\cite{ST-TALG10}
% \blue{(see Section \ref{app:appComp} for more detail)}.
 
 Formally, in \cover with composition of norms, 
 the objective function is defined by a monotone outer norm $\|\cdot\|$ in $\R^k$, monotone inner norms $f_1,\ldots,f_k$ in $\R^n$, and subsets of coordinates $S_1,\ldots,S_\ell \subseteq [n]$ to allow the inner norms to only depend on a subset of the coordinates, i.e.,
 \[
 \|f_1(y|_{S_1}), \ldots, f_k(y|_{S_k})\|, 
 \]
 where $y|_{S_\ell} \in \R^{S_\ell}$ is the sub-vector of $y$ with the coordinates indexed by $S_\ell$. 
 The objective function is given upfront, but the constraints $\ip{A_1}{y} \ge 1, \ip{A_2}{y} \ge 1, \ldots, \ip{A_m}{y} \ge 1$ arrive in rounds, one-by-one, where $A_r \in [0,1]^n$ is the $r$th row of $A$. For each round $r$, the algorithm needs to maintain a non-negative solution $y \in \R^n_+$ that satisfies the constraints $\ip{A_1}{y} \ge 1$, \ldots, $\ip{A_r}{y} \ge 1$ seen thus far, and is only allowed to increase the values of the variables $y$ over the rounds. The goal is to minimize the composed norm objective.

 %Here, the objective function is given upfront, but the constraints $\ip{A_1}{y} \ge 1, \ip{A_2}{y} \ge 1, \ldots, \ip{A_m}{y} \ge 1$ arrive in rounds, one-by-one, where $A_r$ is the $r$th row of $A$. For each round $r$, the algorithm needs to maintain a non-negative solution $x \in \R^n_+$ that satisfies the constraints $\ip{A_1}{y} \ge 1$, \ldots, $\ip{A_r}{y} \ge 1$ seen thus far, and is only allowed to increase the values of the variables $x$ over the rounds. The goal is to minimize the cost of the final solution $x$, namely $\|f_1(x|_{S_1}), \ldots, f_k(x|_{S_k})\|$.
 
 Our next theorem shows that good algorithms exist for \cover even with composition of $p$-\monotone norms objectives. (Since this composed norm may not be $p$-\monotone,  \Cref{cor:cover} does not apply.)

     \begin{restatable}{thm}{thmCover} \label{thm:cover}
        If the outer norm $\|\cdot\|$ is $p'$-\monotone and the inner norms $f_\ell$'s are $p$-\monotone, then there is an $O(p' \, p \log^2 d\rho {\gamma})$-competitive algorithm for \cover, where $d$ is the maximum between the sparsity of the constraints and the size of the coordinate restrictions, namely $d = \max\{\max_r \supp(A_r)\,,\, \max_\ell |S_\ell|\}$, {$\rho = \max_{r,i : (A_r)_i \neq 0} \frac{1}{(A_r)_i}$, and $\gamma = \max_\ell  \frac{\max_{i \in S_{\ell}} f_\ell(e_i)}{\min_{i \in S_{\ell}} f_\ell(e_i)}$.} %$\rho = \max_{r,\ell,i \in S_\ell} \{\frac{f_\ell(e_i)}{(A_r)_i} : (A_r)_i \neq 0\}$.
    \end{restatable}

     Unlike  \Cref{cor:cover} which immediately followed from $p$-\monotonicity, this result needs new ideas to analyze the algorithm. We combine ideas from Fenchel duality used in~\cite{AzarBCCCG0KNNP16} with breaking up the evolution of the algorithm into phases where the gradients the norm behaves almost $p$-\monotone, inspired by~\cite{NS-ICALP17} in the $\ell_p$-case.  
%{\color{red}How does this theorem compare to prior works.}

    \subsubsection{Online Packing} 
    The \pack problem has the form: 
    \begin{align}
        \max\,& \ip{c}{x} \qquad \textrm{s.t.}\quad  Ax \le b \text{ and } x \ge 0,  \label{eq:packClassic}
    \end{align} 
    where $c \in \R^\numberofsteps, A \in \R^{\textrm{\#\,constraints}\times \numberofsteps}$, and $b \in \R^{m}$ have all non-negative entries. At the $t$-th step,  we see the value $c_t$ of the item  and its vector size $(a_{1,t},\ldots,a_{m,t})$), and have to immediately set $x_t$ (which cannot be changed later).  The classic online primal-dual algorithms were designed to address such problems  \cite{BN-MOR09,BuchbinderNaor-Book09}, and we know $O(\log (\rho \cdot \textrm{\#\,constraints}))$-competitive algorithms, where $\rho = \max_i \frac{\max_t a_{i,t} / c_t}{\min_{t : a_{i,t} > 0} a_{i,t} / c_i}$ is the ``width'' of the instance. %\mnote{Buchbinder and Naor assume $c = 1$, for the general case $\rho$ depends on $c$; } 

    For many packing problems, however, the \textrm{\#\,constraints} is exponential in number of items $\numberofsteps$, e.g.,  matroids are given by $\{ \sum_{t \in S} x_t \le r(S), ~\forall S \subseteq [\numberofsteps]\}$ where $r$ is the rank function. In such situations, a competitive ratio that depends logarithmically on the number of constraints is not interesting, and we are interested in obtaining competitive ratios that only depend on the intrinsic dimension of the problem. 
%However, we can also represent the constraints in the problem as $x \in P$, where $P$ is an $n$-dimensional downward closed set, i.e., if $0 \le x \le y$ and $y \in P$, then $x \in P$. We are interested in obtaining competitive ratios that depend on the intrinsic dimension of the problem, in dependent of the number of constraints. 

More formally, we consider the general \pack problem of the form: 
    \begin{align}
        \max\,& \ip{c}{x} \qquad \textrm{s.t.}\quad  Ax \in P \text{ and } x \ge 0,  \label{eq:packPre}
    \end{align} 
    where $P$ is an $\numberofdimensions$-dimensional downward closed set. Again, $\numberofsteps$ items come one-by-one  (along with $c_t$ and $(a_{1,t},\ldots,a_{m,t})$) and we need to immediately set $x_t$. Can we obtain $\poly\!\log(\numberofdimensions,\numberofsteps,\rho)$-competitive online algorithms?     
    In the stochastic setting of this problem, where items come in a random order (secretary model) or from known distributions (prophet model), Rubinstein \cite{Rubinstein-STOC16} obtained $O(\log^2 \numberofsteps)$-competitive algorithms (see also \cite{AD-SODA15}).     
     But in the adversarial online model, despite being a very basic problem, we do not know of good online algorithms beyond very simple $P$.
    
    We propose the use of $p$-\monotonicity as a way of tackling this problem. The connection with norms is because there is a 1-1 equivalence between downward closed sets $P$ and  monotone norms, given by the gauge function $\|x\|_P := \inf\{\alpha > 0 : \frac{x}{\alpha} \in P\}$, where $x \in P \Leftrightarrow \|x\|_P \le 1$. Thus,  the packing constraint $Ax \in P$ in \eqref{eq:packPre} is equivalent to $\|Ax\|_P \le 1$.
    Our next result illustrates the potential of this approach. 

    \begin{restatable}{thm}{thmPacking} \label{thm:packing}
    Consider an instance of the problem \pack where the norm associated with the feasible set $P$ admits an $\alpha$-approximation by a differentiable $p$-\monotone norm. 
    \begin{itemize}
    \item    If a $\beta$-approximation $\OPT \le \widetilde{\OPT} \le \beta \OPT$ of $\OPT$ is known, then there is an algorithm whose expected value is $O(\alpha) \cdot \max\{p, \log \alpha \beta \}$-competitive.  
    \item If no approximation of $\OPT$ is known, then there is an algorithm whose expected value is $O(\alpha) \cdot \max\{p, \log \numberofdimensions \rho \}$-competitive, where $\rho$ is an upper bound on the width $\frac{\max_{i, t} (a_{i, t} \cdot \alpha\, \|e_i\|_P / c_t)}{\min_{i, t: a_{i,t} > 0} (a_{i, t} \cdot \|e_i\|_P / c_t)}$.
    \end{itemize}
    \end{restatable}

    When $P = \{x \in \R^{\numberofdimensions} : 0 \le x \le b\}$ in \eqref{eq:packPre}, the norm $\|\cdot\|_P$ is just $\ell_\infty$ with rescaled coordinates. 
    Hence,  \Cref{thm:packing} together with  $O(\log \numberofdimensions)$-\monotone approximation of $\ell_\infty$ gives an $O(\log (\numberofdimensions \rho))$-competitive algorithm for the setting of \eqref{eq:packClassic}, which  essentially is the same classical guarantee of \cite{BN-MOR09}, albeit with a slightly different notion of width $\rho$. 
    As a side comment, this result/technique highlights a fact that we were unaware of: even for the classical problem \eqref{eq:packClassic},  if an estimate of $\OPT$ within $\poly(\numberofdimensions)$ factors is available, then one can avoid the dependence on any width parameter $\rho$.

    \subsubsection{Adaptivity Gaps and Decoupling Inequalities} \label{sec:introProbe}

    We show that $p$-\monotonicity is related to another fundamental concept, namely the power of adaptivity when making decisions under stochastic uncertainty. 
    To illustrate that, we consider the  problem of Stochastic Probing (\stoch), which was introduced as a generalization of stochastic matching \cite{CIKMR-ICALP09,BGLMNR-Algorithmica12} and has been greatly studied in the last decade \cite{GN-ICALP13,GNS-SODA16,GNS-SODA17,BSZ-RANDOM19,PRS-APPROX23}.    
    
    In this problem, there are $n$ items with unknown values $X_1,\ldots,X_n \ge 0$ that were drawn independently from known distributions. Items need to be \emph{probed} for their values to be revealed. There is a downward-closed family $\mathcal{F} \subseteq [n]$ indicating the feasible sets of probes (e.g., if the items correspond to edges in a graph, $\mathcal{F}$ can say that at most $k$ edges incident on a node can be queried). Finally, there is a monotone function $f : \R^n_+ \rightarrow \R_+$, and the goal is to probe a set $S \in \mathcal{F}$ of elements so as to maximize $\E f(X_S)$, where $X_S$ has coordinate $i$ equal to $X_i$ if $i \in S$ and 0 otherwise (continuing the graph example, $f(x)$ can be the maximum matching with edge values given by $x$).

    The optimal probing strategy is generally \emph{adaptive}, i.e., it probes elements one at a time and may change its decisions based on the observed values. Since adaptive strategies are complicated (can be an exponential-sized decision tree, and probes cannot be performed in parallel), one often resorts to \emph{non-adaptive} strategies that select the probe set $S$ upfront only based on the value distributions. The fundamental question
    %, studied in several papers~\red{\cite{AsadpourN16,GNS-SODA16,GNS-SODA17,BSZ-RANDOM19,EsfandiariKM21,PRS-APPROX23}}, 
    is how much do we lose by making decisions non-adaptively, i.e., if $\adapt(X, \mathcal{F}, f)$ denotes the value of the optimal adaptive strategy and $\NA(X, \mathcal{F}, f)$ denotes the value of the optimal non-adaptive one, then what is the maximum possible \emph{adaptivity gap} $\frac{\adapt(X,\mathcal{F},f)}{\NA(X,\mathcal{F},f)}$ for a class of instances.
    
    For submodular set functions, the adaptivity gap is known to be 2~\cite{GNS-SODA17,BSZ-RANDOM19}. For XOS set functions of width $w$,~\cite{GNS-SODA17} showed the adaptivity gap is at most $O(\log w)$, where a width-$w$ XOS set function $f : 2^{[n]} \rightarrow R_+$ is a max over $w$ linear set functions. The authors conjectured that the adaptivity gap for all XOS set functions should be poly-logarithmic in $n$, independent of their width. Since a monotone norm is nothing but a max over linear functions (given by the dual-norm unit ball), they form an extension of XOS set functions from the hypercube to all non-negative real vectors. Thus, the generalized  conjecture of~\cite{GNS-SODA17} is the following:

    \begin{conjecture} \label{conj:stochProbe}
    The adaptivity gap for stochastic probing with monotone norms is $\poly\!\log n$.
    \end{conjecture}

    We prove this conjecture for $p$-\monotone norms. %\mnote{Improve lead in, there is no punch!}

\begin{restatable}{thm}{thmStochProbing} \label{thm:stochProbing}
   For every $p$-\monotone objective function $f$, \stoch has adaptivity gap at most $O(p)$.
\end{restatable}
        
    This simultaneously recovers the $O(\log w)$ adaptivity gap result of~\cite{GNS-SODA17} (via \Cref{lem:maxOfpmonotone}) and the result of~\cite{PRS-APPROX23} for all  monotone symmetric norms (within $\poly\!\log(n)$). %Moreover, if our \Cref{conj:genNorm} about \monotonicity of general monotone norms is true, this would settle the full \Cref{conj:stochProbe}. Importantly, neither the techniques from~\cite{GNS-SODA17} nor~\cite{PRS-APPROX23} seem able to prove \Cref{conj:stochProbe}: the former uses a ``concentration + union bound'' over the linear functions composing $f$ (leading to the expected $O(\log w)$ loss), and the latter showed an $\Omega(\sqrt{n})$ lower bound for non-symmetric functions with their approach.  

    The proof of \Cref{thm:stochProbing} is similar to the Load-Balancing application of \Cref{sec:intoSuper}: we replace one-by-one the actions of the optimal adaptive strategy \adapt by those of the  ``hallucination-based'' non-adaptive strategy that runs \adapt on ``hallucinated samples'' $\bar{X}_i$'s (but receives value according to the true item values $X_i$'s). However, additional probabilistic arguments are  required; in particular, we need to prove a result of the type ``$\E \|V_1+\ldots+V_n\|^p \lesssim \E \|\bar{V}_1+\ldots+ \bar{V}_n\|^p$ implies $\E \|V_1+\ldots+V_n\| \lesssim p \cdot \E \|\bar{V}_1+\ldots+\bar{V}_n\|$'', where $V_i$'s and $\bar{V}_i$'s will correspond to $\adapt$ and the hallucinating strategy, respectively. We do this via an interpolation idea inspired by Burkholder~\cite{burkholder}.

    In fact, we prove a more general result than \Cref{thm:stochProbing} that shows the connections with probability and geometry of Banach spaces: a decoupling inequality for \emph{tangent sequences} of random variables (\Cref{thm:decoupling}); these have applications from concentration inequalities~\cite{deLaPena} to Online Learning~\cite{sridharanThesis,FosterRS17}. Two sequences of random variables  $V_1,V_2,\ldots,V_n$ and $\bar{V}_1,\bar{V}_2,\ldots,\bar{V}_n$ are called \emph{tangent} if conditioned up to time $t-1$, $V_t$ and $\bar{V}_t$ have the same distribution. We show that for such tangent sequences in $\R^d_+$ and a $p$-\monotone norm $\|\cdot\|$, we have $\E \|V_1 + \ldots + V_n\| \le O(p) \cdot \E \|\bar{V}_1 + \ldots + \bar{V}_n\|$, independent of the number of dimensions. This complements the (stronger) results known for the so-called UMD Banach spaces~\cite{veraarBook}.\footnote{We remark that $\R^n$ equipped with the $\ell_1$ norm is not a UMD space, while it is a $1$-\monotone norm, making our assumptions, and conclusions, distinct from this literature.}

\subsection{Future Directions}

In this work we demonstrate that $p$-\monotonicity is widely applicable to many problems involving norm objectives (from online to stochastic and from maximization to minimization problems). 
Our \Cref{thm:sym} shows that all symmetric norms have an $O(\log n)$-\monotone approximation.  
In an earlier version of this paper we conjectured that such an approximation should exist for all monotone norms but later we found a counter example.

\begin{lemma} \label{lem:counterEgGeneralNorm}
    There exist monotone norms such that if we $\alpha$-approximate it by any  $p$-supermodular subadditive function then $\alpha p = \Omega(\sqrt{n})$.
\end{lemma}

We defer the proof of this lemma to \Cref{sec:missingcounterEgGeneralNorm}. 
Given this counter example, an interesting future direction is to propose an alternate way for attacking the XOS functions adaptivity gap conjecture of \cite{GNS-SODA17}   and for designing online packing/covering algorithms that do not depend on the number of constraints but only  on the ambient dimension.

Another interesting future direction is to obtain integral solutions for the \cover problem. Similar to the work of \cite{NS-ICALP17}, our \Cref{cor:cover}  and \Cref{thm:cover} can only handle the fractional  \cover problem. 
Unlike the classic online set cover ($\ell_1$ objective), where randomized rounding suffices to obtain integral solutions, 
it is easy to show that we cannot round w.r.t. the natural fractional relaxation of the problem since there is a large integrality gap. Hence, a new idea will be required to capture integrality in the objective.

    $p$-\monotonicity is also related to the classic \emph{Online Linear Optimization}  (e.g., see book \cite{Hazan-Book16}). For the maximization version of the problem, in \Cref{app:regret} we show how to obtain total value at least $(1-\e) \OPT - \frac{p \cdot D}{\e}$ when a norm associated to the problem is $p$-\monotone, where $D$ is ``diameter'' parameter. In the case of prediction with experts, this recovers the standard $(1-\e) \OPT - O(\frac{\log d}{\e})$ bound ($d$ being the number of experts), and generalizes the result of~\cite{molinaro} when the player chooses actions on the $\ell_p$ ball. This gives an intriguing alternative to the standard methods like Online Mirror Descent and Follow the Perturbed Leader. It would be interesting to find further implications of this result, and more broadly $p$-\monotonicity, in the future.

\section{Supermodular Approximation of  Norms}
\label{sec:approximation}

In this section we discuss  $p$-\monotonicity and how many general norms  can be approximated by $p$-\monotone norms.  

    \subsection{$p$-\monotonicity and its Basic Properties} \label{sec:pSub}

    $p$-\monotonicity can be understood in a natural and more workable manner through the first and second derivatives of the norms; this is the approach we use in most of our results. While norms may not be differentiable, using standard smoothing techniques, every $p$-\monotone norm can be $(1+\e)$-approximated by another $p$-\monotone norm that is infinitely differentiable everywhere except at the origin; see \Cref{lemma:diff}.
     
%    In order to simplify the exposition, we assume in most statements that the $p$-\monotone norm $\|\cdot\|$ at hand is differentiable (in fact of class $C^{\infty}$, infinitely differentiable) at all points in $\R^d_+ \setminus \{0\}$. This is justified by the fact we can always mollify the norm by defining $\vertm{x} := \E \|(R_1 x_1, R_2 x_2, \ldots, R_d x_d)\|$, where the $R_i$'s are independent random variables in $[1,1+\e]$ that have pdf's of class $C^{\infty}$. Clearly for every non-negative vector $x$, $\|x\| \le \vertm{x} \le (1+\e) \|x\|$, and standard arguments show that $\vertm{\cdot}$ is of class $C^{\infty}$~\cite[Section 3.4]{schneider}; and $\vertm{\cdot}$ is still $p$-\monotone, since for each scenario $x \mapsto \|(R_1 x_1, R_2 x_2, \ldots, R_d x_d)\|$ is $p$-\monotone and this is property preserved by taking averages.

    We have the following equivalent characterizations of $p$-\monotone norms via their gradients and Hessians.

%     \snote{Change sentence; make start better.} Since any norm in $\R^d$ is twice differentiable almost everywhere (due to Rademacher's theorem \cite{}), it will be convenient for us to assume that the norm is twice differentiable everywhere (excluding the origin). In this case, we get the following equivalent definitions.

	\begin{lemma}[Equivalent characterizations] \label{lemma:equivalentPmono}
		For a  differentiable norm $\|\cdot\|$,  the following are equivalent:
				
\begin{itemize}
    \item  \emph{($p$-\monotonicity):} $\|\cdot\|$ is $p$-\monotone.
    
    \item \emph{(Gradient property):} $\|\cdot\|^p$ has monotone gradients over the non-negative orthant, i.e.,  for all $u,v \in \R^n_+$ and $\forall i\in [n]$, 
    $$\nabla_i \Big(\|u+v\|^p \Big) \ge \nabla_i \Big(\|u\|^p\Big)  \qquad \Longleftrightarrow \qquad 
    \frac{\nabla_i \|u+v\|}{ \nabla_i \|u\|} \geq   \Big( \frac{\|u\|}{\|u+v\| } \Big)^{p-1} .$$ 
			
   \item  \emph{(Hessian property):} If $\|\cdot\|$ is twice differentiable, then these are equivalent to: For all $u \in \R^n_+$ and $ \forall i,j\in [n]$,
   $$\nabla^2_{i,j} \Big(\|u\|^p \Big) \ge 0   \qquad \Longleftrightarrow \qquad 
 \nabla^2_{i,j}\|u\| \ge - (p-1) \frac{1}{\|u\|} \nabla_i\|u\| \cdot \nabla_j\|u\| .$$ 

   %\item (Hessian of $\|\cdot\|^2$) $\nabla^2_{i,j}\|u\|^2 \ge - 2 (p-2) \nabla_i \|u\| \nabla_j \|u\|$ for all $u \in \R^n_+$  and $i\neq j$. (note that on the RHS the exponent is 1) \snote{Shall we remove this property?}
		\end{itemize}
	\end{lemma}
	
	\begin{proof}
            The first part of the Gradient property follows when we take $\|w\| \rightarrow 0$. For the second part, use $\nabla \|u\|^p = p\cdot\|u\|^{p-1} \cdot \nabla \|u\|$.
            
		The first part of the Hessian property follows from monotonicity of gradients. For the second part, use
		\[
			\frac{1}{p}\, \nabla^2_{i,j} \Big(\|u\|^p \Big) =  \|u\|^{p-2} \cdot  \Big( (p-1)  \cdot  \nabla_i \|u\| \cdot  \nabla_j \|u\| + \|u\| \cdot  \nabla^2_{i,j} \|u\| \Big).    \qedhere 
		\]
		%The second follows by comparing the hessian of $\|\cdot\|$ and $\|\cdot\|^2$ by using the same identity with $p=2$.
	\end{proof}

        %For a differentiable norm $\|\cdot\|$, this is equivalent to saying that  $\|\cdot\|^p$ has monotone gradients over the non-negative orthant, i.e.,  for all $u,v \in \R^n_+$, $$\nabla_i \|u+v\|^p \ge \nabla_i \|u\|^p  \qquad \forall i\in [d].$$ 

    % Since $\|u+v\| \geq \|u\|$,  the RHS can only decrease with increase in $p$, which implies the following.
    Two  implications of the  Hessian property  are the following: \Cref{obs:pmonot} directly follows due to non-negativity of gradients and \Cref{obs:normScaling} uses $\nabla^2 \big(\vertm{x}^p \big) = A^T\nabla^2 \big(\|y\|^p \big)A$ for $y=Ax$.

    %\tkc{Even without the norm being twice differentiable this is be true: We can use, for example, convexity of $x \mapsto x^{p'/p}$.} 

    \begin{obs} \label{obs:pmonot}
        A differentiable $p$-\monotone  norm $\|\cdot\|$ is  also $p'$-\monotone for  $p' \ge p$. 
    \end{obs}

     %It will also be useful to note that positive linear transformations of $p$-\monotone norms are also $p$-\monotone. 
    
%    \begin{obs} \label{lemma:linearpMono}
%    If $\|\cdot\|$ is $p$-\monotone and $A$   is a \snote{square} \mnote{I think it holds for non-square as well} matrix with non-negative entries, then the function $x \mapsto \|Ax\|$ is also a $p$-\monotone norm. 
%    \end{obs}

    \begin{obs} \label{obs:normScaling}
    If $\|\cdot\| \colon \R^n \to \R$ is $p$-\monotone and $A \in \R^{n \times m}_{\geq 0}$ then the norm $\vertm{\cdot}$ in $\R^m$ given by $\vertm{x} := \|Ax\|$ is $p$-\monotone.
    \end{obs}

%	\red{M: Not necessarily differentiable at every point. Two solution: 1) Define without using gradient; 2) Every norm is differentiable almost everywhere, so suffices to add ``for every $x$ where is differentiable'' to the above def. Both defs should be equivalent. Thomas: Do we want to use $n$ or $d$ for the dimension? It seems we are a little inconsistent so far.} 

    As mentioned in the introduction, for every $p \ge 1$ the $\ell_p$ norm is $p$-\monotone. This follows, e.g., from the gradient property of $p$-\monotone norms.     
    For $p \ge \log n$, the $\ell_p$ norm is $O(1)$-approximated by $\ell_{\log n}$. So in particular, $\ell_{\infty}$ can be $O(1)$-approximated by $(\log n)$-\monotone norm. 
    We first generalize this fact ($\ell_{\infty}$ is  max of $n$ inequalities that are each $1$-\monotone).
    
    %and show that if the norm is defined by max of $w$  $p$-\monotone (e.g., $\ell_{\infty}$ is defined by max of $n$ inequalities that are each $1$-\monotone), then the norm can be approximated by a $\max\{p, \log w\}$--\monotone norm. 

    \begin{lemma} \label{lem:maxOfpmonotone}
        If $f_1, f_2, \ldots, f_w$ are differentiable  $p$-\monotone norms, then the norm $x \mapsto \max_i f_i(x)$ can be $2$-approximated by a $\max\{p, \log w\}$-\monotone norm. 
    \end{lemma} 

    \begin{proof}
         Let $p' = \max\{p, \log w\}$ and consider $\vertm{x} := (\sum_i f_i(x)^{p'})^{1 / p'}$. As $\max_i f_i(x)^{p'} \leq \sum_i f_i(x)^{p'} \leq w \cdot \max_i f_i(x)^{p'}$, we have
         \[
         \max_i f_i(x) = (\max_i f_i(x)^{p'})^{1 / p'} \leq \vertm{x} \leq (w \cdot \max_i f_i(x)^{p'})^{1 / p'} = w^{1 / p'} \max_i f_i(x) \leq 2 \max_i f_i(x).
         \]
         Furthermore, for all $u, v \in \R^n_+$, we have
         \[
         \nabla \vertm{u + v}^{p'} = \sum_i f_i(u + v)^{p'} \geq \sum_i \nabla f_i(u)^{p'} = \nabla \vertm{u}^{p'},
         \]
         since each $f_i$ is $p'$-\monotone.
    \end{proof}

    An implication of this is that any norm in $n$ dimensions can be $O(1)$-approximated by an $n$-\monotone norm. This is because  we can find a $\frac{1}{4}$-net $\mathcal{N} \subseteq \cA$ of the unit ball of the dual norm  of size $2^{O(n)}$. Since,  $\vertm{x} := \max_{a \in \mathcal{N}} \ip{a}{|x|}$ is an $O(1)$ approximation of $\|x\|$ and $\ip{a}{|x|}$  is a re-weighted $\ell_1$ norm, \Cref{lem:maxOfpmonotone} implies that $\vertm{x}$  is  $n$-\monotone norm. 

    \begin{cor}
        Any monotone norm in $n$-dimensions can be $O(1)$-approximated by an $n$-\monotone norm.  
    \end{cor}

    Although $p$-\monotone norms have several nice properties, they also exhibit some strange properties. For instance, sum of two $p$-\monotone norms can be very far from being $p$-\monotone. 
    
    \begin{lemma} \label{lemma:sumNotP}
        The norm $\|x\| = \|x\|_1 + \|x\|_2$ is not $p$-\monotone for any $p = o(\sqrt{n})$.
    \end{lemma}
    \begin{proof}
        Consider some $i \neq j \in [n]$. 
        By Hessian property in \Cref{lemma:equivalentPmono}, for $\|x\|_1 + \|x\|_2$ to be $p$-\monotone, we must have
        \[
         -\frac{\nabla_i \|x\|_2 \cdot \nabla_j \|x\|_2}{ \|x\|_2} = \nabla^2_{i,j}\|x\| \ge - (p-1) \frac{\nabla_i\|x\| \cdot \nabla_j\|x\|}{\|x\|}  = - (p-1) \frac{\big(1 + \nabla_i \|x\|_2 \big) \cdot \big(1 + \nabla_j \|x\|_2 \big)}{\|x\|_1 + \|x\|_2} .
        \]
        Since $\nabla_i \|x\|_2 = \frac{x_i}{ \|x\|_2}$, we can simplify to get
        \[
        \frac{x_i \cdot x_j}{\|x\|_2^3} \leq (p-1) \cdot \frac{\big( \|x\|_2 + x_i\big) \cdot \big( \|x\|_2 + x_j\big)}{\big(\|x\|_1 + \|x\|_2 \big)\cdot \|x\|_2^2} .
        \]
        Now consider the vector $x = (\sqrt{n},\sqrt{n}, 1, 1, \ldots, 1)$, i.e., a vector having the first two coordinates $\sqrt{n}$ and every other coordinate $1$. Note that $\|x\|_1 = \Theta(n)$ and $\|x\|_2 = \Theta(\sqrt{n})$.
        For $i=1$ and $j=2$, the last inequality gives
        \[
            \frac{n}{\Theta(n^{3/2})} \leq (p-1) \cdot \frac{\Theta(\sqrt{n}) \cdot \Theta(\sqrt{n})}{\Theta(n)\cdot \Theta(\sqrt{n})^2} = \frac{p-1}{\Theta(n)},
        \]
        which is only possible for $p = \Omega(\sqrt{n})$.
    \end{proof}

\subsection{Orlicz Norms and a Sufficient Condition for $p$-\monotonicity} \label{sec:Orlicz}

    %\red{M: Define a ``nice'' convex function as one with the desired limits, maybe there is a standard name.}

    %\red{Is this nice property really required? Our Orlicz approx for \topk does not satisfy them.}

    The following class of Orlicz functions and Orlicz norms will play a crucial role in all our norm approximations.
    
    \begin{definition}[Orlicz Function]
    A {continuous} function $G: \R_+ \rightarrow\R_+$ is called an \emph{Orlicz function} if  it is convex, increasing, and satisfies $G(0)=0$ and $\lim_{t \rightarrow\infty}G(t) = \infty$.
    \end{definition}

    \begin{definition}[Orlicz Norm]
        Given an Orlicz function $G$, the associated \emph{Orlicz norm} is defined by $$\|x\|_G := \inf\bigg\{\alpha > 0: \sum_i G\bigg(\frac{|x_i|}{\alpha}\bigg) \le 1\bigg\}.$$
        Since we only focus on non-negative vectors, we will  ignore throughout the absolute value $|x_i|$.
    \end{definition}

    For example, any $\ell_p$ is an Orlicz norm when we select $G(t) = t^p$. 
     Orlicz norms are fundamental in functional analysis \cite{KosmolBook11}, but have also found versatile applications in TCS. For instance, in regression the choice between $\ell_1$ and $\ell_2$ norms depends on outliers and stability, so an Orlicz norm based on the popular Huber convex loss function is better suited \cite{orliczRegression1,orliczRegression2}. Later we will show that Orlicz norms can be used to approximate any symmetric norm.

    The following lemma is our main tool for working with Orlicz norms. It states that for such a norm to be $p$-\monotone, it suffices that its generating function $G$ grows ``at most like power $p$''. The key is that this reduces the analysis of the $n$-dimensional norms $\|\cdot\|_G$ to the analysis of 1-dimensional functions, which is significantly easier.  
    
    %The following lemma formalizes the idea that an Orlicz norm is $p$-\monotone as long as the function $G$ grows ``at most like power $p$''. Notice that in particular, the function $G(t) = t^p$ satisfies this condition, at equality (although in this special case, $\|\cdot\|_G = \ell_p$ is $p$-\monotone instead of $(2p-1)$-\monotone). 
    
    \begin{lemma} \label{lemma:orliczpMono}
	Consider a twice differentiable   convex function $G : \R_+ \rightarrow \R_+$. If $G$ satisfies $$G''(t) \cdot t \le (p-1) \cdot G'(t) ~~~\forall t \ge 0,$$ then the Orlicz norm $\|x\|_G$ is $(2p-1)$-\monotone.
    \end{lemma}

    Notice that the function $G(t) = t^p$ satisfies this condition, at equality. While in this special case the norm $\|\cdot\|_G = \ell_p$ is $p$-\monotone, in general we obtain the slightly weaker conclusion of $(2p-1)$-\monotonicity.

     The rest of the subsection proves this lemma. The proof will rely on  the Hessian property of $p$-\monotone norms. First, we observe the following formula for the gradient of the Orlicz norm $\|\cdot\|_G$; this can be found on page 24 of \cite{KosmolBook11}, but we repeat the proof for completeness.

	\begin{claim}  \label{lemma:gradientOrlicz}
		If $G$ is differentiable, then the gradient of the Orlicz norm $\|\cdot\|_G$ is given by 
		\begin{align*}
			\nabla_i \|x\|_G = \frac{G'(\tfrac{x_i}{\|x\|_G})}{\sum_\ell \tfrac{x_\ell}{\|x\|_G} \cdot G'(\tfrac{x_\ell}{\|x\|_G})} .
		\end{align*}
  \end{claim}
    \begin{proof}
        Consider the function  $H(x,c) := \sum_\ell G(\tfrac{x_\ell}{c})$. Since $H(x,\|x\|_G)=1$ is constant, we get
        \[
        0 ~=~ \frac{\partial}{\partial x_i} H(x,\|x\|_G) ~=~ \frac{1}{\|x\|_G} G'(\tfrac{x_i}{\|x\|_G}) - \sum_\ell \Big(G'(\tfrac{x_\ell}{\|x\|_G}) \cdot \frac{x_\ell}{\|x\|_G^2} \Big) \cdot \nabla_i \|x\|_G . \qedhere
        \]
    \end{proof}
  
  To simplify notation, we define the following.
  
  \begin{definition} Let 
  \begin{gather*}
			\tilde{x}_\ell := \frac{x_\ell}{\|x\|} \quad \text{and} \quad \gamma(x) := \sum_\ell \tfrac{x_\ell}{\|x\|} \cdot G'(\tfrac{x_\ell}{\|x\|}). \qquad \text{Hence,} \quad 
   \nabla_i \|x\|_G =  \frac{G'(\tilde{x}_i)}{\gamma(x)} .
 \end{gather*}
 \end{definition}
	
    Differentiating the expression for the gradient $\nabla_i \|x\|_G$ gives a close-form formula for the Hessian of the Orlicz norm. (To be careful with the chain rules, we use brackets; for example $\nabla_j (g(h(x)))$ to denote the gradient of the composed function $g \circ h$, not of just $g$.)
	
	% \begin{lemma} \label{lemma:hessianOrlicz}
	% 	If $G$ is twice differentiable, then for every $i,j$ the Hessian of the Orlicz norm $\|\cdot\|_{G}$ (we use $\|\cdot\|$ instead of $\|\cdot\|_G$ to simplify the expression)
	% 	{\small 
	% 		\begin{align*}
	% 		\nabla^2_{ij} \|x\| = \frac{\nabla_i \|x\| \,\nabla_j \|x\|}{\gamma(x) \, \|x\|} \, \bigg[\sum_\ell \tilde{x}_\ell^2 \cdot G''(\tilde{x}_\ell) \bigg] - \frac{1}{\gamma(x) \, \|x\|} \bigg[ \nabla_i\|x\| \cdot \tilde{x}_j \cdot G''(\tilde{x}_j) + \nabla_j\|x\| \cdot \tilde{x}_i \cdot G''(\tilde{x}_i) \bigg] + \ones(i=j)\, \frac{G''(\tilde{x}_i)}{\gamma(x) \,\|x\|},
	% 	\end{align*}
	% 	}
	% \end{lemma}
	
  \begin{claim} \label{claim:HessianGrad}
  If $G$ is twice differentiable, then the Hessian of the norm 
      {\small
	\begin{align}
		\small \nabla^2_{ij}\|x\| %&= \frac{1}{\gamma(x)} \cdot \nabla_j (G'(\tilde{x}_i)) ~+~ \frac{\nabla_j \|x\| \cdot \nabla_j \|x\|}{\|x\|} ~-~ \frac{\nabla_i \|x\|}{\gamma(x)} \cdot \sum_\ell \bigg(\tilde{x}_\ell \cdot \nabla_j (G'(\tilde{x}_\ell))\bigg) ~-~ \frac{\nabla_i \|x\| \cdot G'(\tilde{x}_j)}{\gamma(x) \cdot \|x\|} \notag\\
		&= \frac{1}{\gamma(x)} \cdot \nabla_j (G'(\tilde{x}_i)) ~-~ \frac{\nabla_i \|x\|}{\gamma(x)} \cdot \sum_\ell \bigg(\tilde{x}_\ell \cdot \nabla_j (G'(\tilde{x}_\ell))\bigg). \label{eq:hess22}
		\end{align}	
	}
  \end{claim}

Before proving the claim (which is mostly  algebra), we complete the proof of the  lemma.
  \begin{proof}[Proof of \Cref{lemma:orliczpMono}]
    When $\ell \neq j$ we have $\nabla_j \tilde{x}_\ell = \nabla_j (\frac{x_\ell}{\|x\|_G}) = - \frac{x_\ell \cdot \nabla_j \|x\|_G}{\|x\|_G^2} = - \tilde{x}_\ell \cdot \frac{\nabla_j \|x\|_G}{\|x\|_G}$, and when $\ell = j$ we get an extra $+\frac{1}{\|x\|_G}$ from the product rule. Letting $\ones(\ell = j)$ denote the indicator that $\ell=j$, this implies
        \begin{align} \label{eq:gradxtilde}
            \nabla_j \tilde{x}_\ell \,=\, -   \frac{x_\ell \cdot \nabla_j \|x\|}{\|x\|^2} \,+\, \ones(\ell = j)\cdot \frac{1}{\|x\|} .
        \end{align}
	Applying this to \eqref{eq:hess22} and using  $\nabla_j (G'(\tilde{x}_\ell)) = G''(\tilde{x}_\ell) \cdot \nabla_j \tilde{x}_\ell$, we get
	\begin{align}
		\small \nabla^2_{ij}\|x\| &=  - \frac{G''(\tilde{x}_i) \cdot x_i \cdot \nabla_j \|x\|}{\gamma(x) \cdot \|x\|^2} ~+~  \ones(i=j)\cdot \frac{G''(\tilde{x}_i)}{\gamma(x) \cdot \|x\|}  \notag\\
		& \qquad - \frac{\nabla_i \|x\|}{\gamma(x)} \cdot\bigg[ - \sum_\ell \bigg(\tilde{x}_\ell \cdot G''(\tilde{x}_\ell) \cdot \frac{x_\ell \cdot \nabla_j \|x\|}{\|x\|^2}\bigg) ~+~ \frac{\tilde{x}_j \cdot G''(\tilde{x}_j)}{\|x\|} \bigg] \notag \\
		&\ge - \frac{1}{\|x\|} \bigg[ \nabla_i\|x\| \cdot \frac{\tilde{x}_j \cdot G''(\tilde{x}_j)}{\gamma(x)} + \nabla_j\|x\| \cdot \frac{\tilde{x}_i \cdot G''(\tilde{x}_i)}{\gamma(x)} \bigg] \label{eq:pMonoOHave},
	\end{align}	
where the inequality uses that the missing terms   are non-negative for $x \ge 0$.

        %With this expression for the Hessian at hand, we can prove \Cref{lemma:orliczpMono}. 

	Moreover, the assumption on $G$ implies that
  \[ \frac{\tilde{x}_j \cdot G''(\tilde{x}_j)}{\gamma(x)} \,\le\, (p-1) \frac{G'(\tilde{x}_j)}{\gamma(x)} \,=\, (p-1) \nabla_j \|x\|.\] 
  Similarly, we get  for $i$ that 
  $  \frac{\tilde{x}_i \cdot G''(\tilde{x}_i)}{\gamma(x)}  \le  (p-1) \nabla_i \|x\|$. 
   Plugging these bounds into \eqref{eq:pMonoOHave} gives 
   \begin{align*} 
	\nabla^2_{ij} \|x\| \ge - (2p-2) \frac{1}{\|x\|} \nabla_i \|x\| \cdot \nabla_j \|x\|,
	\end{align*}
 which proves \Cref{lemma:orliczpMono} by \Cref{lemma:equivalentPmono}.
	\end{proof}    

Finally, we prove the missing claim.
 \begin{proof}[Proof of \Cref{claim:HessianGrad}]
  Differentiating w.r.t. $x_j$ the gradient $\nabla_i \|x\|_G = \frac{G'(\tilde{x}_i)}{\gamma(x)}$ from Lemma \ref{lemma:gradientOrlicz} gives 
		\begin{align}
			\nabla^2_{ij}\|x\|_G &= \frac{1}{\gamma(x)} \cdot \nabla_j (G'(\tilde{x}_i)) \,-\, {G'(\tilde{x}_i) \cdot \frac{1}{\gamma(x)^2}} \cdot \nabla_j \gamma(x) \notag\\
			&= \frac{1}{\gamma(x)} \cdot \nabla_j (G'(\tilde{x}_i)) \,-\, {\frac{\nabla_i \|x\|_G}{\gamma(x)}} \cdot \nabla_j \gamma(x) .  \label{eq:hess1}
		\end{align}
		We expand the gradient $\nabla_j \gamma(x)$ of the second term: 
        \begin{align*}
		  \nabla_j \gamma(x) ~=~ \sum_\ell \nabla_j \bigg(\tilde{x}_\ell G'(\tilde{x}_\ell)\bigg) ~
		  = \sum_\ell \bigg({\nabla_j \tilde{x}_\ell} \cdot G'(\tilde{x}_\ell) + {\tilde{x}_\ell \cdot \nabla_j (G'(\tilde{x}_\ell))} \bigg) .
        \end{align*}
        By \eqref{eq:gradxtilde}, we have 
        \begin{align*}
            \sum_\ell {\nabla_j \tilde{x}_\ell} \cdot G'(\tilde{x}_\ell) &= 
          - \sum_\ell  \tilde{x}_\ell \cdot \frac{\nabla_j \|x\|_G}{\|x\|_G} \cdot G'(\tilde{x}_\ell)   +  \frac{1}{\|x\|_G}  \cdot G'(\tilde{x}_j) \\
          &= - \frac{\nabla_j \|x\|_G}{\|x\|_G} \cdot \gamma(x)+ \frac{G'(\tilde{x}_j)}{\|x\|_G}  \quad =\quad  0.
        \end{align*}
        This implies
         \[
		  \nabla_j \gamma(x) ~=~  \sum_\ell {\tilde{x}_\ell \cdot \nabla_j (G'(\tilde{x}_\ell))} ,
        \]
        which proves the claim by substitution in \eqref{eq:hess1}.
    \end{proof}

    \subsection{Approximation of Orlicz Norms} 

        This section shows that every Orlicz norm can be approximated by an $O(\log n)$-\monotone norm.

    \thmorlicz*
    
    Before giving an overview of the proof of the theorem, it will help the discussion to have the following lemma that shows that to approximate an Orlicz norm $\|\cdot\|_G$, it suffices to approximate the corresponding Orlicz function $G$.

    \begin{lemma}
    \label{lemma:approximationofG}
    Suppose $\tilde{G}$ is an Orlicz function satisfying for all $t$ with $G(t) \leq 1:$
    \begin{enumerate}
        \item $G(t) \leq \tilde{G}(t)$.
        \item    $\tilde{G}(t/\gamma) ~\leq~ \alpha \, G(t) + \tfrac{1}{n}$  for some universal constants $\alpha \geq 0$ and $\gamma \geq 1$.  
    \end{enumerate}
    %with $G(t) \leq \tilde{G}(t)$ and satisfying for some $\alpha \geq 1$, $\gamma \geq 1$, and $1 \geq 0$ that  \snote{Do we only need $1=1$?}
    %\[ \tilde{G}(t/\gamma) \cdot \ones_{G(t) \leq 1} ~\leq~ \alpha G(t) + \frac{1}{n} \qquad \text{for all $t \geq 0$.} \]
     Then,   $\|x\|_G  \le \|x\|_{\tilde{G}} \leq \gamma (\alpha + 1) \|x\|_G$.
    \end{lemma}

    \begin{proof}
        The first inequality $G(t) \le \tilde{G}(t)$  implies that $\|x\|_G \le \|x\|_{\tilde{G}}$. Moreover, by convexity and $\alpha \geq 0$, we have $\tilde{G}(\tfrac{t}{\gamma(\alpha+1)}) \le (1 - \frac{1}{\alpha+1}) \tilde{G}(0) + \frac{1}{\alpha+1} \tilde{G}(t/\gamma) = \frac{1}{\alpha+1} \tilde{G}(t/\gamma)$ since $\tilde{G}$ is an Orlicz function. So, 
        \[
        \sum_i \tilde{G}\bigg(\frac{x_i}{\gamma(\alpha+1)\|x\|_G} \bigg) \le \frac{1}{\alpha+1} \sum_i \tilde{G}\bigg(\frac{x_i}{\gamma \|x\|_G} \bigg) \le \frac{1}{\alpha+1} \sum_i \bigg[ \alpha \cdot G\bigg(\frac{x_i}{\|x\|_G} \bigg) + \frac{1}{n} \bigg]  = 1,
        \]
        where the last inequality uses $\gamma \geq 1$. By definition of Orlicz norm, this implies $\|x\|_{\tilde{G}} \le \gamma(\alpha+1) \|x\|_G$.
    \end{proof}

    Observe that we do %\Cref{thm:orlicz} does
    not care how the Orlicz function $\tilde{G}$ behaves after $G(t)>1$, since these values do not matter for Orlicz norm $\|\cdot\|_G$.

    \paragraph{Proof Overview of \Cref{thm:orlicz}.} Given the sufficient condition for $p$-\monotonicity via the growth rate of the Orlicz function from \Cref{lemma:orliczpMono} and \Cref{lemma:approximationofG} above, the proof of \Cref{thm:orlicz} involves three steps. First, we simplify the structure of the Orlicz function $G$ by approximating it with a sum of (increasing) ``hinge'' functions $\tilde{G}(t) := \sum_i \tilde{g}_i(t)$ in the interval where $G(t) \leq 1$. These hinge function by definition have a sharp ``kink'', hence do not satisfy the requisite growth condition. Thus, the next step is to approximate them by smoother functions $f_i(t)$ that grow at most like power $p$. However, the standard smooth approximations of hinge functions (e.g. Hubber loss) do not give the desired properties, so we use a subtler approximation that depends on the relation between the slope and the location of the kink of the hinge function (this is because the approximation condition required by \Cref{lemma:approximationofG} is mostly multiplicative, while standard approximations focus on additive error). Finally, we  show  that the Orlicz norm $\|\cdot\|_F$, where $F(t)= \sum_i f_i(t)$, both approximates $\|\cdot\|_G$ and is $O(\log n)$-\monotone. % This uses properties of Orlicz norms  generated by Orlicz functions that are sums of Orlicz functions.\mnote{I think the last sentence is overselling :)}

    %We first shows that the Orlicz function $G(\cdot)$ can be approximated by the sum of piecewise linear functions $\tilde{G}(t) := \sum_i \tilde{g}_i(t)$ in the interval where $G(t) \leq 1$. \red{Since these piecewise linear functions $\tilde{g}_i(t)$ are not Orlicz functions (and hence do not give any Orlicz norms), the next step shows how to approximate them with Orlicz functions $f_i(t)$ while satisfying that the Orlicz norm $\|\cdot\|_{f_i(t)}$ is $O(\log n)$-\monotone.} \snote{Say 1-2 lines to sell and describe this more.}
    %Finally, we  show  that the Orlicz norm $\|\cdot\|_F$, where $F(t)= \sum_i f_i(t)$, both approximates $\|\cdot\|_G$ and is $O(\log n)$-\monotone. This uses properties of Orlicz norms  generated by Orlicz functions that are sums of Orlicz functions.

    \begin{proof}[Proof of \Cref{thm:orlicz}]
    This first claim gives the desired approximation of $G$ by piecewise linear functions with $n$ slopes. 
    %We first show that the function $G$ generating the norm can be approximated by a piecewise-linear function with $n$ slopes.

    \begin{claim} \label{claim:piecewise}
    There are $a_1, \ldots, a_n, b_1, \ldots, b_n \ge 0$ such that $\tilde{G} : \R_+ \rightarrow \R_+$ defined by $\tilde{G}(t) = \sum_{i=1}^n \max\{0, a_i t - b_i\}$ fulfills
    $$\|x\|_{G} \le \|x\|_{\tilde{G}} \le 2 \|x\|_G, ~~\forall x \in \R^n_+.$$
    \end{claim}

    \begin{proof}
        Since $G$ is an Orlicz function, it is continuous and satisfies $G(0) = 0$ with $\lim_{t \rightarrow \infty} G(t) = \infty$. Hence, there are points $t_0=0,t_1,t_2,\ldots,t_n \in \R_+$ such that $G(t_i) = \frac{i}{n}$. Choose $a_i$ and $b_i$ such that $a_i t_{i-1} - b_i = 0$ and $a_i t_i - b_i = \frac{1}{n} - \sum_{j<i} a_j (t_i - t_{i-1})$. By this definition $\tilde{G}(t_i) = \sum_{i=1}^n \max\{0, a_i t - b_i\} = G(t_i) = \frac{i}{n}$ for all $i = 0,1,\ldots,n$.

        We claim that $G(t) \le \tilde{G}(t) \le G(t) + \frac{1}{n}$ for all $t$ with $G(t) \in [0,1]$. The first inequality follows from the convexity of $G$, and the second inequality follows because for all $t \in [t_i, t_{i+1}]$ we have $\tilde{G}(t) \le \tilde{G}(t_{i+1})= \frac{i+1}{n} \le G(t) + \frac{1}{n}$. Hence, \Cref{lemma:approximationofG} concludes the proof of the claim.
    \end{proof}

     Next, we will approximate the piecewise linear functions $\max\{0, a_i t - b_i\}$ with Orlicz functions. This approximation will depend on whether $b_i \geq 1$ or not.

     \begin{definition}
    Let $H$ be the set of indices $i \in[n]$ such that $b_i \ge 1$ and $L = [n] \setminus H$ be the other indices. For $p \geq  2(\ln n) + 1$, define
    \begin{align*}
    F(t) := \sum_{i = 1}^n f_i(t)\,, ~~\textrm{ where } f_i(t) = \left\{\begin{array}{ll}
        2 \cdot (\frac{2a_i}{b_i + 1})^p \cdot t^p&    \textrm{, if $i \in H$}\\%\textrm{, if $b_i \ge 1$}\\
        (b_i^p + (a_i t)^p)^{1/p} - b_i& \textrm{, if $i \in L$} %\textrm{, if $b_i \in [0,1)$} 
    \end{array} \right. .
    \end{align*}
    \end{definition}

    The idea behind this construction is the following: first write $\tilde{g}_i(t) := \max\{0, a_i t - b_i\} = \max\{b_i, a_i t\} - b_i$ and notice that $\tilde{G}(t) = \sum_{i=1}^n \tilde{g}_i(t)$. When $b_i \ge 1$, then the points $t$ where $\tilde{g}_i(t)$ equals 0 and 1 (respectively, $\frac{b_i}{a_i}$ and $\frac{b_i + 1}{a_i}$) are within a factor of 2, namely $\tilde{g}_i$ fairly sharply jumps from 0 to 1; in this case, we replace it by the sharply increasing function $f_i(t) = (\frac{2 a_i}{b_i+1})^p \cdot t^p$. Otherwise, the function $\tilde{g}_i$ does not increase so sharply and we just replace the maximum in $\tilde{g}_i(t) = \max\{b_i, a_i t\} - b_i$ by the $\ell_p$ norm to obtain $f_i$. Then to obtain $F$, we take the sum of the functions $f_i$.

   % We first claim that the norm $\|\cdot\|_F$ is a constant approximation to $\|\cdot\|_{\tilde{G}}$. 
   We first prove that $f_i(t)$ approximates $\tilde{g_i}(t)$ in  a suitable way. We will also show that $f_i$ grows at most like power $p$. (In the following claim, the intuition behind the truncation $\min\{\cdot,2\}$ is that in definition of the Orlicz norm, the places where the generating function $G$ is bigger than 1 are not important; instead of 2, one can use any value strictly bigger than 1.) 

    \begin{claim} \label{claim:fi}
        Consider $p \ge 2(\ln n) + 1$. For all $i \in [n]$, we have 
        \begin{enumerate}
            \item $f_i(t) ~\ge~ \min\{\tilde{g}_i(t), 2\}$ ~~ for all $t \ge 0$. \snote{We only need $1$ in the min instead of 2? Also, do we write it as all t with $\tilde{g}_i(t) \leq 2$?}
            \item $f_i(\frac{t}{4}) ~\le~ 2\tilde{g}_i(t) + \frac{1}{n^2}$ ~~ for all $t$ with $\tilde{g}_i(t) \leq 1$.
            \item $f''_i(t) \cdot t ~\le~ (p-1) \cdot f'_i(t)$ ~~ for all $t \ge 0$. 
        \end{enumerate}
    \end{claim}

    \begin{proof}
        We prove these properties separately for the cases $b_i \ge 1$ and $b_i \in [0,1)$. 

        \paragraph{Case 1: $b_i \ge 1$, so  $f_i(t) = 2 \, (\frac{2 a_i}{b_i + 1})^p\cdot  t^p$.}  ~\\
        
        For Item 1, notice that for $t \in [0, \frac{b_i}{a_i}]$ we have $\min\{\tilde{g}(t), 2\} = 0$ and for $t > \frac{b_i}{a_i}$ we have $\min\{\tilde{g}(t), 2\} \le 2$, by definition. Since $f_i(t) \ge 0$ for $t \in [0, \frac{b_i}{a_i}]$, and for $t \ge \frac{b_i}{a_i}$
        \begin{align*}
            f_i(t) \ge 2 \bigg(\frac{2b_i}{b_i + 1} \bigg)^p  \ge 2,  
        \end{align*}
        where the last inequality uses $b_i \ge 1$. Thus, we have $f_i(t) \ge \min\{\tilde{g}_i(t), 2\}$ for all $t \ge 0$. 

        For Item 2, for all $t \in [0, \tilde{g}_i^{-1}(1)]$ (this interval is the same as $[0, \frac{b_i + 1}{a_i}]$) we have 
        \begin{align*}
            f_i(t/4) \,\le\, 2 \cdot \bigg(\frac{2 a_i}{b_i + 1}\bigg)^p \cdot \bigg(\frac{b_i + 1}{4 a_i}\bigg)^p = \frac{1}{2^{p-1}} \,\le\, 2\tilde{g}(t) + \frac{1}{n^2}.
        \end{align*}
        
        %For $t \ge \tilde{g}_i^{-1}(1)$, we have $2 \le 2 \tilde{g}_i(t) \le 2 \tilde{g}_i(t) + \frac{1}{n}$. Thus, for all $t \ge 0$ we have $\min\{f_i(t/4), 2\} \le 2\tilde{g}_i(t) + \frac{1}{n}$ as desired. 
        
        Item 3 holds with equality. Namely, by taking the second-derivative of $f_i(t)$, we get
        \[
        f_i''(t) \cdot t = p \cdot (p-1) \cdot 2 \cdot \bigg(\frac{2a_i}{b_i + 1}\bigg)^p \cdot t^{p-1} = (p-1) \cdot f_i'(t).
        \]%can be readily verified by taking  the derivative of $f_i(t)$. 

        \paragraph{Case 2: $b_i \in [0,1)$, so $f_i(t) = (b_i^p + (a_i t)^p)^{1/p} - b_i$.} ~\\
        
        For Item 1, observe that $f_i(t) = (b_i^p + (a_i t)^p)^{1/p} - b_i \ge \max\{b_i, a_i t\} - b_i = \tilde{g}(t)$. 
        
        For Item 2,  for all $t \in [0, \frac{2b_i}{a_i})$, we have
        \begin{align*}
        f_i(t/4) \,\le\, \Big((b_i)^p + (b_i/2)^p \Big)^{1/p} - b_i \,=\, b_i  \bigg(1 + \frac{1}{2^p}\bigg)^{1/p} - b_i \,\le\, b_i  \bigg(1 + \frac{1}{p 2^p}\bigg) - b_i \,\le\,  \tilde{g}_i(t) +  \frac{1}{n^2}, 
        \end{align*}
        where the  last inequality uses the fact that we are in a case $b_i \le 1$. On the other hand, when $t \ge \frac{2b_i}{a_i}$, then $b_i \le \frac{a_i t}{2}$ and so $\tilde{g}_i(t) = \max\{0, a_i t - b_i\} \ge \frac{a_i t}{2}$; at the same time,  
        \begin{align*}
            f_i(t/4) \le \Big( (a_i t / 2)^p + (a_i t / 4)^p \Big)^{1/p} = ((1/2)^p + (1/4)^p)^{1/p} \cdot a_i t \le a_i t.
        \end{align*}
        Putting these observations together, gives $f_i(t/4) \le 2 \tilde{g}_i(t)$, proving Item 2.  

        For Item 3, compute the derivatives to get 
        \begin{align*}
            f'_i(t) = \frac{a_i^p t^{p-1}}{(b_i^p + (a_i t)^p)^{1 - \frac{1}{p}}} ~~~\textrm{ and }~~~ f''_i(t) = \frac{(p-1) a_i^p t^{p-2}}{(b_i^p + (a_i t)^p)^{1 - \frac{1}{p}}} \,-\, (p-1) \frac{a_i^{2p} t^{2(p-1)}}{(b_i^p + (a_i t)^p)^{2 - \frac{1}{p}}}.
        \end{align*}
        The last term in  $f''_i(t)$ is non-positive, and so it follows that $f''_i(t) \cdot t \le (p-1) \cdot f'_i(t)$.
    \end{proof}

    Now we use the last claim to prove that $\|\cdot\|_F$ approximates $\|\cdot\|_{\tilde{G}}$.
    \begin{claim} \label{claim:orliczF}
        If $p \ge \log n + 1$, then for every $x \in \R^n_+$ we have $\|x\|_{\tilde{G}} \le \|x\|_F \le 12 \|x\|_{\tilde{G}}$.
    \end{claim}

    \begin{proof}
        First, from  \Cref{claim:fi} we get % that $F(t) \ge \min\{2, \tilde{G}(t)\}$. 
        \[ 
            F(t) = \sum_{i = 1}^n f_i(t) \stackrel{\Cref{claim:fi}}{\ge} \sum_{i = 1}^n \min\{2, \tilde{g}_i(t)\} \ge \min\Big\{2, \sum_{i = 1}^n \tilde{g}_i(t)\Big\} = \min\{2, \tilde{G}(t)\}. %\label{eq:FLB}
        \]
        Moreover, for any $t$ with $1 \geq \tilde{G}(t)  \geq \tilde{g}_i(t)$, we have from  \Cref{claim:fi} that
        %, we have that $F(t/4) \le 2 \tilde{G}(t) + \frac{1}{(n+1)^2}$ for every $t \ge 0$ such that $\tilde{G}(t) \le 1$: for such $t$ we have $\sum_{i=1}^{n+1} \tilde{g}_i(t) \le 1$ and so 
        %
        \[
            F(t/4) = \sum_{i=1}^n f_i(t/4) \stackrel{\Cref{claim:fi}}{\le} \sum_{i=1}^n \bigg( 2  \tilde{g}_i(t) + \frac{1}{n^2}\bigg) = 2\tilde{G}(t) + \frac{1}{n}. %\label{eq:FUB}
        \]

        Now, applying Lemma~\ref{lemma:approximationofG} for $\alpha=2$ and $\gamma=4$ implies $\|x\|_G \leq \|x\|_{\tilde{G}} \leq 4(2+1) \|x\|_G$.
    \end{proof}
    
    Finally, we show that the norm $\|\cdot\|_F$ is $(2p-1)$-\monotone. 

    \begin{claim} \label{claim:Fmono}
        The norm $\|\cdot\|_F$ is $(2p-1)$-\monotone. 
    \end{claim}

    \begin{proof}
        Due to \Cref{lemma:orliczpMono}, it suffices to show that $F''(t) \cdot t \le (p - 1) \cdot F'(t)$ for all $t \ge 0$. We have
        \[
        F''(t) \cdot t ~=~ \sum_{i=1}^n f''_i(t) t ~\leq~ \sum_{i=1}^n (p-1) f'_i(t) ~=~ (p-1) \cdot F'(t). \qedhere
        \]
    \end{proof}

    Claims \ref{claim:piecewise}, \ref{claim:orliczF}, and \ref{claim:Fmono} together give the desired approximation to the Orlicz norm $\|\cdot\|_G$, proving \Cref{thm:orlicz}.
    \end{proof}

%##########################################################################
%##########################################################################
%##########################################################################
%##########################################################################

%####################################################################
%####################################################################
    
    \subsection{Approximation of \topk and Symmetric Norms}

    In this section we will give $p$-\monotone norm approximations of  \topk and Symmetric Norms. 
    The strategy is to first construct such an approximation for \topk norms; general symmetric norms are then handled by writing them as a composition of \topk norms and applying the $p$-\monotone approximation to each term.
    
    \paragraph{Approximation of \topk norms.} Even though the \topk norms have a simple structure, it is not clear how to approximate them by a $p$-\monotone norm  directly. Instead, we resort to an intermediate step of expressing a \topk norm (approximately) as an  {Orlicz norm}.

    %We will prove that \topk norms can be approximated within a constant factor by Orlicz norms. 

     \begin{thm} \label{thm:topk}
        For every $k \in [n]$, the \topk norm $\|\cdot\|_{\topk}$ in $n$-dimensions can be $2$-approximated by an Orlicz norm.  
    \end{thm}

    Together with \Cref{thm:orlicz} from the previous section, this implies the following.
    
    \begin{cor} \label{cor:topk}
        For every $k \ge 1$, the \topk norm $\|\cdot\|_{\topk}$ in $n$-dimensions can be $2$-approximated by an $O(\log n)$-\monotone norm.  
    \end{cor}

 The construction in the proof of  \Cref{thm:topk} is inspired by the embedding of \topk norms into $\ell_\infty$ by Andoni et al.~\cite{ANNNorms}. They considered the ``Orlicz function'' $G(t)$ that is 0 until $t = \frac{1}{k}$ and behaves as the identity afterwards, i.e., $G(t) := t \cdot \ones(t\geq \frac1k)$. The rough intuition of why the associated ``Orlicz norm'' approximately captures the \topk norm of a vector $u$ is because $\frac{u}{\|u\|_{\topk}}$ has $\approx k$ coordinates with value above $\frac{1}{k}$ (the top $\approx k$ coordinates), which are picked up by $G$ and give $\sum_i G(\frac{u_i}{\|u\|_{\topk}}) \approx \sum_{i \textrm{ in top $k$}} \frac{u_i}{\|u\|_{\topk}} \approx 1$; thus, by definition of Orlicz norm, $\|u\|_G \approx \|u\|_{\topk}$. However, this function $G$ is not convex due to a jump at $t=1/k$, so it does not actually give a norm. Convexitfying  this function also does not work: the convexified version of $G$ is the identity, which yields the $\ell_1$ norm,  does not approximate $\topk$. Interestingly, a  modification of this convexification actually works. 
 
 \begin{proof}[Proof of \Cref{thm:topk}]
 We define the Orlicz function $G(t) := \max\{0, t - \frac{1}{k}\}$. %\red{M: Can also say that can skip the discretization lemma, and the final $p$-mono approximation is just the Orlicz given by $G(t) := \frac{1}{k} [ (1+ (kt)^p)^{1/p} - 1]$}
We show that the norm $\|\cdot\|_G$ generated by this function is a $2$-approximation to the \topk norm.

\emph{Upper bound $\|x\|_G \le \|x\|_{\topk}$.} By the definition of Orlicz norm, it suffices to show that $\sum_i G(\frac{x_i}{\|x\|_{\topk}}) \le 1$. For that, since there are at most $k$ coordinates having $x_i \geq \frac{\|x\|_{\topk}}{k}$, we get
%; for the other coordinates, $G(\frac{x_i}{\|x\|_{\topk}) = 0$. Thus, 
%
\[ \sum_i G\bigg(\frac{x_i}{\|x\|_{\topk}}\bigg) ~=~ \sum_{i:~ x_i \,\geq\, {\|x\|_{\topk}}/{k}} \bigg(\frac{x_i}{\|x\|_{\topk}} - \frac{1}{k}\bigg) 
~\le~ \frac{\|x\|_{\topk}}{\|x\|_{\topk}} - 1 ~<~ 1.
\]

\emph{Lower bound $\|x\|_G \ge \frac{\|x\|_{\topk}}{2}$.} By the definition of  Orlicz norm, it suffices to show that for any $\alpha < \frac{1}{2}$, we have $\sum_i G(\frac{x_i}{\alpha \|x\|_{\topk}}) > 1$. Let $T_k$ denote the set of the $k$ largest coordinates of $x$. Then,
    \begin{align*}
        \sum_i G\bigg(\frac{x_i}{\alpha \|x\|_{\topk}}  \bigg) \ge \sum_{i \in T_k} G\bigg(\frac{x_i}{\alpha \|x\|_{\topk}}  \bigg) \ge \sum_{i \in T_k} \bigg(\frac{x_i}{\alpha \|x\|_{\topk}} - \frac{1}{k}\bigg) = \frac{1}{\alpha} - 1,
    \end{align*}
    which is $> 1$ whenever $\alpha < \frac{1}{2}$. This concludes the proof of \Cref{thm:topk}.  
\end{proof}

%#########################################################################
%#########################################################################
%#########################################################################
%#########################################################################

%\subsection{Approximation of Symmetric Norms} 

%\subsection{$p$-\monotone Approximation of Norms}

Given \Cref{thm:topk}, one might wonder whether all symmetric norms can be approximated within a constant factor by Orlicz norms. The following lemma shows that this is impossible. 

\begin{lemma} \label{lem:SymmetricToOrlicz}
    There exist  symmetric norms  that cannot be $O(\log n)^{1-\epsilon}$-approximated by an Orlicz norm for any constant $\epsilon>0$.
\end{lemma}

We defer the proof of this observation to \Cref{sec:missingProofs}. 

%Interestingly, we show that every symmetric norm 

\paragraph{Approximation of  symmetric norms.} Although \Cref{lem:SymmetricToOrlicz} rules out the possibility of approximating any symmetric norm by an Orlicz norm within a constant factor,   we show that every symmetric norm can be $O(\log n)$-approximated by an an $O(\log n)$-\monotone norm.

\thmSymmetric*

    As mentioned before, the idea is to write a     
    general symmetric norm  as composition of \topk norms and applying the $p$-\monotone approximation to each term. 
    More precisely, the following lemma, proved in~\cite{KMS-SODA23} (and a similar property in \cite{ANNNorms,CS-STOC19}), shows that the any monotone symmetric norm can be approximated by \topk norms. 

    \begin{lemma}[{\cite[Lemma 2.5]{KMS-SODA23}}] \label{lem:strucSymm}
	For any monotone symmetric norm $\|\cdot\|$ in $\R^d$, there are $\log n$ non-negative scalars $c_1,c_2\ldots,c_{\log n}$ such that the norm 
	\begin{align}
%			\vertm{x} := \max_{j \le \log n} c_j \|x\|_{\mtop{2^j}}  \label{eq:strucSymm}
		\vertm{x} := \bigg\|\Big(c_1 \|x\|_{\topx{2^1}},~\ldots~, c_{\log n} \|x\|_{\topx{2^{\log n}}}\Big)  \bigg\|_{\infty} \label{eq:strucSymm}
	\end{align}
	satisfies $\|x\| \,\le\, \vertm{x} \,\le\, 2 \log n \cdot \|x\|.$
    \end{lemma}

    With the decomposition of monotone symmetric norms into \topk norms in \Cref{lem:strucSymm} and the  $p$-\monotone approximation to the latter in \Cref{cor:topk}, we can now prove that every symmetric norm can be $O(\log n)$-approximated by an $O(\log n)$-\monotone norm. 

\begin{proof}[Proof of \Cref{thm:sym}]
    Consider a monotone symmetric norm and its approximation $\vertm{x}$ given by \Cref{lem:strucSymm}. Let $f_k$ be the $p$-\monotone $2$-approximation of the \topk norm as given by \Cref{cor:topk}, where $p = \Theta(\log n)$. We replace in $\vertm{x}$ the \topk norms by these approximations, and the outer $\ell_\infty$-norm by the $\ell_p$-norm 
    to obtain the norm 
    \[ g(x) := \bigg(\sum_{i=1}^{\log n} c_i^p \cdot  \big( f_{2^i}(x) \big)^p \bigg)^{1/p}.\] 
    By the standard $\ell_p$ to $\ell_\infty$ comparison, we that $g(x)$ is a constant approximation to $\vertm{x}$ since $p = \Theta(\log n)$. Hence, $g(x)$ is  an $O(\log n)$-approximation to the original norm $\|x\|$. 
    
    Moreover, to see that $g$ is $p$-\monotone, consider the gradient of $g^p$, which is given by 
    \begin{align*}
        \nabla (g(x)^p) = \sum_{i=1}^{\log n} c_i^p \cdot \nabla \Big(f_{2^i}(x)^p\Big).
    \end{align*}
    Since each norm $f_j$ is $p$-\monotone and the multipliers $c_i$  are non-negative,   $\nabla (g(x)^p)$ is non-decreasing. This implies $p$-\monotonicity by the Gradient property in \Cref{lemma:equivalentPmono}.
\end{proof}

    {We remark that given a \ballopt oracle, we can evaluate at a given point the value and gradient of the approximating norm constructed in \Cref{thm:sym}, up to error $\e$, in time $\poly(\log \frac{1}{\e}, n)$. This is because the decomposition into $\topk$ norms from \Cref{lem:strucSymm} can be found in polytime given this oracle (e.g., see~\cite{KMS-SODA23,CS-STOC19}), the Orlicz function of the Orlicz norm approximation of each $\topk$ can be constructed explicitly, and the value and gradient of this Orlicz norm can be evaluated by binary search on the scaling $\alpha$ in the definition of the Orlicz norm (and \Cref{lemma:gradientOrlicz}).}

%%%%%%%%%%%%%%%%%%%%%%%%%%%%%%%%%%%%%%

\section{Applications to Coverage Problems} \label{sec:cover}
        
     The \cover problem (with $\ell_1$ objective) has been greatly influential in online algorithms,  leading to online primal-dual techniques (see book \cite{BuchbinderNaor-Book09}). In this section, we consider this problem where the objective function is a \emph{composition} of $p$-\monotone norms, which is not necessarily monotone; see \Cref{lemma:sumNotP}. This generality allows as to capture fractional versions of other classical problems, such as Online Vector Scheduling and Online Facility Location with norm-based costs (see \Cref{app:appComp}). 
       %Recall from \Cref{sec:introCoverSimple} that by using the algorithm of \cite{AzarBCCCG0KNNP16} for Online Covering with convex costs, we directly obtain guarantees when the cost function $f : \R^d_+ \rightarrow \R$ is a $p$-\monotone/Orlicz/ symmetric norm ($O(p \log d)$-, $O(\log^2 d)$-, and $O(\log^3 d)$-competitivity, respectively). 
    
    Let us recall the problem definition from \Cref{sec:introCovering}.  In its offline version, there is an $m \times n$ constraint matrix $A$ with entries in the interval $[0,1]$. The objective function is given by nested norms, defined by a monotone outer norm $\|\cdot\|$ in $\R^k$, monotone inner norms $f_1,\ldots,f_k$ in $\R^n$, and sets of coordinates $S_1,\ldots,S_\ell \subseteq [n]$ to allow the inner norms to only depend on a subset of the coordinates. The offline version of the problem is given by
	\begin{align*}
		\min ~&\bigg\|\Big(f_1(y|_{S_1})\,,\, \ldots \,,\, f_k(y|_{S_k})\Big)\bigg\| \qquad \text{s.t.} \qquad\quad  Ay \ge \ones  \quad  \text{and }	\quad y \in \R^n_+,
	\end{align*}
    where $x|_{S_\ell} \in \R^{S_\ell}$ is the sub-vector of $x$ with the coordinates indexed by $S_\ell$. We use $\OPT$ to denote the optimum of this problem. (We note constraints $Ax \ge b$ with more general right-hand side $b$ can be handled by rescaling the constraints.)

    In the online version of the problem, the objective function is given upfront, but the constraints $\ip{A_1}{y} \ge 1, \ip{A_2}{y} \ge 1, \ldots, \ip{A_m}{y} \ge 1$ arrive in rounds, one-by-one, where $A_r$ is the $r$th row of $A$. For each round $r$, the algorithm needs to maintain a non-negative solution $y \in \R^n_+$ that satisfies the constraints $\ip{A_1}{y} \ge 1$, \ldots, $\ip{A_r}{y} \ge 1$ seen thus far, and is only allowed to increase the values of the variables $y$ over the rounds. The goal is to minimize the cost of the final solution $y$, namely $\|f_1(y|_{S_1}), \ldots, f_k(y|_{S_k})\|$.
    Note that the objective function is a norm that in general is not $p$-\monotone, even if  $f_\ell$'s are $p$-\monotone. It is also in general not a symmetric norm. Hence, it cannot be handled by \Cref{cor:cover}. 
    
    The main result of this section is a competitive algorithm for \cover with this general objective function.  Its proof requires new ideas in the analysis of the algorithm, in particular generalizing and reconciling arguments introduced in~\cite{AzarBCCCG0KNNP16} and \cite{NS-ICALP17}.
    In the following result, $\supp(u)$ denotes the size of the support of the vector $u$ and the parameter $d$ in  is always at most $n$.

    \thmCover*

    For the remainder of this section, we prove this result, assuming without loss of generality the following condition that parallels the one used in~\cite{NS-ICALP17}.

    \begin{assumption} \label{ass:cover}
    The restricting sets $S_1,\ldots,S_k$ partition the set of variables $[n]$.
    \end{assumption}

Since the guarantee in \Cref{thm:cover} does not depend on  $n$, this can be achieved by introducing new coordinates, as done in \cite{NS-ICALP17}. This assumption is formally discharged in \Cref{app:cover}.

%    \begin{assumption} \label{ass:cover}
%    ~
%    \vspace{-4pt}
%    \begin{enumerate}
%        \item The restricting sets $S_1,\ldots,S_k$ partition the set of variables $[n]$.
%
%        \vspace{-4pt}
%        \item For every norm $f_\ell$ and coordinate $i$, we have $f_\ell(e_i) = 1$. 
%    \end{enumerate}
%    \end{assumption}

%Since the guarantee in \Cref{thm:cover} does not depend on the number of coordinates $n$, the first property can be achieved by introducing new coordinates, as done in \cite{NS-ICALP17}; once the $S_\ell$'s are made disjoint, the second property can be achieved by rescaling the coordinates. These assumptions are formally discharged in \Cref{app:cover}.

%############################################################################
%############################################################################

    \subsection{Algorithm}

  The algorithm we consider is the ``continuous online mirror-descent'' also used in \cite{AzarBCCCG0KNNP16} and \cite{NS-ICALP17}. To state it, let $F(y) = (f_1(y|_{S_1}),\ldots,f_k(y|_{S_k}))$, so the objective function can be more comfortably stated as $\|F(y)\|$. However, the algorithm (and analysis) is actually based on this function raised to the power $p'$ (recall the norm $\|\cdot\|$ is assumed to be $p'$-\monotone), so define $\Psi(y) := \frac{1}{p'} \|F(y)\|^{p'}$. Then the algorithm can be described as follows: (Set $\delta > 0$ small enough so that $\Psi^{\star}(\delta \ones) \le \Psi(x^*)$. This can be done online by seeing the minimum cost of satisfying the first (non-trivial) constraint $\ip{A_1}{y} \ge 1$, which gives a lower bound on $\OPT$, and use this to set $\delta$ small enough.)

\vspace{4pt}
\begin{mdframed}
    \vspace{-7pt}
    \begin{proc} \label{proc:advCont}
     \normalfont
     \textbf{Online Covering} 

    \vspace{4pt}
    \noindent Initialize $x(0) = 0$, $\tau = 0$.
    
    \vspace{4pt}
    \noindent For each round $r = 1,2\ldots, m$:

    \vspace{-8pt}
    \begin{enumerate}
        \item Receive the new constraint $\ip{A_r}{y} \ge 1$. While $\ip{A_r}{x(\tau)} < 1$, increase the continuous time $\tau$ at rate 1, and increase all coordinates of $x$ continuously using $$\dot{x}_i(\tau) = \frac{(A_r)_i \cdot (x_i(\tau) + \frac{1}{d})}{\nabla_i \Psi(x(\tau)) + \delta},$$  where $\dot{x}(\tau)$ means derivative with respect to the continuous time $\frac{\d x(\tau)}{\d \tau}$.
    \end{enumerate}
\end{proc}
\end{mdframed}

%###########################################################################
%###########################################################################

    \subsection{Analysis}

%    \mnote{Double check we are giving enough credit to Vish, specially proof of \Cref{lemma:dual}}
    
    We show that the above algorithm is $O(p'\,p \log^2 d\rho)$-competitive, proving \Cref{thm:cover} (under the current assumptions). We will track the cost of the algorithm with respect to the function $\Psi(y) = \frac{1}{p'} \|F(y)\|^{p'}$ instead of the original objective $\|F(y)\|$. Also, it will be convenient to put the constraints $A_r$ also in continuous time together with the solution $x(\tau)$ constructed by the algorithm. So let $A(\tau)$ be the constraint $A_r$ corresponding to time $\tau$, namely, let $\tau_r$ be the time $\tau$ at the start of round $r$ and let $A(\tau) = A_r$ iff $\tau \in [\tau_r, \tau_{r+1})$. Also let $\tau_{\final}$ be the last time of the process, and $x_{\final} := x(\tau_{\final})$ be the final solution output by the algorithm. We assume without loss of generality that $p' \ge 2$, since the $p'$-\monotonicity of $\|\cdot\|$ implies its $(p' + 1)$-\monotonicity (\Cref{obs:pmonot}) and replacing $p'$ for $p' + 1$ does not change the target $O(p'\,p \log^2 d \rho)$-competitiveness.

%		It will be more convenient to forget about the rounds $r=1,2,\ldots,m$ in which the constraints come, and instead consider the whole process of the online instance happening in continuous time. To formalize this, we let $A(\tau)$ and $x(\tau)$ be the constraint and solution constructed by the algorithm at continuous-time $\tau$, namely letting $\tau_r := \Delta \tau_1 + \ldots \Delta \tau_{r-1}$ denote the start time of the round $r$, we  
    
    We first show that the increase in $\Psi$-cost at a round is proportional to how long we kept raising the variables, essentially due to the fact that the change in $x(\tau)$ is ``moderated'' by the instantaneous cost $\nabla\Psi$ in the denominator. %(notice that $x(\tau_r)$ and $x(\tau_{r+1})$ are the solutions computed by the algorithm at the beginning and at the end of round $r$, respectively).
	
	\begin{lemma} \label{lemma:coverUB}
	For any round $r$, $$\Psi(x(\tau_{r+1})) - \Psi(x(\tau_r)) \le 2 (\tau_{r+1} - \tau_r).$$ In particular, the final solution satisfies $\Psi(x_{\final}) \le 2\, \tau_{\final}.$
	\end{lemma}
	
	\begin{proof}
        For any time $\tau \in [\tau_r, \tau_{r+1})$, $\Psi(x(\tau))$ increases at rate at most 2: by chain rule the derivative of $\Psi(x)$ with respect to the continuous time is given by 
         \begin{align*}
            \frac{\d \Psi(x(\tau))}{\d \tau} = \ip{\nabla \Psi(x(\tau))}{\dot{x}(\tau)} = \sum_i \frac{(A_r)_i \cdot (x_i(\tau) + \frac{1}{d})}{1 + \delta/(\nabla_i \Psi(x(\tau)))} \le \ip{A_r}{x(\tau)} + \frac{1}{d} \sum_i (A_r)_i \le 2,
        \end{align*}
        where the first inequality uses  non-negativity of $\nabla_i \Psi(x(\tau))$ and the
        last inequality uses the fact that during this round we always have $\ip{A_r}{x(\tau)} < 1$ and that $(A_r)_i \le 1$ by assumption. Integrating this change over the duration $\tau_{r+1} - \tau_r$ of the round gives the result.
	\end{proof}
	
%	From now on it is best to ignore the rounds and just consider the dynamics as a single continuous time process. So we use $a(\tau)$ to denote the $a(r)$ of the round where this new time $\tau$ lies in, and use $x(\tau)$ in the same way. So letting $\bar{\tau}$ be the last of these times (so $x(\bar{\tau})$ is the final solution of the algorithm), adding the previous lemma over all rounds gives.
	
%	\begin{cor} \label{cor:time}
%		$h(x(\bar{\tau})) \le 2 \bar{\tau}$. 
%	\end{cor} 

    We now switch to lower bounding $\Psi$-cost of $\OPT$, namely $\Psi(x^*)$. This is done via the  appropriate notion of duality, namely \emph{convex conjugacy}; we recall this definition and its main involutory property (e.g. Corollary E.1.3.6~\cite{HUL}).

    \begin{obs}[Convex conjugate] \label{obs:conjugate}
        Given a convex function $h : \R^w \rightarrow \R$, its \emph{convex conjugate} is defined as $h^{\star}(u) := \sup_v\{ \ip{v}{u} - h(v)\}$. We note that $h$ is also the convex conjugate of $h^{\star}$, namely $h(u) = \sup_v\{ \ip{v}{u} - h^{\star}(v)\}$.
    \end{obs}

	Applying this to $\Psi$, we see that for every ``dual'' vector $v$ we obtain the lower bound $\Psi(x^*) \ge \ip{x^*}{v} - \Psi^{\star}(v)$. The crucial step is then finding the right dual. As it is often the case, such dual is obtained by taking a positive combination of the constraint vectors $A(\tau)$ of the problem. In fact, we will show that using $\bar{v} = \beta \cdot \int_0^{\tau_{\final}} A(\tau)\, \d \tau$, for a scalar $\beta > 0$ to be chosen later, is an adequate choice. To start, since $x^*$ satisfies all the constraints, we have
	\begin{align*}
		\Psi(x^*) \ge \ip{x^*}{\bar{v}} - \Psi^{\star}(\bar{v}) = \beta \cdot \int_0^{\tau_{\final}} \underbrace{\ip{x^*}{A(\tau)}}_{\ge 1} \, \d \tau - \Psi^{\star}(\bar{v}) &\ge \beta \cdot \tau_{\final} - \Psi^{\star}(\bar{v}) \notag\\
		&\ge \frac{\beta}{2} \Psi(x_{\final}) - \Psi^{\star}(\bar{v}), %\label{eq:coverLB}
	\end{align*}
	where the last inequality uses the upper bound on the algorithm from \Cref{lemma:coverUB}. 
 
    Below, we will show that
    \begin{align}
    \label{eq:psiStarFinal}
        \Psi^{\star}(\bar{v}) \le\, \big(\beta \cdot c \big)^{q'} \cdot p' \cdot \Psi(x_{\final})
    \end{align}
    for some $c = O(p \log^2 {d\rho\gamma})$, where $q'$ is the H\"older dual of $p'$, namely the scalar satisfying $\frac{1}{p'} + \frac{1}{q'} = 1$. But first we use this to complete the proof of \Cref{thm:cover}.
    
    \Cref{eq:psiStarFinal} implies
	\[
		\Psi(x^*) \ge \bigg(\frac{\beta}{2} - (\beta \cdot c)^{q'} \cdot p' \bigg)\, \Psi(x_{\final})\,.
	\]
	Setting $\beta = \big(\frac{1}{2 q' p' c^{q'}}\big)^{1/(q'-1)}$, which maximizes the right-hand side, we get
	\begin{align*}
			\Psi(x^*) \,\ge\, \frac{1}{2} \cdot \bigg( \frac{1}{2q'} \bigg)^{\frac{1}{q'-1}} \cdot \bigg( \frac{1}{p' c} \bigg)^{\frac{q'}{q'-1}} \cdot \Psi(x_{\final}) \,\ge\, \bigg(\frac{1}{O(p' p \log^2 {d\rho\gamma})} \bigg)^{p'} \cdot \Omega(\Psi(x_{\final})),
	\end{align*}
    the last inequality using the fact $\frac{q'}{q' - 1} = p'$, which follows since by definition $q'$ satisfies $\frac{1}{p'} + \frac{1}{q'} = 1$. Finally, recalling $\Psi(y) = \frac{1}{p'} \|F(y)\|^{p'}$, we can multiply both sided by $p'$, take $p'$-roots and reorganize the expression to obtain that $$ \ALG \,=\, \|F(x_{\final})\| \,\le\, O(p'\,p \log^2 {d\rho\gamma})\cdot \|F(x^*)\| \,=\, O(p'\,p \log^2 {d\rho\gamma})\cdot \OPT.$$ This proves that our algorithm is $O(p'\,p \log^2 {d\rho\gamma})$-competitive and concludes the proof of \Cref{thm:cover}. Thus, we need to show that the dual value $\Psi^{\star}(\bar{v})$ is comparable to (a scaling of) the primal quantity $\Psi(x_{\final})$.

    \smallskip
    \paragraph{Proof of \Cref{eq:psiStarFinal}.}  The key for relating these primal and dual space is the gradient $\nabla \Psi$. 
%	Given the ``implicit'' definition of $\Psi^{\star}$, it is less intuitive how to do this. The key insight is that it suffices to show that the dual $\bar{v}$ is approximately upper bounded by $\nabla \Psi(x_{\final})$. 
	The intuition being this is the following: it is a classical fact that $\Psi^{\star}(\nabla \Psi(y)) = \ip{y}{\nabla \Psi(y)} - \Psi(y)$ for every $y$ (this holds for every convex function, not just $\Psi$, see Theorem E.1.4.1 of~\cite{HUL}). Moreover, since $\Psi(y) = \frac{1}{p'} \|(f_1(y|_{S_1}), \ldots, f_k(y|_{S_k}))\|^{p'}$, where each $f_\ell$ is a norm, $\Psi(y)$ should ``grow at most like power $p'$'', and so ``$\nabla \Psi(y)$ times $y$'' should not be larger than $p' \cdot \Psi(y)$; thus, heuristically we should have $$\underbrace{\Psi^{\star}(\nabla \Psi(y))}_{dual} = \ip{y}{\nabla \Psi(y)} - \Psi(y) \lesssim p'\cdot \Psi(y) - \Psi(y) = (p' - 1) \cdot \underbrace{\Psi(y)}_{primal}.$$
    This heuristic argument is indeed correct as we show in the first part of \Cref{lemma:psiStar}. In the second part of \Cref{lemma:psiStar}, we show that $\Psi(\alpha z) = \alpha^{q'} \cdot \Psi^\star(z)$. This way, showing \eqref{eq:psiStarFinal} is reduced to showing that $\Psi^\star(\bar{v}) \leq \Psi^\star(\beta \cdot c \cdot \nabla \Psi(x_{\final}))$, which will be \Cref{lemma:dual}.

    To make the argument formal, we first find an expression of $\Psi^{\star}$ in terms of the convex conjugate $g^{\star}$ of $g$ and in terms of the dual norms $f_{\ell,\star}$ of $f_\ell$.
    \begin{lemma}
    \label{lemma:psiStarForm}
    Let $g(y) = \frac{1}{p'} \|y\|^{p'}$, so $\Psi(y) = g(F(y))$. For every $z = (z^1,\ldots,z^k)$, where $z^\ell$ is $|S_\ell|$-dimensional, we have
	\begin{align*}
		\Psi^{\star}(z) = g^{\star}\Big(f_{1,\star}(z^1),\ldots,f_{k,\star}(z^k) \Big),
	\end{align*} 
	where $f_{\ell,\star}$ is the dual norm of $f_\ell$ defined by $f_{\ell,\star}(z) = \max_{w: f_{\ell}(w) = 1} |\ip{z}{w}|$.
    \end{lemma}
    
    \begin{proof}
    Writing the definition of $\Psi^{\star}$ we have (writing vectors as unit-sized directions $w^\ell$ and lengths $\lambda_\ell$)
	\begin{align*}
		\Psi^{\star}(z) &= \max\bigg\{\sum_\ell \ip{z^\ell}{\lambda_\ell\, w^\ell} - \Psi(\lambda_1\, w^1,\ldots, \lambda_k\, w^k) : f_1(w^1) = \ldots = f_k(w^k) = 1, \lambda_\ell \ge 0 ~\forall \ell  \bigg\}   \notag\\
		 &= \max\bigg\{\sum_\ell \lambda_\ell\, \ip{z^\ell}{w^\ell} - g(\lambda_1, \ldots, \lambda_k) : f_1(w^1) = \ldots = f_k(w^k) = 1, \lambda_\ell \ge 0 ~\forall \ell  \bigg\}  \notag\\
		 &= \max\bigg\{\sum_\ell \lambda_\ell \,f_{\ell,\star}(z^\ell)  - g(\lambda_1, \ldots, \lambda_k) : \lambda_\ell \ge 0 ~\forall \ell  \bigg\} \notag\\
		 &= g^{\star}(f_{1,\star}(z^1), \ldots, f_{k,\star}(z^k)),  
	\end{align*}
	where the second equation is because $f_\ell(\lambda\, w^\ell) = \lambda_\ell$ by the normalization of $w^\ell$, the third equation is by the definition of the dual norm, and the last equation is by the definition of convex conjugate, proving the lemma.
    \end{proof}

    Next, we show the lemma which relates $\Psi^\star$ to $\Psi$ for multiples of the gradient of $\Psi$.
	\begin{lemma} \label{lemma:psiStar}
	We have the following:
	\begin{enumerate}
		\item If $p' \ge 2$, then for any $y \in R^k_+$ it holds $\Psi^{\star}(\nabla \Psi(y)) \le (p'-1) \cdot \Psi(y)$.

		\item  Let $q'$ be the H\"older dual of $p'$, namely the scalar satisfying $\frac{1}{p'} + \frac{1}{q'} = 1$. Then $\Psi^{\star}(\alpha y) = \alpha^{q'} \cdot \Psi^{\star}(y)$ for every $\alpha > 0$.
	\end{enumerate}		
	\end{lemma}	 
	
	\begin{proof}
	
%	Now back to $\Psi^{\star}(\nabla \Psi(y))$. 
%		\begin{align*}
%			\Psi^{\star}(\nabla \Psi(y)) \red{=} g^{\star}(\|\nabla_{S_1} \psi(y)\|_{1,\star}, \ldots, \|\nabla_{S_k} \psi(y)\|_{k,\star}). 
%		\end{align*}
		We have
		\begin{align*}
			(\nabla \Psi(y))|_{S_\ell} = (\nabla_\ell g)(F(y)) \cdot \nabla f_\ell (y|_{S_\ell}).
		\end{align*}
		Since for any norm $\vertm{\nabla \vertm{y}}_{\star} = 1$, we have
        \begin{align}
            f_{\ell,\star}((\nabla \Psi(y))|_{S_\ell}) = (\nabla_\ell g)(F(y)) \label{eq:fstarGradg}
        \end{align}
        (which is a scalar). Then from \Cref{lemma:psiStarForm} we get $\Psi^{\star}(\nabla \Psi(y)) = g^{\star}( (\nabla g)(F(y)))$. Moreover, $g$ ``grows at most like power $p'$\,'', namely $\ip{\nabla g(y)}{y} = \|y\|^{p'-1} \ip{\nabla \|y\|}{x} = \|x\|^{p'} = p'\cdot g(y)$ for every $y$. Lemma 4.b of \cite{AzarBCCCG0KNNP16} then guarantees that $g^{\star}(\nabla g(y)) \le (p'-1) \cdot g(y)$ for every $y \ge 0$. Combining these observations gives
		\begin{align*}
			\Psi^{\star}(\nabla \Psi(y)) = g^{\star}( (\nabla g)(F(y))) \le (p'-1) \cdot g(F(y)) = (p'-1) \cdot \Psi(y),
		\end{align*}
        proving the first item in the lemma. 
	 
	 For the second item, it can be shown that $g^{\star}(y) = \frac{1}{q'} \|y\|^{q'}_{\star}$~(see for example \cite{borweinUnif}), and so $g^{\star}(\alpha y) = \alpha^{q'} \cdot g^{\star}(y)$. Then using \Cref{lemma:psiStarForm} we have for every $y = (y^1,\ldots,y^k)$ and scaling $\alpha \ge 0$ 
	 \begin{align*}
	 	\Psi^{\star}(\alpha y) = g^{\star}\Big(f_{1,\star}(\alpha \,y^1), \ldots, f_{k,\star}(\alpha \, y^k) \Big) =  g^{\star}\Big(\alpha \cdot \big(f_{1,\star}(y^1), \ldots, f_{k,\star}(y^k)\big) \Big) = \alpha^{q'} \cdot \Psi^{\star}(y),
	 \end{align*}
	 as desired.
	\end{proof}

	Thus, the core of the argument is showing that our dual's size $\Psi^{\star}(\bar{v})$ can be upper bounded using the gradient's size $\Psi^{\star}(\nabla \Psi(x_{\final}))$. This is precisely what is done in the next lemma.
	
	\begin{lemma} \label{lemma:dual}
		It holds that $$\Psi^{\star}(\bar{v}) \,\le\, 4\,\Psi^{\star}\big(\beta \cdot O(p \log^2 {d\rho \gamma}) \cdot \nabla \Psi(x_{\final})\big) + 4\,\Psi^{\star}\big(\beta \cdot O(p \log^2 {d\rho \gamma})) \cdot \delta \ones\big).$$
	\end{lemma}

    Note that combining \Cref{lemma:psiStar} and \Cref{lemma:dual}, \Cref{eq:psiStarFinal} is immediate because we can use \Cref{lemma:psiStar}.(2) to pull out the constant terms
	\begin{align*}
		\Psi^{\star}(\bar{v}) \,&\le\, \big(\beta \cdot O(p \log^2 {d\rho\gamma})\big)^{q'} \cdot \Big( \Psi^{\star}(\nabla \Psi(x_{\final})) + \Psi^{\star}(\delta \ones)\Big)\\
  &\le\, \big(\beta \cdot O(p \log^2 {d\rho\gamma})\big)^{q'} \cdot \Big((p'-1) \cdot \Psi(x_{\final}) + \Psi(x_{\final})\Big),
	\end{align*}
    where the last inequality also uses \Cref{lemma:psiStar}.(1) and the choice of $\delta$ that guarantees $\Psi^{\star}(\delta \ones) \le \Psi(x^*) \le \Psi(x_{\final})$.

    %(We remark that the last term is just a technicality and can be made as small as desired by making $\delta \approx 0$.)
    So, it only remains to show \Cref{lemma:dual}. We note that this is the only place in the argument where we use the fact that the norms $\|\cdot\|$ and $f_1,\ldots,f_k$ in the objective function are $p'$- and $p$-\monotone, respectively. In fact, suppose the gradient $\nabla \Psi(y)$ were monotone, which is the case considered in~\cite{AzarBCCCG0KNNP16}, and happens when the inner norms $f_\ell$'s are trivial, e.g. they are over just 1 coordinate each. In this case, since the update of algorithm satisfies $A(\tau) \approx \frac{\dot{x}(\tau)}{x(\tau)} \nabla \Psi(x(\tau))$, integrating gives (we will cheat and start the integration at $\tau = \e$, the initial times can be handled separately)
	\begin{align*}
		\bar{v} \approx \beta \cdot \int_\e^{\tau_{\final}} A(\tau) \, \d \tau \stackrel{mono}{\lesssim} \beta \cdot \nabla \Psi(x_{\final}) \cdot \int_0^{\tau_{\final}} \frac{\dot{x}(\tau)}{x(\tau)} \, \d \tau = \beta \cdot \nabla \Psi(x_{\final}) \cdot \log\bigg(\frac{x_{\final}}{x(\e)}  \bigg);
	\end{align*}
	using the monotonicity of $\Psi^{\star}$, one quickly obtains \Cref{lemma:dual} in this case. Unfortunately, the presence of the (non-trivial) norms $f_\ell$'s makes the gradient $\nabla \Psi$ non-monotone, which complicates matters.

%########################################################################
%########################################################################
		
	\subsection{Finding the right dual: Proof of \Cref{lemma:dual}} \label{sec:dual}
	
To simplify the notation, we use the following to denote the needed restrictions to a set of coordinates $S_\ell$: $\nabla_{S_\ell} \Psi(y) := (\nabla \Psi(y))|_{S_\ell}$,  $\bar{v}^\ell := \bar{v}|_{S_\ell}$, $A^\ell(\tau) := A(\tau)|_{S_\ell}$, and $x^\ell(\tau) := x(\tau)|_{S_\ell}$. As in the proof of \Cref{lemma:psiStar}, let $g(y) := \frac{1}{p'} \|y\|^{p'}$, so that $\Psi(y) = g(F(y))$. %Before delving into the proof, we record here a property of the norms $f_\ell$ and their duals $f_{\ell \star}$, conferred by \Cref{ass:cover}, namely that they are sandwiched between the $\ell_1$ and $\ell_\infty$ norms. 

    Fix a part $\ell$ throughout, and we prove the above inequality for it. Recall from the discussion in the previous section that the main difficulty is that $\nabla_{S_\ell} \Psi(y) = \nabla_\ell g(F(y)) \cdot \nabla f_\ell(y)$ may not be non-decreasing. Since the outer norm $\|\cdot\|$ is assumed to be $p'$-\monotone and $g(y) = \frac{1}{p'} \|y\|^{p'}$, the first term in this gradient is actually monotone, so the issue is that the gradient of the norm $\nabla f_\ell(y)$ may not be non-decreasing. To handle this, we use the same idea as in \cite{NS-ICALP17}, namely to break the evolution of our algorithm into phases where $f_\ell$ behaves as if it had (almost) monotone gradient. It is not clear that for a general norm we can obtain an effective bound on the number of these phases, since the coordinates of $\nabla f_\ell(x^\ell(\tau))$ may increase and decrease multiple times as $\tau$ evolves. Here is where we crucially rely on the assumption that the norm $f_\ell$ is $p$-\monotone, which, as we will see, guarantees that it suffices to control the \emph{value} of the norm $f_\ell(x^\ell(\tau))$ to obtain the desired control on its gradient.  
 
    Recall that the norm $f_\ell(x^\ell(\tau))$ of our solution only increases over time $\tau$. {Let $\max_{f_\ell} := \max_{i \in S_\ell} f_\ell(e_i)$, and $\min_{f_\ell} := \max_{i \in S_\ell} f_\ell(e_i)$ denote the maximum and minimum values of the norm $f_\ell$ for a coordinate vector in $S_\ell$.} Then define the times $t_1, t_2,\ldots, t_w = \tau_{\final}$ as follows: 
	
	\begin{enumerate}
		\item (Phase zero) $t_1$ is the largest time $\tau$ such that $f_\ell(x^\ell(\tau))^{p-1} \le {\big(\frac{\min_{f_\ell}^2}{d^2 \cdot \max_{f_\ell}}\big)^{p-1}}$.
		
		\item (Other phases) A new phase starts when $f_\ell(x^\ell(\tau))^{p-1}$ doubles. More precisely, $t_j$ is the largest time $\tau$ such that $$f_\ell(x^\ell(\tau))^{p-1} \le 2 \cdot f_\ell(x^\ell(t_{j-1}))^{p-1}.$$
	\end{enumerate}
			
%	The main point of this definition of phases is that the gradient $\nabla f_1$ (without raising to any power) is almost monotone within a phase. This is the only place we use the fact that the norm $f_\ell$ is $p$-\monotone. 

    The following lemma formalizes the almost monotonicity of the gradient $\nabla \Psi$ within a phase. 
 
	\begin{lemma} \label{lemma:monoPsi}
		For any $\tau \in [t_j, t_{j+1}]$, we have $$\nabla f_\ell(x^\ell(\tau)) \le 2\, \nabla f_\ell(x^\ell(t_{j+1})).$$ In particular, we have $$\nabla_{S_\ell} \Psi(x(\tau)) \le 2\, \nabla_{S_\ell} \Psi(x(t_{j+1})).$$
	\end{lemma}
	
	\begin{proof}
            Since the norm $f_\ell$ is $p$-\monotone, $\nabla_i f_\ell(x^\ell(\tau))^p$ is non-decreasing as we increase $\tau$. From chain rule, we can relate this quantity to the gradient of the norm as $\nabla_i f_\ell(x^\ell(\tau))^p = p \cdot f_\ell(x^\ell(\tau))^{p-1} \cdot \nabla_i f_\ell(x^\ell(\tau))$, which rearranging gives
		\begin{align*}
			\nabla_i f_\ell(x^\ell(\tau)) = \frac{\nabla_i f_\ell(x^\ell(\tau))^p}{p\cdot f_\ell(x^\ell(\tau))^{p-1}} \le \frac{\nabla_i f_\ell(x^\ell(t_{j+1}))^p}{p \cdot f_\ell(x^\ell(\tau))^{p-1}} 
   %&= \frac{f_\ell(x^\ell(t_{j+1}))^{p-1}}{f_\ell(x^\ell(\tau))^{p-1}} \cdot \frac{\nabla_i f_\ell(x^\ell(t_{j+1}))^p}{p \cdot f_\ell(x^\ell(t_{j+1}))^{p-1}}\\
			&= \frac{f_\ell(x^\ell(t_{j+1}))^{p-1}}{f_\ell(x^\ell(\tau))^{p-1}} \cdot \nabla_i f_\ell(x^\ell(t_{j+1}))\\
			&\le 2\, \nabla_i f_\ell(x^\ell(t_{j+1})),
		\end{align*}
		where the first inequality uses p-\monotonicity because $x^\ell(\tau) \leq x^\ell(t_{j+1})$ and the second inequality follows from the definition of a phase. This proves the first statement of the lemma. 

        The second statement follows from $\nabla_{S_\ell} \Psi(y) = \nabla_\ell g(F(y)) \cdot \nabla f_\ell(y)$ and the fact that $\nabla_\ell g(F(y))$ is non-decreasing, due to the $p'$-\monotonicity of $\|\cdot\|$, as discussed before.
	\end{proof}
	
	We also show that there are not too many phases. 
		
	\begin{lemma} \label{lemma:numberPhases}
		There are at most $O(p \log d \rho {\gamma})$ phases. 
	\end{lemma} 
	
	\begin{proof} 
            We need to upper bound how large $f_\ell(x^\ell(\tau))^{p-1}$ can be. 
		Since the non-zero entries of the constraint vectors $A(\tau)$ are at least $1/\rho$, our solution $x(\tau)$ never raises a coordinate above $\rho$. Thus, the monotonicity of the norm $f_\ell$ gives $f_\ell(x^\ell(\tau)) \le f_\ell(\rho\, \ones_{S_\ell})$, where $\ones_{S_\ell}$ denotes the incidence vector of the coordinates $S_\ell$. Moreover, using triangle inequality 
        %and the assumption $f_\ell(e_i) = 1$ 
        we have $f_\ell(\rho\, \ones_{S_\ell}) \le \rho \cdot \sum_i f_\ell(e_i) \le {d \rho \max_{f_\ell}}$; so $f_\ell(x^\ell(\tau))^{p-1} \le ({d \rho \max_{f_\ell}})^{p-1}$. 
		
		Since the first phase starts with $f_\ell(x^\ell(\tau))^{p-1} = {{\big(\frac{\min_{f_\ell}^2}{d^2 \cdot \max_{f_\ell}}\big)^{p-1}}}$ and the value doubles with each phase, the total number of phases is at most {$(p-1) \log_2 \big(d^3 \rho \,\frac{\max_{f_\ell}^2}{\min_{f_\ell}^2}\big) = O(p \log d \gamma)$, recalling that by definition $\gamma = \max_\ell \frac{\max_{f_\ell}}{\min_{f_\ell}}$}. This proves the lemma.  
	\end{proof}

    Recall from \Cref{lemma:psiStarForm} that $\Psi^{\star}(\bar{v}) = g^{\star}(f_{1,\star}(\bar{v}^1),\ldots,f_{k,\star}(\bar{v}^k))$ and for any $\alpha > 0$
    \begin{align*}
    \Psi^{\star}(\alpha \cdot \nabla \Psi(x_{\final})) = g^{\star}( f_{1,\star}(\alpha \nabla \Psi_{S_1}(x_{\final})),\ldots, f_{k,\star}(\alpha \nabla \Psi_{S_k}(x_{\final} )));
    \end{align*}
    the following bound on the inner norms is then the core for proving \Cref{lemma:dual}.

    \begin{lemma} \label{lemma:dualMain}
    We have
	\begin{align*}
	f_{\ell \star}\big(\bar{v}^\ell\big)  \,\le\, \beta \cdot O(p \log^2 {d\rho \gamma}) \cdot \Big(f_{\ell \star} \big(\nabla_{S_\ell} \Psi(x_{\final})\big) + \delta \cdot f_{\ell \star}(\ones)\Big).
	\end{align*}
    \end{lemma}

    \begin{proof}
        Recall $\bar{v}^\ell = \beta \cdot \int_0^{\tau_{\final}} A^\ell(\tau)\, \d \tau$. We upper bound the quantity $f_{\ell \star}\big(\int_{t_j}^{t_{j+1}} A^\ell(\tau)\,\d \tau\big)$ for each phase $j$ and then put them together to obtain the result. For that, recall that by definition of our algorithm, the continuous updates the solution $x(\tau)$ satisfies
	\begin{align}
	A_i(\tau) = \frac{\dot{x}_i(\tau)}{x_i(\tau) + \frac{1}{d}} \cdot (\nabla_i \Psi(x(\tau)) + \delta), ~~~~~~\forall i.  \label{eq:coverUpdate}
	\end{align}
 
%########################################################################

	\paragraph{Phase zero.}
 We have for all $i \in S_\ell$ and all $\tau \in [0, t_1]$
 \begin{align*}
     \nabla_i \Psi(x(\tau)) \, = \, (\nabla_\ell g)(F(x(\tau))) \cdot \nabla_i f_\ell(x^\ell(\tau)) \, \leq \, (\nabla_\ell g)(F(x(t_1))) {\cdot {\textstyle \max_{f_\ell}}}, 
 \end{align*}
 	where we use that {$\|\nabla f_\ell(y)\|_{\infty} = \max_{i \in S_\ell} \ip{e_i}{\nabla f_\ell(y)} \leq \max_{z \geq 0, f_\ell(z) \le \max_{f_\ell}} \ip{z}{\nabla f_\ell(y)} = f_{\ell, \star}(\nabla f_\ell(y)) \cdot \max_{f_\ell} = \max_{f_\ell}$ because $f_\ell(e_i) \le \max_{f_\ell}$ for all $i \in S_\ell$} and, for any norm, the dual norm of any of its gradients is always $1$. 
 
Therefore, integrating \Cref{eq:coverUpdate} from time $0$ to time $t_1$, we get for every $i \in S_\ell$
	\begin{align*}
		\int_0^{t_1} A_i(\tau) \,\d \tau &\le \int_0^{t_1} \bigg(d\cdot  \dot{x}_i(\tau) \cdot (\nabla_i \Psi(x(\tau)) + \delta) \bigg) \d\tau\\
		% &\le d \cdot \int_0^{t_1} \bigg(\dot{x}_i(\tau) \cdot (\nabla_\ell g)(F(x(\tau))) \cdot \nabla_i f_\ell(x^\ell(\tau)) \bigg) \,\d\tau \,+\, d \delta \cdot x_i(t_1)  \\
		% &\le d \cdot (\nabla_\ell g)(F(x({t_1}))) \cdot  \int_0^{t_1} \bigg(\dot{x}_i(\tau) \cdot \underbrace{\nabla_i f_\ell(x^\ell(\tau))}_{\le 1} \bigg) \,\d\tau  \,+\, d \delta \cdot x_i(t_1)\\
  		&\le d \cdot \Big((\nabla_\ell g)(F(x(t_1)))  {\cdot {\textstyle \max_{f_\ell}}} + \delta\Big) \cdot \int_0^{t_1} \dot{x}_i(\tau) \d\tau \\
		&= d \cdot \Big((\nabla_\ell g)(F(x(t_1)))  {\cdot {\textstyle \max_{f_\ell}}} + \delta\Big) \cdot x_i(t_1).
	\end{align*}

 Collecting all coordinates $i \in S_{\ell}$ and applying the dual norm $f_{\ell \star}$ on both sides gives
	\begin{align}
	f_{\ell \star}\bigg(\int_0^{t_1} A^\ell(\tau) \bigg) &\le d \cdot \Big((\nabla_\ell g)(F(x(t_1))) {\cdot {\textstyle \max_{f_\ell}}}+  \delta\Big) \cdot f_{\ell\star}(x^\ell(t_1)). \label{eq:coverZero}
    \end{align}
    We now need the following estimate relating $f_{\ell \star}$ to $f_{\ell}$, which uses the fact that $f_\ell$ is monotone.  

    \begin{claim}
        For every $y \ge 0$, we have $f_{\ell \star}(y|_{S_\ell}) \le \frac{d}{\min_{f_\ell}^2}\, f_\ell(y|_{S_\ell})$. 
    \end{claim}

    \begin{proof}
        First, for any vector $x \ge 0$, by triangle inequality we have the following upper bound on $f_\ell(x)$: $f_\ell(x) \le \sum_i x_i \,f_\ell(e_i) \le \|x\|_1 \cdot \max_{f_\ell}$. We also have the lower bound $f_\ell(x) \ge \|x\|_{\infty} \cdot \min_{f_\ell}$: To see this, let $x_{i'}$ be the largest coordinate of $x$; then by monotonicity of $f_\ell$, we have $f_\ell(x) \ge f_\ell(x_{i'}) = \|x\|_{\infty} \cdot f_\ell(e_{i'}) \ge \|x\|_{\infty} \cdot\min_{f_\ell}$. Finally, by duality, the latter lower bound implies $f_{\ell\star}(x) \le \frac{1}{\min_{f_\ell}} \|x\|_1$: 
        \begin{align*}
        f_{\ell \star}(x) = \max_{z : f_\ell(z) \le 1} \ip{z}{x} \le \max_{z : \|z\|_\infty \cdot \min_{f_\ell} \le 1} \ip{z}{x} = \max_{z : \|z\|_\infty \le 1} \ip{\tfrac{z}{\min_{f_\ell}}}{x} = \frac{1}{ \min_{f_\ell}} \|x\|_1,
        \end{align*}
        as desired. 

        Since the vector $y|_{S_\ell}$ has at most $d$ non-zero coordinates, it satisfies $\|y|_{S_\ell}\|_1 \le d \cdot \|y|_{S_\ell}\|_\infty$. Combining this with the above upper bound on $f_{\ell \star}$ and lower bound on $f_\ell$, we get $f_{\ell \star}(y|_{S_\ell}) \le \frac{d}{\min_{f_\ell}^2} f_{\ell}(y|_{S_\ell})$ as desired. 
    \end{proof}

    Taking \eqref{eq:coverZero} then employing the above claim and then the definition of the time $t_1$ of the first phase, we obtain
    \begin{align*}
	f_{\ell \star}\bigg(\int_0^{t_1} A^\ell(\tau) \bigg) 	&\le \frac{d^2}{\min_{f_\ell}^2} \cdot \Big((\nabla_\ell g)(F(x(t_1))) {\cdot {\textstyle \max_{f_\ell}}} + \delta \Big)\cdot f_\ell(x^\ell(t_1)) \\
	&\le (\nabla_\ell g)(F(x(t_1))) + \frac{\delta}{\max_{f_\ell}}\\
    &\le (\nabla_\ell g)(F(x_{\final})) + \frac{\delta}{\max_{f_\ell}}, 
    \end{align*}
    
	%where the second inequality uses $f_{\ell\star}(y) \le \sum_{i \in S_\ell} y_i f_{\ell\star}(e_i) = \sum_{i \in S_\ell} y_i = \sum_{i \in S_\ell} y_i f_{\ell}(e_i) \leq \lvert S_\ell \rvert f_{\ell}(y)$; the third inequality follows from the definition of the phase zero, and the
    where the last inequality follows from the monotonicity of the gradient $\nabla g(\cdot) = \nabla \frac{1}{p'} \|\cdot\|^{p'}$. 
	
%########################################################################

	\paragraph{Phase $j$.}  We now move on to upper bounding the integral $\int A_i(\tau) \, \d \tau$ for each phase $j > 0$. Integrating \eqref{eq:coverUpdate} between $t_j$ and $t_{j+1}$ and using the approximate monotonicity of $\nabla \Psi$ within a phase (\Cref{lemma:monoPsi}), we get
	\begin{align*}
		\int_{t_j}^{t_{j+1}} A_i(\tau) ~\le~ \Big(2 \, \nabla_i \Psi(x(t_{j+1})) + \delta\Big) \cdot \int_{t_j}^{t_{j+1}} \frac{\dot{x}_i(\tau)}{x_i(\tau) + \frac{1}{d}} \,\d\tau\,.
	\end{align*}
	Computing the last integral with the change of variables $y = x_i(\tau)$ (so $\frac{\d y}{\d \tau} = \dot{x}_i(\tau)$):
	\begin{align*}
		\int_{t_j}^{t_{j+1}} \frac{\dot{x}_i(\tau)}{x_i(\tau) + \frac{1}{d}} \,\d\tau ~=~ \int_{x_i(t_j)}^{x_i(t_{j+1})} \frac{1}{y + \frac{1}{d}} \,\d y &~\le~ \int_0^{1/d} \frac{1}{\frac{1}{d}} \,\d y + \int_{1/d}^{\max\{1/d, x_i(t_{j+1})\}} \frac{1}{y} \,\d y \\
		&~=~ 1 + \max\{0, \ln( d\cdot x_i(t_{j+1}))\};
	\end{align*}
    again since all coordinates of $x(\tau)$ are at most $\rho$, this integral is at most $O(\log d \rho)$.  Thus, collecting all coordinates $i \in S_\ell$ and applying the dual norm $f_{\ell \star}$, we obtain
	\begin{align*}
		f_{\ell \star}\bigg(\int_{t_j}^{t_{j+1}} A^\ell_i(\tau) \,\d \tau\bigg) ~&\le~ O(\log d\rho) \cdot f_{\ell \star} \Big(2 \, \nabla_{S_\ell} \Psi(x(t_{j+1})) + \delta \ones\Big) \\
  &\le~ O(\log d\rho) \cdot f_{\ell \star}\big(\nabla_{S_\ell} \Psi(x(t_{j+1}))\big) + O(\log d\rho) \cdot \delta \cdot f_{\ell \star}(\ones)\, .
	\end{align*}
    Using the fact $f_{\ell \star}\big(\nabla_{S_\ell} \Psi(x(t_{j+1}))\big) =  (\nabla_\ell g)(F(x(t_{j+1})))$ (from \eqref{eq:fstarGradg}), and then the fact the $\nabla g$ is non-decreasing gives the final bound
    \begin{align*}
  f_{\ell \star}\bigg(\int_{t_j}^{t_{j+1}} A^\ell_i(\tau) \,\d \tau\bigg) ~\le~ O(\log d\rho) \cdot (\nabla_\ell g)(F(x_{\final})) + O(\log d\rho) \cdot \delta \cdot f_{\ell \star}(\ones).
    \end{align*}

%########################################################################
	
	\paragraph{Adding over all phases.} Using triangle inequality on $f_{\ell \star}$ and adding the previous bounds over all the  $O(p \log {d \rho \gamma})$ phases (from \Cref{lemma:numberPhases}), we get
	\begin{align*}
		f_{\ell \star}\bigg(\int_0^{\tau_{\final}} A^\ell(\tau) \, \d \tau \bigg) ~&\le~ f_{\ell \star}\bigg(\int_0^{t_1} A^\ell(\tau) \, \d \tau \bigg) + \sum_{j = 1}^w f_{\ell \star}\bigg(\int_{t_j}^{t_{j+1}} A^\ell(\tau) \, \d \tau \bigg)  \\
  &\le~ O(p \log^2 {d\rho \gamma}))  \cdot (\nabla_\ell g)(F(x_{\final})) + \delta \cdot \bigg({\frac{1}{\max_{f_\ell}}} +  O(p \log^2 {d\rho \gamma}) \cdot f_{\ell \star}(\ones) \bigg).
	\end{align*}
    To clean up the last term, let $i' \in S_\ell$ be the coordinate achieving $f_\ell(e_{i'}) = \max_{f_\ell}$, so $f_\ell(\frac{e_{i'}}{\max_{f_\ell}}) = 1$. Then using the monotonicity of $f_{\ell \star}$, we have $$f_{\ell \star}(\ones) \ge f_{\ell \star}(e_{i'}) = \max_{z : f_\ell(z) \le 1} \ip{z}{e_{i'}} \ge \ip{\tfrac{e_{i'}}{\max_{f_\ell}}}{e_{i'}} = \frac{1}{\max_{f_\ell}}.$$ Thus, we obtain the cleaner expression
	\begin{align*}
		f_{\ell \star}\bigg(\int_0^{\tau_{\final}} A^\ell(\tau) \, \d \tau \bigg) ~&\le~ O(p \log^2 {d\rho \gamma}))  \cdot (\nabla_\ell g)(F(x_{\final})) + \delta \cdot O(p \log^2 {d\rho \gamma}) \cdot f_{\ell \star}(\ones).
	\end{align*}

    Using again $(\nabla_\ell g)(F(x_{\final})) = f_{\ell,\star}(\nabla_{S_\ell} \Psi(x_{\final}))$ (from \eqref{eq:fstarGradg}) and multiplying both sides by $\beta$ concludes the proof of \Cref{lemma:dualMain}. 	
	\end{proof}

    \Cref{lemma:dual} now follows by just tidying things up. 

    \begin{proof}[Proof of \Cref{lemma:dual}]
        Let $C_\ell := f_{\ell \star}\big(O(\beta p \log^2 {d\rho \gamma}) \cdot \nabla_{S_\ell} \Psi(x_{\final})\big)$ and $D_\ell := f_{\ell \star}(O(\beta p \log^2 {d\rho \gamma})) \cdot \delta\ones)$ be the terms in the right-hand side of the previous lemma, and $C, D$ be their respective vectors. Then recalling the formula for $\Psi^{\star}$ from \Cref{lemma:psiStarForm} and noticing that $g^{\star}$ is non-decreasing (seen from $g^{\star}(z) = \frac{1}{q'} \|z\|^{q'}_{\star}$), we have $\Psi^{\star}(\bar{v}) = g^{\star}\Big(f_{1,\star}(\bar{v}^1),\ldots,f_{k,\star}(\bar{v}^k) \Big) \le g^{\star}(C + D)$. This last term is at most $4 (g^{\star}(C) + g^{\star}(D))$, as we can see using the formula for $g^{\star}$ as
    	\begin{align*}
         g^{\star}(C + D) = \frac{1}{q'} \|C+D\|_{\star}^{q'} \le \frac{1}{q'} \Big(\|C\|_{\star} + \|D\|_{\star}  \Big)^{q'} &\le 2^{q'} \frac{1}{q'} \Big(\max\Big\{\|C\|_{\star}\,,\, \|D\|_{\star}\Big\} \Big)^{q'}\\
         &\le 4 (g^{\star}(C) + g^{\star}(D)),
    	\end{align*}         
     where the last inequality uses the fact $q' \le 2$, which is implied by the assumption $p' \ge 2$. Again by \Cref{lemma:psiStarForm} we have $g^{\star}(C) = \Psi^{\star}(O(\beta p \log^2 {d\rho \gamma}) \cdot \nabla \Psi(x_{\final}))$ and $g^{\star}(D) = \Psi^{\star}(O(\beta p \log^2 {d\rho \gamma}) \cdot \delta \ones)$, which finally proves \Cref{lemma:dual}.
    \end{proof}

\section{Applications to Packing Problems}

\mnote{Say something about computational complexity, and how $P$ is presented?}

We now consider a general online packing problem (\pack). In the offline version of this problem, there are $\numberofsteps$ items, each with a positive value $c_t > 0$ and a multidimensional size $(a_{1,t}, a_{2,t}, \ldots, a_{\numberofdimensions,t}) \in \R^{\numberofdimensions}_{\ge 0}$. There is a downward closed feasible set $P \subseteq \R^{\numberofdimensions}_{\ge 0}$ (i.e., for any two vectors $0 \le y \le x$, if $x$ belongs to $P$, then so does $y$). The goal is to fractionally select items that give maximum value and packing into $P$, namely
    \[
        \max\, \ip{c}{x} \qquad \textrm{s.t.}\quad  Ax \in P \text{ and } x \ge 0.
    \]
In the online version of the problem, the packing set $P$ is given upfront but the $\numberofsteps$ items arrive online one-by-one. When the $t$-th item arrives, its value $c_t$ and size vector $(a_{1,t}, a_{2,t}, \ldots, a_{\numberofdimensions,t})$ is revealed, and the algorithm needs to immediately and irrevocably set $x_t \geq 0$. The final vector $x$ has to fulfill $Ax \in P$. As always, we use $\OPT$ to denote the optimum value of the problem. 
    
Note that by taking $P = \{x \in \R^{\numberofdimensions}_{\ge 0} : x \le b\}$ for some vector $b$, the packing constraints become $Ax \le b$, and the problem becomes the classical one of online packing LPs~\cite{BN-MOR09}. %The flexibility of $P$ allows one to capture much more general constraints, such as the independence polytope of matroid intersection, and be even non-polyhedral constraint sets.  

Each downward-closed set $P \subseteq \R^{\numberofdimensions}_{\ge 0}$ has an associated  (semi-) norm $\|\cdot\|_P$ via the Minkowski functional, namely for every $x \ge 0$, $\|x\|_P := \inf_{\alpha > 0} \{ \alpha : \frac{x}{\alpha} \in P\}$ (page 53 of~\cite{schneider}). Since $P$ is the unit-ball of this norm, the packing constraint is equivalent to $\|Ax\|_P \le 1$, and \pack can be restated as
    \[
        \max\, \ip{c}{x} \qquad \textrm{s.t.}\quad  \|Ax\|_P \le 1 \text{ and } x \ge 0.
    \]
We give an online algorithm when $\|\cdot\|_P$ can be approximated by a $p$-\monotone norm.

\thmPacking*

For the remainder of this section we prove this result, starting from the case where the norm $\|\cdot\|_P$ itself is $p$-\monotone, i.e., $\alpha = 1$. We assume throughout that the instance is feasible, and it has bounded optimum, or equivalently, that for every non-negative direction $v \in \R^{\numberofdimensions}_{\ge 0} \setminus \{0\}$ we have $\|A v\|_P > 0$ (else $\gamma v$ would satisfy the packing constraint $\|Ax\|_P \le 1$ for all $\gamma \ge 0$, and give unbounded value as $\gamma \rightarrow \infty$). We also assume without loss of generality that $p \ge 2$, recalling that $p$-\monotonicity implies $p'$-\monotonicity for all $p'\ge p$.

%\red{XX} Due to monotonicity of the norm and non-negativity of the matrix entries, this implies that also every partial vector filled up with zeros also fulfills the constraint. 

%###############################################################

\subsection{Starting point: $\|\cdot\|_P$ is already $p$-\monotone}

\paragraph{Algorithm under a $\beta$-approximation of $\OPT$.} Without loss of generality, we can assume that all item values $c_t$ are equal to 1, by replacing the variables by $y_t := c_t x_t$ otherwise. 
%by rescaling the matrix entries by $a'_{i, t} = \frac{1}{c_t} a_{i, t}$ otherwise (notice this does not change the width parameter $\rho$ NOT TRUE!!). 
The starting point is the algorithm of Azar et al.~\cite{AzarBCCCG0KNNP16} for the related problem of online welfare maximization with convex costs: A convex cost function $\Psi$ is given upfront. As before, items come online, and when the $t$-th item arrives its size vector $(a_{1,t}, a_{2,t}, \ldots, a_{\numberofdimensions,t})$ is revealed, and the algorithm needs to set the variable $x_t$ irrevocably. Now the goal is to maximize the profit $\sum_t x_t - \Psi(Ax)$. 

Azar et al. \cite[Lemma 13]{AzarBCCCG0KNNP16} gives a $O(\frac{p \lambda}{\lambda - 1})$-competitive algorithm for this problem under the following assumptions:%\footnote{The result of~\cite{Anupam} holds in more generality than what is state in their Theorem 2, and the assumptions below suffice for their result; see the discussion in the beginning of their Section 4, or~\cite{molinaro}.}

    \begin{enumerate}[itemsep=0pt]
        \item $\Psi$ is non-decreasing with $\Psi(0) = 0$.
        
        \item $\Psi$ is differentiable everywhere except at $0$ and has non-decreasing gradients. Moreover, it satisfies the growth condition $\ip{\nabla \Psi(x)}{x} \le p \cdot \Psi(x)$ for all $x \in \R^{\numberofdimensions}_{\ge 0}$. 

        \item For every $\gamma \ge 1$ and $x \in \R^{\numberofdimensions}_{\ge 0}$ we have $\nabla \Psi(\gamma x) \ge \gamma^{\lambda-1} \cdot \nabla \Psi(x).$
        
        \item The optimal value of the instance is bounded, i.e., not $\infty$. 
    \end{enumerate}

    The idea for solving \pack to use the estimate $\widetilde{\OPT}$ to define a Lagragian relaxation $\sum_t x_t - \Psi(Ax)$ for a function $\Psi$ satisfying the requirements above, then apply the algorithm from~\cite{AzarBCCCG0KNNP16}. However, instead of using the estimate $\widetilde{\OPT}$ directly, it will pay off to actually randomly guess a better estimate within a factor of $\delta \in [1,\beta]$. Set $\delta := e^{p-1}$ if $p-1 \le \log \beta$, and $\delta := \beta$ otherwise.
    
\vspace{4pt}
\begin{mdframed}
    \vspace{-9pt}
    \begin{proc} \label{proc:pack}
     \normalfont
     \textbf{Online Packing\,($\widetilde{\OPT}$)} 

    \vspace{4pt}
    \noindent 1. Select $I$ uniformly randomly among the powers of $\delta$ $\{\delta, \delta^2,\ldots, \delta^{\lceil\log_{\delta} \beta\rceil}\}$. Define $\Psi(\cdot ) := \frac{I \cdot \widetilde{\OPT}}{\beta} \|\cdot\|_P^p$.

    \vspace{4pt}
    \noindent 2. In an online fashion, run the algorithm from Theorem 2 of~\cite{AzarBCCCG0KNNP16} on the problem $\sum_t x_t - \Psi(Ax)$, which computes a solution $\tilde{x}$. Play this solution until the packing constraints $Ax \in P$ are going to be violated, in which case play $x_t = 0$ from then on. Let $\bar{x}$ be the solution played
\end{proc}
\end{mdframed}

    First notice that by construction the solution $\bar{x}$ played by the algorithm is feasible. It remains to show that it is in expectation $O(\max\{p, \log \beta\})$-competitive for \pack. We first show that the result of Azar et al.~\cite{AzarBCCCG0KNNP16} can indeed be applied to our problem, and so $\tilde{x}$ has the desired guarantees. 
    
    \begin{lemma} \label{lemma:packAss}  
        For every scenario of $I$, $\tilde{x}$ is $O(p)$-competitive for the problem of maximizing $\sum_t x_t - \Psi(Ax)$. 
    \end{lemma}

    \begin{proof}
     We show that the problem $\sum_t x_t - \Psi(Ax)$ satisfies the assumptions 1-4 above for the guarantees Azar et al.~\cite{AzarBCCCG0KNNP16} to hold. 
    
    Item 1 follows from the fact since $P$ is a packing set, the norm $\|\cdot\|_P$ is monotone, and so it $\Psi$. For Item 2, since $\|\cdot\|_P$ was assumed to be differentiable and $p$-\monotone, $\Psi(x) = \frac{I \cdot \widetilde{\OPT}}{\beta}\, \|\cdot\|_P^p$ has non-decreasing gradients. For the growth condition in this item, we observe that $\nabla \Psi(x) = \frac{I \cdot \widetilde{\OPT}}{\beta}\, p \|x\|^{p-1}_P \cdot \nabla \|x\|_P$, so we get $\ip{\nabla \Psi(x)}{x} = \frac{I \cdot \widetilde{\OPT}}{\beta}\, p \|x\|^{p-1}_P \cdot \ip{\nabla \|x\|_P}{x} = \frac{I \cdot \widetilde{\OPT}}{\beta}\, p \|x\|^{p-1}_P \cdot \|x\|_P = p \cdot \Psi(x)$, where the next-to-last equation uses the fact that for every norm $\ip{\nabla 
\|x\|}{x} = \|x\|$ (\Cref{lemma:gradNorms}). For Item 3, recall that the gradient of any norm is invariant to positively scaling the argument (also \Cref{lemma:gradNorms}); thus, $\nabla \Psi(\gamma x) = \frac{I \cdot \widetilde{\OPT}}{\beta}\, p \|\gamma x\|^{p-1}_P \cdot \nabla \|\gamma x\|_P = \frac{I \cdot \widetilde{\OPT}}{\beta}\, \gamma^{p-1} p \|x\|^{p-1}_P \cdot \nabla \|x\|_P = \gamma^{p-1} \cdot \nabla \Psi(x)$ for all $\gamma \geq 1$. Finally, for Item 4, for every non-negative direction $v \in \R^{\numberofdimensions}_{\ge 0} \setminus \{0\}$ and $\gamma \ge 0$, we have $\sum_t (\gamma v_t) - \Psi(A (\gamma v)) = \gamma \sum_t v_t - \gamma^p \Psi(A v)$. Our assumption that the \pack instance has bounded optimum implies that the last term grows as $\Omega(\gamma^p)$, and since we assumed $p > 1$, the whole expression goes to $-\infty$ as $\gamma \rightarrow \infty$, and so the problem of maximizing $\sum_t x_t - \Psi(Ax)$ has bounded optimum. 
    
    Consequently,  the guarantee of $O(\frac{p^2}{p-1}) = O(p)$-competitiveness (the equation using the assumption that $p \ge 2$) from~\cite{AzarBCCCG0KNNP16} holds for the computed solution $\tilde{x}$, proving the lemma. 
    \end{proof}

    We say that the random guess $I$ is \emph{good} if the adjusted guess $\frac{I \cdot \widetilde{\OPT}}{\beta}$ of $\OPT$ is in the interval $[\OPT, \delta \cdot \OPT]$, or equivalently $I \in [\frac{\beta \OPT}{\widetilde{\OPT}}, \delta \cdot \frac{\beta \OPT}{\widetilde{\OPT}}]$. By the guarantees of $\widetilde{\OPT}$, this is an interval of multiplicative width $\delta$ within the interval $[1, \delta \beta]$, so one of the possibilities of $I$ lie in this interval; thus, $I$ is good with probability $\frac{1}{\lceil \log_\delta \beta \rceil}$. We show that whenever $I$ if good, then the algorithm did not have to stop playing $\tilde{x}$, so $\bar{x} = \tilde{x}$.

    \begin{lemma} \label{lemma:packFeas}
        Whenever $I$ is good, $\bar{x} = \tilde{x}$.
    \end{lemma}

    \begin{proof}
        If actually suffices to show that $\tilde{x}$ is feasible, which implies $\bar{x} = \tilde{x}$. Since $\tilde{x}$ is a $O(p)$-competitive solution for maximizing $\sum_t x_t - \Psi(Ax)$, comparing it against the all-zeros solution gives $\sum_t \tilde{x}_t - \Psi(A\tilde{x}) \ge 0$, i.e. $\Psi(A\tilde{x}) \le \sum_t \tilde{x}_t$. We can upper bound the right-hand side by observing that $\sum_t \frac{\tilde{x}_t}{\|A\tilde{x}\|_P} \leq \OPT$, since $\frac{\tilde{x}}{\|A\tilde{x}\|_P}$ is a feasible solution to the \pack problem (recall we assumed $c= \ones$). Combining these facts we get
        \begin{align*}
            \frac{I \cdot \widetilde{\OPT}}{\beta}\, \|A\tilde{x}\|^p_P \,=\, \Psi(A\tilde{x}) \,\le\, \OPT \cdot \|A\tilde{x}\|_P,
        \end{align*}
        and the goodness of $I$ then implies that $\|A\tilde{x}\|_P^{p-1} \le 1$, an hence $\|A\tilde{x}\|_P \le 1$. This proves the feasibility of $\tilde{x}$. 
    \end{proof}

    We can now prove that that expected value of $\bar{x}$ is at least $\frac{1}{O(\max\{p, \log \beta\})} \cdot \OPT$. Let $x^*$ be the optimal solution of \pack, hence $\sum_t x_t^* = \OPT$. Again using the fact that $\tilde{x}$ is $O(p)$-competitive for maximizing $\sum_t x_t - \Psi(Ax)$, comparing it against the solution $\gamma x^*$ for $\gamma = \frac{1}{(2\delta)^{1/(p-1)}}$, we get
    \begin{align*}
        \sum_t \tilde{x}_t - \Psi(A\tilde{x}) \ge \frac{1}{O(p)} \bigg(\sum_t \gamma x^*_t - \Psi(A \gamma x^*) \bigg) &= \frac{1}{O(p)} \bigg(\gamma \OPT - \gamma^p\, \frac{I \cdot \widetilde{\OPT}}{\beta}\, \underbrace{\|A x^*\|_P^p}_{\le 1} \bigg)\
    \end{align*}
    the underbrace following from the feasibility of $x^*$. Now, whenever $I$ is good, from \Cref{lemma:packFeas} we have $\sum_t \bar{x}_t = \sum_t \tilde{x}_t$ and $\frac{I \cdot \widetilde{\OPT}}{\beta} \le \delta \cdot \OPT$, which gives $\sum_t \bar{x}_t \ge \frac{\OPT}{O(p)} (\gamma - \gamma^p\cdot\delta ) \ge \frac{\gamma \OPT}{O(p)},$ the last inequality following from the definition of $\gamma$.
    
    Since $I$ is good with probability $\frac{1}{\lceil \log_\delta \beta \rceil}$, the expected value of the solution returned by our algorithm is at least 
    \begin{align*}
        \E\,\textrm{algo} \,\ge\, \frac{1}{\lceil \log_\delta \beta \rceil} \cdot \frac{1}{(2\delta)^{1/(p-1)}} \cdot \frac{1}{O(p)} \cdot \OPT \,=\, \frac{\log \delta}{(2\delta)^{1/(p-1)}} \cdot \frac{1}{O(p \log \beta)} \cdot \OPT \,.
    \end{align*}    

    When $p-1 \le \log \beta$, we defined $\delta = e^{p-1}$, and the above lower bound gives $\E\,\textrm{algo} \ge \frac{1}{O(\log \beta)} \cdot \OPT$. Otherwise, $p-1 > \log \beta$ and we defined $\delta = \beta$, and the bound becomes $\E\,\textrm{algo} \ge \frac{1}{(2 \beta)^{1/(p-1)}} \cdot \frac{1}{O(p)} \cdot \OPT \ge \frac{1}{\beta^{1/\log \beta}} \cdot \frac{1}{O(p)} \cdot \OPT = \frac{1}{O(p)} \cdot \OPT$. This proves that our algorithm is $O(\max\{p, \log \beta\})$-competitive, giving the first part of \Cref{thm:packing} when $\alpha = 1$. 

    \paragraph{Analysis without an approximation of $\OPT$.} 
% \tkc{I defined the width a little differently and I think it has to be like this. At least unless we open the black box: Only with the definition, we can get an estimate of $\OPT$ in the first step. Consider, for example
% \[
% A = \left(\begin{array}{cc}
%     1 & 0 \\
%     0 & \epsilon
% \end{array}\right)
% \]
% In the sense of Buchbinder-Naor, it has width $1$. But you really can't estimate $\OPT$ if you don't know the second column.
% }
\mnote{Should we make an algorithm box for this? Used when we have to approximate the norm by a $p$-mono}
The idea is to use the first item of the problem to compute an estimate $\widetilde{\OPT}$ of $\OPT$ and run the previous algorithm. More precisely, after seeing the information $c_1$ and $(a_{1,1}, \ldots, a_{\numberofdimensions,1})$ of the first item, let $a_{1,k}$ be any one of its non-zero sizes $a_{1,i}$ (which exists, since we assumed the instance has bounded optimum). We show below that $\frac{1}{\numberofdimensions \rho} \frac{c_1}{a_{1,k}\cdot \|e_k\|_P} \leq \OPT \leq \numberofdimensions \rho\cdot \frac{c_1}{a_{1,k}\, \|e_k\|_P}$. Therefore, we set the $\OPT$ estimate $\widetilde{\OPT} := \frac{1}{\numberofdimensions \rho} \frac{c_1}{a_{1,k}\cdot \|e_k\|_P}$, which is then a $(\numberofdimensions \rho)^2$-approximation of $\OPT$, and run  \Cref{proc:pack}. The previous analysis shows that this returns a solution that is in expectation $O(\max\{p, \log (\numberofdimensions \rho)^2\}) = O(\max\{p, \log \numberofdimensions \rho\})$ -competitive. This concludes the proof of \Cref{thm:packing} for the case $\alpha = 1$.  

\smallskip
It only has to be shown that $\OPT$ indeed falls into the described interval.

\begin{lemma}
 It holds that $\frac{1}{\numberofdimensions \rho} \frac{c_1}{a_{1,k}\cdot \|e_k\|_P} \leq \OPT \leq \numberofdimensions \rho\, \frac{c_1}{a_{1,k}\cdot \|e_k\|_P}$.
\end{lemma}

\begin{proof}
To obtain a lower bound on $\OPT$, consider the solution $x'$ given by $x'_1 = \frac{1}{\numberofdimensions \rho} \frac{1}{a_{1,k}\cdot \|e_k\|_P}$ and $\bar{x}_t = 0$ for $t \ge 2$. This solution is feasible: By triangle inequality $\|(a_{1,1}, \ldots, a_{\numberofdimensions,1})\|_P \le \sum_i a_{i,1} \|e_i\|_P \le \numberofdimensions \rho \cdot a_{i,k} \|e_k\|_P$, and so $\|A x'\|_P = x'_1 \cdot \|(a_{1,1}, \ldots, a_{\numberofdimensions,1})\|_P \le 1$, giving feasibility. Since $x'$ has value $c_1 x'_1 = \frac{1}{\numberofdimensions \rho} \frac{c_1}{a_{1,k}\cdot \|e_k\|_P}$, the lower bound on $\OPT$ follows. 

We now prove the desired upper bound on $\OPT$. For any item $t$, since  at least one of the $a_{t,i}$'s is strictly positive, by definition of $\rho$ we get $\sum_i \frac{a_{t,i} \cdot \|e_i\|_P}{c_t} \ge \frac{1}{\rho} \cdot \frac{a_{1,k} \cdot \|e_k\|_P}{c_1}$, or equivalently $\frac{\rho\,c_1}{a_{1,k} \cdot \|e_k\|_P}\cdot \sum_i (a_{t,i} \cdot \|e_i\|_P) \ge c_t$. Letting $x^*$ be an optimal solution and applying this upper bound on $c_t$, we get
\begin{align*}
    \OPT \,=\, \sum_t c_t x^*_t  \,\le\, \frac{\rho\,c_1}{a_{1,k} \cdot \|e_k\|_P} \sum_i \sum_t x^*_t \cdot (a_{t,i} \cdot \|e_i\|_P) \,&=\, \frac{\rho\,c_1}{a_{1,k} \cdot \|e_k\|_P}\, \sum_i \|(A x^*)_i \cdot e_i\|_P\\
    &\le  \frac{\numberofdimensions \rho\,c_1}{a_{1,k} \cdot \|e_k\|_P}\, \|Ax^*\|_P,
\end{align*}
where the last inequality follows from the monotonicity of $\|\cdot\|_P$. Since $x^*$ is feasible, $\|Ax^*\|_P \le 1$, and we obtain the desired upper bound on $\OPT$.
\end{proof}

%###############################################################
%###############################################################

\subsection{Extending to case $\alpha > 1$}

\mnote{M: Move this to the appendix? Just say `` apply the previous algorithm to the approximant $\vertm{\cdot}$}

Now suppose $\|\cdot\|_P$ is not necessarily $p$-\monotone, but it has a $p$-\monotone $\alpha$-approximation $\vertm{\cdot}$, i.e. $\|x\|_P \le \vertm{x} \le \alpha \cdot \|x\|_P$ for all $x \in \R^{\numberofdimensions}_+$. Then we can simply apply the results from the previous section to the approximant $\vertm{\cdot}$. 

More precisely, %let $P^{\vertm{\cdot}} := \{x \in \R^{\numberofdimensions}_+ : \vertm{x} \le 1\}$ be the packing set relative to the norm $\vertm{\cdot}$, 
and let $\OPT^{\vertm{\cdot}}$ be the optimal value for the \pack instance $\mathcal{I}^{\vertm{\cdot}}$ given by $\max \{\ip{c}{x} : \vertm{Ax} \le 1, ~x \ge 0\}$ relative to the new norm. Since $\|Ax\|_P \le \vertm{Ax}$, we get $\OPT \ge \OPT^{\vertm{\cdot}}$, and since $\vertm{\frac{Ax}{\alpha}} \le \|Ax\|_P$, we have $\OPT^{\vertm{\cdot}} \ge \frac{1}{\alpha}\,\OPT$. 

This means that if a $\beta$-approximation $\widetilde{\OPT}$ of $\OPT$ is available, then it gives an $\alpha \beta$-approximation to $\OPT^{\vertm{\cdot}}$. Thus, we can run \Cref{proc:pack}  over the new instance $\mathcal{I}^{\vertm{\cdot}}$ with estimate $\widetilde{\OPT}$ to obtain a solution $\bar{x}$. This solution is feasible for the original instance and has value at least $\frac{1}{O(\max\{p, \log \alpha \beta\})} \cdot \OPT^{\vertm{\cdot}} \ge \frac{1}{O(\max\{p, \log \alpha \beta\})} \frac{1}{\alpha} \OPT$, thus we obtain a $O(\alpha) \cdot \max\{p, \log \alpha \beta\}$-competitive solution for the original instance as desired. 

If an estimate of $\OPT$ is not available, we run the algorithm from the previous section that does not require such estimate and obtain a solution that is feasible for the original instance and has value $\frac{1}{O(\max\{p, \log \numberofdimensions \rho\})} \frac{1}{\alpha} \OPT$ (notice the definition of $\rho$ already has the factor $\alpha$ relative to the norm approximation). This concludes the proof of \Cref{thm:packing}.

\section{Applications to Stochastic Probing} \label{sec:stochProbing}

    Recall the stochastic probing problem (\stoch) introduced in \Cref{sec:introProbe}: There is a set $[n]$ of items, each with a non-negative value $X_i$ that is distributed according to some distribution $\mathcal{D}_i$. The values of the items are independent, but do not necessarily follow the same distribution. While the distributions $\mathcal{D}_i$'s are known to the algorithm, the actual values $X_i$'s are not; an item needs to be \emph{probed} for its value to be revealed. There is a downward-closed family of subsets of items $\mathcal{F} \subseteq[n]$ indicating the feasible sets of probes (e.g., $\mathcal{F}$ can consist of all subsets of size at most $k$ values from $[n]$, indicating that there is a budget of at most $k$ probes). Finally, there is a monotone norm $f : \R^n_+ \rightarrow \R_+$ indicating that if the set $S$ of items is probed, then the actual value obtained from them is $f(X_S)$, where $X_S$ is the vector what has coordinate $i$ equal to $X_i$ if $i \in S$, and equal to 0 otherwise. 
    
    For example, if we think of each item as a candidate, the function $f(x) = \max_i x_i$ models that while you can probe/interview a set $S \subseteq [n]$ of candidates, you may only hire the single best one, obtaining value $f(X_S) = \max_{i \in S} X_i$. The algorithm must then decide which feasible set of items $S \in \mathcal{F}$ to probe in order to maximize the expected value $\E f(X_S)$.

Let $\adapt = \adapt(\mathcal{D},\mathcal{F}, f)$ denote expected value of the optimal \emph{adaptive} strategy, namely the best strategy that probes items one-by-one, using the realized values $X_i$'s of the items already probed to decide which item to probe next. Let $\NA = \NA(\mathcal{D},\mathcal{F}, f)$ denote the expected value of the best \emph{non-adaptive} strategy that selects the whole set $S$ of probes upfront; that is, $\NA = \max_{S \in \mathcal{F}} \E f(X_S)$. We are interested in bounding the \emph{adaptivity gap} $\frac{\adapt(\mathcal{D},\mathcal{F}, f)}{\NA(\mathcal{D},\mathcal{F}, f)}$, namely the largest advantage that adaptivity can offer, for a family of instances. 

%Recall $f : \R^n_+ \rightarrow \R_+$ is $XOS$ if it is the maximum of monotone linear function, i.e., $f(x) = \max_{c \in C} \ip{c}{x}$ for a set $C \subseteq \R^n_+$ of non-negative vectors. Notice that these functions are in 1-1 correspondence to monotone norms (more precisely, the flip-symmetrized functions $x \mapsto f(|x|)$). 

    We show that $p$-\monotonicity suffices to bound the advantage offered by  adaptivity.

\thmStochProbing*

%\begin{thm} \label{thm:stochProbing}
%   For every $p$-\monotone objective function $f$, \stoch has adaptivity gap at most $O(p)$.
%\end{thm}

    To prove this result, we consider the non-adaptive strategy that ``hallucinates'' the values of the items, i.e., draws sample $\bar{X}_i \sim  \mathcal{D}$ for each value, and runs that optimal adaptive strategy using these samples, but obtaining true value given by the $X_i$'s. Notice that this strategy is indeed non-adaptive, since it never uses the $X_i$'s for decision-making. The idea of the analysis is to replace one-by-one the probes performed by $\adapt$ and the hallucinating strategy, similar to what was done for Load Balancing in \Cref{thm:loadBalancing}. 

%\red{M: Adaptive vs. non-adaptive, adaptivity gap, etc. Let $\OPT$ be the adaptive opt, and $S^*$ the set chosen by the optimal adaptive strategy (which is correlated with the $X_i$'s).     Consider an optimal adaptive solution. Let $nat^*$ be the natural non-adaptive strategy based on \emph{hallucination}: Construct the optimal set $\bar{S}^*$ based on an independent set of item values $\bar{X}_1,\ldots,\bar{X}_n$, and just select the items $\bar{S}^*$, obtaining value $f(X_{\bar{S}^*})$.} 

    \medskip In the remainder of the section, we prove this result under the following assumptions, which are discharged in Appendix \ref{app:stochProbing} (the first two are obtained by truncation, and the third by adding dummy items of 0 value):
   
    \begin{enumerate}[itemsep=1pt]
        \item For every $i$, $f(X_{\{i\}}) \le \frac{\adapt}{4c p}$ in every scenario.
        \item $f(X_{S^*}) \le 12\, \adapt$ in every scenario. 
        \item The optimal adaptive set of probes $S^*$ has the same size $m \le n$ in every scenario.   
    \end{enumerate}

    Since $f$ is a norm, from now on we use the notation $\|\cdot\| = f(\cdot)$, which is more natural. Let $I_1,\ldots,I_m \in [n]$ be the (random) sequence of items \adapt probes (so $S^* = \{I_i\}_i$). Recall that $\bar{X}_1,\ldots,\bar{X}_n$ is an independent copy of the sequence $X_1,\ldots,X_n$, and let $\bar{I}_1,\ldots,\bar{I}_m$ be the sequence of probes obtained by running \adapt over this copy (so $\bar{I}_1,\ldots,\bar{I}_m$ is an independent copy of $I_1,\ldots,I_m$). Define the vector $V_j := e_{I_j} X_{I_j}$ as the value of the item probed at the $j$th round, placed in the appropriate coordinate; notice that $X_{S^*} = V_1 + \ldots + V_m$, and so $\adapt = \E \|V_1 + \ldots + V_m\|$. Similarly, the value vector of the hallucinating strategy is given by the sum $\sum_j e_{\bar{I}_j} V_{\bar{I}_j}$ (i.e., probe $\bar{I}_j$ according to hallucination and see the real value $X_{\bar{I}_j}$). Notice this sum has the same distribution as using the true real optimal probing $I_j$ but receiving hallucinated value $\bar{X}_j$ (i.e., sequences $(\bar{I}_1, V_1), (\bar{I}_2, V_2), \ldots, (\bar{I}_m, V_m)$ and $(I_1, \bar{V}_1), (I_2, \bar{V}_2), \ldots, (I_m, \bar{V}_m)$ have the same distribution); the latter will be more convenient to work with. In summary, we define the vectors $\bar{V}_j := e_{I_j} \bar{X}_{I_j}$, and note that hallucinating policy has value distributed according to $\|\bar{V}_1 + \ldots + \bar{V}_m\|$. Thus, our goal for the remainder of the section is to prove that
    \begin{align}
        \underbrace{\E \|\bar{V}_1 + \ldots + \bar{V}_m\|}_{\textrm{hallucination}} \ge \frac{1}{O(p)} \underbrace{\E \|V_1 + \ldots + V_m\|}_{\adapt}. \label{eq:stochDecoupling}
    \end{align}
    
    To simplify the notation, we use $U_t := V_1 + \ldots + V_t$ and $\bar{U}_t := \bar{V}_1 + \ldots + \bar{V}_t$. As mentioned, to prove \eqref{eq:stochDecoupling} we replace one-by-one the terms of the sum $V_1 + \ldots + V_m$ by the terms of the sum $\bar{V}_1 + \ldots + \bar{V}_m$ and track the change in $\E \|\cdot\|^p$. However, we will also need an additional truncation to be able to move from $\E \|\cdot\|^p$ to $\E \|\cdot\|$. For that, let $\tau$ be the stopping time defined as the first $t$ such that $\|\bar{V}_1 + \ldots + \bar{V}_t\| > \frac{\adapt}{4 cp}$ (or $\tau = m$ if no such $t$ exists), where we set in hindsight the constant $c = \frac{3}{(2-e^{1/2})^{1/p}}$.

    We now perform the replacement of the terms. By tangency, conditioned on $\F_{t-1}$, the random variable $\|V_1 + \ldots + V_{t-1} + V_t\|^p$ has the same distribution as $\|V_1 + \ldots + V_{t-1} + \bar{V}_t\|^p$. Since the event $\tau \ge t$ (i.e., up to time $t-1$, the sum $\|\bar{V}_1 + \ldots + \bar{V}_{t-1}\|$ has not reached above $\lambda$) only depends on the history up to time $t-1$, we have
    \begin{align*}
        \E_{t-1} \bigg[\ones(\tau \ge t) \cdot \bigg( \|U_t\|^p - \|U_{t-1}\|^p \bigg)\bigg] &= \ones(\tau \ge t)  \cdot \E_{t-1} \bigg( \|U_t\|^p - \|U_{t-1}\|^p\bigg) \\
        &= \E_{t-1} \bigg[\ones(\tau \ge t) \cdot \bigg( \|U_{t-1} + \bar{V}_t\|^p - \|U_{t-1}\|^p \bigg)\bigg].
    \end{align*}
    Then taking expectations and adding over all times $t$, we get
    \begin{align}
        \E \|U_\tau\|^p = \E \sum_{t \le \tau} \bigg( \|U_t\|^p - \|U_{t-1}\|^p \bigg) \le \E \sum_{t \le \tau} \bigg( \|U_{t-1} + \bar{V}_t\|^p - \|U_{t-1}\|^p \bigg). \label{eq:partial}
    \end{align}
    We can now upper bound the right-hand side using the $p$-\monotonicity of $\|\cdot\|$, using the same steps employed in the Load-Balancing problem in \Cref{thm:loadBalancing}: for every scenario,
    \begin{align*}
        \sum_{t \le \tau} \bigg( \|U_{t-1} + \bar{V}_t\|^p - \|U_{t-1}\|^p \bigg) \le \sum_{t \le \tau} \bigg( \|U_\tau + \bar{U}_{t-1} + \bar{V}_t\|^p - \|U_\tau + \bar{U}_{t-1}\|^p \bigg) = \|U_\tau + \bar{U}_\tau\|^p - \|U_\tau\|^p. 
    \end{align*}
    Plugging this into \eqref{eq:partial} and using the fact $(a+b)^p \le e^{1/2}\, a^p + (3p)^p\, b^p,$ for all $a,b \ge 0$, which can be checked by considering the cases $a \ge 2 pb$ and $a < 2 pb$, we get
    \begin{align*}
        \E \|U_\tau\|^p \,\le\, \E \|U_\tau + \bar{U}_\tau\|^p - \E \|U_\tau\|^p \,\le\, e^{1/2} \E \|U_\tau\|^p + (3p)^p \cdot \E \|\bar{U}_\tau\|^p - \E \|U_\tau\|^p.
    \end{align*}
    Rearranging and calling the constant $c := \frac{3}{(2 - e^{1/2})^{1/p}}$, gives the upper bound 
    \begin{align}
        \E \|U_\tau\|^p \le  (c p)^p \cdot \E \|\bar{U}_\tau\|^p. \label{eq:stochMid}
    \end{align}    

    By the monotonicity of the norm, this implies that $\E \|\bar{U}_\tau\|^p \ge \frac{1}{O(p)^p} \cdot \E \|U_\tau\|^p$, which ``morally'' says that the non-adaptive policy $\bar{U}_\tau$ gets at least a $\frac{1}{O(p)}$-fraction of the value of the optimal adaptive policy $\adapt$ (regarding the presence stopping time $\tau$, notice that in the scenarios where it kicks in, i.e. $\tau < m$, then by definition $\bar{U}_\tau$ has value at least $\frac{\adapt}{4cp}$). To make this precise, we show the following interpolation result that converts the $\ell_p$-type inequality \eqref{eq:stochMid} (plus the boundedness of $\bar{U}_\tau$ guaranteed by the stopping time $\tau$) into a weak-(1,1)-type inequality, which is inspired by a similar inequality for martingales from Burkholder~\cite{burkholder}. 

    \begin{lemma} \label{lemma:weakIneq}
        %Let $V_1,\ldots,V_m$ and $\bar{V}_1,\ldots,\bar{V}_m$ be tangent sequences adapted to a filtration $\F_1,\ldots,\F_m$ such that $\|V_j\|$ and $\|\bar{V}_t\|$ are all at most $M$. Suppose the $V_t$'s and $\bar{V}_t$'s are non-negative and $\|\cdot\|$ is a monotone norm. Suppose that for some $p \ge 1$, for every stopping time w.r.t. $\{\F_t\}_t$ we have the moment comparison
        %
        %\begin{align*}
        %\E \bigg\|\sum_{t \le \tau} V_t\bigg\| ^p \le \E \bigg\| \sum_{t \le \tau} \bar{V}_t \bigg\|^p.
        %\end{align*}
        %Then for every $\lambda \ge p M$, we have
        %
        \begin{align*}
            \Pr\bigg(\,\bigg\|\frac{U_m}{cp}\bigg\| \ge \frac{\adapt}{2cp}\,\bigg) ~\le~ O(1) \cdot \frac{\E \|\bar{U}_m\|}{\adapt/2cp}.  
        \end{align*}
    \end{lemma}

    \begin{proof}
        Let $\lambda := \frac{\adapt}{4cp}$ (the threshold for the stopping time $\tau$) to simplify the notation. Observe that the event ``$\|\frac{U_m}{cp}\| \ge 2\lambda$ and $\|\bar{U}_m\| \le \lambda$'' is contained in the event ``$\|\frac{U_\tau}{cp}\| \ge 2\lambda$'': any scenario that belongs to the first event needs to have $\tau = m$ (since $\tau < m$ implies that $ \|\bar{U}_m\| > \lambda$, by the monotonicity of the norm), and so it is clear that such scenario also belongs to the second event. Thus, using Markov's inequality and then the moment comparison \eqref{eq:stochMid}, we get
        \begin{align*}
            \Pr\bigg(\,\bigg\|\frac{U_m}{cp}\bigg\| \ge 2\lambda\,,\, \|\bar{U}_m\| \le \lambda\bigg) \le \Pr\bigg(\,\bigg\|\frac{U_\tau}{cp}\bigg\| \le 2\lambda \bigg) = \Pr\bigg(\,\bigg\|\frac{U_{\tau}}{cp}\bigg\|^p \le (2\lambda)^p \bigg) \le \frac{\E\, \|\frac{U_{\tau}}{cp}\|^p}{(2\lambda)^p} \le \frac{\E\, \|\bar{U}_{\tau}\|^p}{(2\lambda)^p}.
        \end{align*}
        To upper bound the right-hand side, by the definition of the stopping time $\tau$ and the fact that  by hypothesis the increments satisfy $\|\bar{V}_t\| \le \frac{\adapt}{4cp} = \lambda$, we have $\|\bar{U}_\tau\| \le 2 \lambda$. Plugging this in the displayed inequality, we get 
        \begin{align*}
            \Pr\bigg(\bigg\|\frac{U_m}{cp}\bigg\| \ge 2\lambda\,,\, \|\bar{U}_m\| \le \lambda \bigg)  \le \frac{\E \|\bar{U}_{\tau}\|}{2\lambda} \cdot \frac{(2\lambda)^{p-1}}{(2\lambda)^{p-1}}  \le O(1) \cdot \frac{\E \|\bar{U}_m\|}{\lambda} ,
        \end{align*}
        where the last inequality also uses the monotonicity of the norm. Moreover, also by Markov's inequality we have $\Pr(\|\bar{U}_m\| > \lambda) \le \frac{\E \|\bar{U}_m\|}{\lambda}$. Thus,
        \begin{align*}
            \Pr\bigg(\,\bigg\|\frac{U_m}{cp}\bigg\| \ge 2\lambda \bigg) \le  O(1) \cdot \frac{\E \|\bar{U}_m\|}{\lambda} + \frac{\E \|\bar{U}_m\|}{\lambda},
        \end{align*}        
        which proves the lemma. 
    \end{proof}

    Since $\E \|U_m\| = \adapt$ and $\|U_m\| \le 12\, \adapt$, in addition to the above upper bound we have the lower bound $\Pr(\|U_m\| \ge \frac{\adapt}{2}) \ge \frac{1}{23}$ (e.g., by applying Markov's inequality to $12\, \adapt - \|U_m\|$). Combining this with the previous lemma, it gives
    \begin{align*}
        \E \|\bar{U}_m\| \ge \frac{\adapt}{O(p)} \cdot \Pr\bigg(\bigg\| \frac{U_m}{cp} \bigg\| \ge \frac{\adapt}{2cp} \bigg) \ge \frac{\adapt}{O(p)} \cdot \frac{1}{23} = \frac{\adapt}{O(p)}. 
    \end{align*}
    This proves \eqref{eq:stochDecoupling}, which then gives \Cref{thm:stochProbing}.

    \begin{obs}
    We note that the proof only relied on the \emph{tangency} of the sequences $V_1,\ldots,V_m$ and $\bar{V}_1, \ldots, \bar{V}_m$. Recall that two sequences of random variables $V_1,\ldots,V_m$ and $\bar{V}_1,\ldots,\bar{V}_m$ adapted to a filtration $\F_1,\ldots,\F_m$ are \emph{tangent} if, for all $t$, conditioned on $\F_{t-1}$ the random variables $V_t$ and $\bar{V}_t$ have the same distribution (see \Cref{sec:introProbe} for their applications). The above argument gives the following comparison of averages between Banach-valued tangent sequences.
    
    \begin{thm} \label{thm:decoupling}
        Let $V_1,\ldots,V_m$ and $\bar{V}_1,\ldots,\bar{V}_m$ be tangent sequences taking values in $\R^d_+$. If $\|\cdot\|$ is a $p$-\monotone norm, then $\E\|V_1 + \ldots + V_m\| \le O(p) \cdot \E \|\bar{V}_1 + \ldots + \bar{V}_m\|$.
    \end{thm}

    \noindent This complements the (stronger) results known for the so-called UMD Banach spaces.
    \end{obs}

%%%%%%%%%%%%%%%%%%%%%%%%%%%%%%%%%%%%%%

\section{Applications via Gradient Stability} \label{sec:gradStability}

    The notion of gradient-stable approximation of norms was introduced in \cite{KMS-SODA23} to handle  problems like online load balancing (as in \Cref{sec:intoSuper}) and Bandits with Knapsacks with general norms. 
    We show that if the norm is $p$-\monotone, then it admits a good gradient-stable approximation; essentially this means that $p$-\monotonicity is a stronger property than gradient stability. 
    
    \subsection{Relation to gradient stability}

First, recall the definition of a gradient-stable approximation of a norm. 
    
\begin{definition}[Gradient-Stable Approximation \cite{KMS-SODA23}] \label{def:gradStable}
	We say that a norm $\|\cdot\|$ admits a $\delta$-gradient-stable approximation with error $(\alpha, \gamma$) 
 %\footnote{We say $\alpha$ is the multiplicative error and $\gamma$ is the additive error.}
 if for every $\e > 0$ there is a monotone, subadditive, convex function $\Psi_\e : \R^d_+ \rightarrow \R$ such that:
     \begin{enumerate} [topsep=6pt,itemsep=1pt]
     		\item \emph{Gradient Stability:} $\nabla \Psi_\e(x + y) \geq \exp(-\epsilon \cdot\|y\| -\delta)\cdot \nabla \Psi_\e(x)$ coordinate-wise for all  $x, y \in  \R^d_+$.
     		\item \emph{Norm Approximation:} $\|x\|\leq \Psi_\e(x) \leq \alpha \|x\| + \frac{\gamma}{\e}$ for all  $x \in  \R^d_+$.
	\end{enumerate}
\end{definition}

The key insight from \cite{KMS-SODA23} is that if a norm admits a $\delta$-gradient-stable approximation with error $(\alpha,\gamma)$ for $\delta \leq \frac{1}{4}$ then it can be used to construct $O(\alpha+\gamma)$-competitive algorithms for multiple problems. We can show that such approximations exist for all $p$-\monotone norms with $\alpha + \gamma = p$.

    \begin{lemma}
    \label{lemma:gradientstable}
        Every differentiable $p$-\monotone norm $\|\cdot\|$ admits a $0$-gradient stable approximation with error $(1, p-1)$.  
    \end{lemma}

    \begin{proof}
        We claim that the function $\Psi_\e(x) = \max\{\frac{p-1}{\e}, \|x\|\}$ is the desired gradient-stable approximation of $\|\cdot\|$. The desired $(1,p-1)$-approximation property $\|x\| \le \psi_\e(x) \le \|x\| + \frac{p-1}{\e}$ follows directly from the definition. 

        For the gradient stability, consider $x,y \in \R^d_+$ and assume $\|x\| > \frac{p-1}{\e}$, else $\nabla \Psi_\epsilon(x) = 0$, so the claim follows. If $\|x\| > \frac{p-1}{\e}$, then $\Psi_\e(x + y) = \|x+y\|$ and $\Psi_\e(x) = \|x\|$, and 
        \begin{align*}
            \nabla \Psi_\e(x+y) = \nabla \|x+y\| &= \frac{\nabla (\|x+y\|^p)}{p \|x+y\|^{p-1}} \\
            &\ge \frac{\nabla (\|x\|^p)}{p \|x+y\|^{p-1}} = \frac{\nabla (\|x\|^p)}{p \|x\|^{p-1}} \cdot \frac{\|x\|^{p-1}}{\|x+y\|^{p-1}} = \Psi_\e(x) \cdot \frac{\|x\|^{p-1}}{\|x+y\|^{p-1}}, 
        \end{align*}
        where the inequality uses the $p$-\monotonicity of $\|\cdot\|$. To bound the last term, by triangle inequality
        \begin{align*}
            \frac{\|x+y\|^{p-1}}{\|x\|^{p-1}} \le  \frac{(\|x\| + \|y\|)^{p-1}}{\|x\|^{p-1}} = \bigg(1 + \frac{\|y\|}{\|x\|}\bigg)^{p-1} \le \bigg(1 + \frac{\|y\|}{(p-1)/\e}\bigg)^{p-1} \le e^{-\e \|y\|},   
        \end{align*}
        where the second inequality uses the fact that $\|x\| > \frac{p-1}{\e}$. Plugging this on the previous displayed inequality we have the gradient-stability $\nabla \Psi_\e(x+y) \ge e^{-\e \|y\|}\cdot \nabla \Psi_\e(x)$, as desired. This concludes the proof of the lemma.
    \end{proof}

    So, in particular, from \Cref{thm:orlicz} every Orlicz norm can be $O(1)$-approximated by an $O(\log n)$-\monotone norm. Therefore, every Orlicz norm admits a $0$-gradient-stable approximation with error $(1, O(\log n))$. This improves over the bound in \cite{KMS-SODA23}, which only gave a guarantee of $(O(\log n), O(\log^2 n))$.

\subsection{Applications}

    We can use \Cref{lemma:gradientstable} to get algorithms for all applications considered in \cite{KMS-SODA23}, where $\alpha = 1$ and $\gamma = p-1$. In particular, this yields a $O(p)$-competitive algorithm for online load balancing. Note that this mirrors the bound we obtained in \Cref{sec:intoSuper} in a more direct way.
    
    There are two more applications in \cite{KMS-SODA23} for bandit problems, which can also be combined with \Cref{lemma:gradientstable}. 
    For both these applications, 
    the following results  improve the approximation factors for Orlicz norms in $n$ dimensions from $O(\log^2 n)$ in \cite{KMS-SODA23}   to $O(\log n)$ via $p$-\monotonicity.
    
    The first one is \emph{Bandits with Knapsacks} \cite{ISSS-JACM22} (for the problem definition, see \cite{KMS-SODA23}).
    \begin{cor}
    Consider the \emph{Bandits with Knapsacks problem} for adversarial arrivals with $k$ actions and a $p$-\monotone norm $\| \cdot \|$. Let $B \geq 4 \cdot p \cdot \| \ones \|$. Then there exists an algorithm that takes $\mathrm{OPT}_{\textsc{BwK}}$ as its input and obtains reward at least 
    \[
    \Omega\left(\frac{1}{p} \mathrm{OPT}_{\textsc{BwK}}\right) - O\left(\frac{\mathrm{OPT}_{\textsc{BwK}} \cdot \| \ones \|}{p \cdot B}\right) \cdot \textsc{Regret}
    \]
    with probability $1 - q$, where $\textsc{Regret} = O( T k \log(k/q))$ and $q \in [0, 1]$ is a parameter. Moreover, this algorithm is efficient given gradient oracle access to the norm.
    \end{cor}

    The second one is \emph{Bandits with Vector Costs} \cite{KS-COLT20} (again, for the problem definition, see \cite{KMS-SODA23}).
    
    \begin{cor}
    Consider the problem \emph{Bandits with Vector Costs} with $k$ actions and a $p$-\monotone norm $\| \cdot \|$. There exists an algorithm that guarantees
    \[
    \left\| \sum_{t=1}^T C^{(t)} \cdot x^{(t)} \right\| = O(p) \cdot \left\| \sum_{t=1}^T C^{(t)} \cdot x^\ast \right\| + \| \ones \| \cdot \textsc{Regret}
    \]
    with probability $1-q$, where $\textsc{Regret} = O(\sqrt{T k \log(k/q)})$ and $q \in [0,1]$ is a parameter.
    \end{cor}

    Note that the dependencies on $p$ are essentially tight because $\ell_p$-norms are $p$-\monotone and there are impossibility results for $\ell_p$-norms given in \cite{KS-COLT20}. 

%%%%%%%%%%%%%%%%%%%%%%%%%%%%%%%%%%%%%%

%\section{Applications to Online Linear/Convex Optimization?}

\pagebreak

\appendix

\noindent {\LARGE \bf Appendix}

\section{Applications of Covering with Composition of Norms} \label{app:appComp}

To illustrate the scope of applications for the problem of Covering with Composition of Norms, we illustrate how it can model Online fractional Facility Location problem and fractional version of the Generalized Load-Balancing problem of \cite{DLR-SODA23}.

\paragraph{Online fractional Facility Location.} In this problem, there are multiple facilities $j = 1,\ldots,m$, each with an opening cost $c_j$, and multiple demand points $i = 1,\ldots,n$ with an associated connected cost $d_{ij}$ to connect to facility $j$ (note we do not require that the connection costs come from a metric space). The goal is to open a set of facilities and connect each demand to one facility in a way that minimizes the total opening and connection costs. This can be modeled by the convex program 
\begin{align*}
    \min ~& \sum_j c_j \cdot \max_i y_{ij} + \sum_{i,j} d_{ij} y_{ij}\\
    \textrm{s.t.} ~& \sum_j y_{ij} \ge 1, ~\forall i\\
    & y_{ij} \in \{0,1\}, ~\forall i,j,
\end{align*}
where $y_{ij}$ indicates whether demand $i$ connected to facility $j$. (In the fractional version of the problem, the variables $y_{ij}$ are allowed to take value in $[0,1]$.) This is a special case of Covering with Composition of Norms: the constraints are precisely of covering type, and the objective function can be expressed as the composed norm $\|(f_1(y|_{S_1}),\ldots, f_m(y|_{S_m}), f_{11}(y|_{S_{11}}), \ldots, f_{nm}(y|_{S_{nm}}))\|_1$, where for each $j=1,\ldots,m$, $f_j(x) = c_j \cdot \|x\|_{\infty}$ and $S_j = \{(1,j), (2,j), \ldots, (n,j)\}$, and for each $i = 1,\ldots,n$ and $j = 1,\ldots,m$ we have $f_{ij}(x) = d_{ij} \cdot \|x\|_{\infty}$ and $S_{ij} = \{(i,j)\}$. 

In \emph{online} (fractional) Facility Location, the demands $i = 1,\ldots,n$ come one by one, and when a demand arrives it is revealed its connection costs $d_{i1}, \ldots, d_{im}$ (the opening costs are known upfront); thus, part of the objective function is revealed online. As defined, in \cover the whole objective function is available to the algorithm, and thus, it does not formally capture online (fractional) Facility Location. However, we remark that our algorithm for \cover from \Cref{thm:cover} only require the current gradient of the objective function, which can be computed based on the online arrival of the demands, since it always maintains at value 0 the variables $y_{ij}$ of unseen demands $i$. Thus, \Cref{thm:cover} can indeed be used to solve online fractional Facility Location in non-metric spaces. This leads to a guarantee of $O(\log n \cdot \log^2 \max\{n,m\})$ for this problem (since the $\ell_1$ norm is $1$-\monotone and \Cref{thm:orlicz} and approximating each inner $\ell_\infty$ by an $O(\log n)$-\monotone norm). This can be compared to the $O(\log m \cdot (\log n + \log \log m))$ approximation for the (harder) integral fraction version of the problem, but using a specialized algorithm~\cite{nonMetricFL}. 
%{\color{red}We should also talk how this captures about online facility location in \cite{PRS-APPROX23}, i.e., when we don't put $\ell_1$ but more general norm on the distances.}

%{\color{red} Issue: Actually need to know the whole obj function to do the appropriate scaling}

\paragraph{Generalized Load-Balancing problem of \cite{DLR-SODA23}.} In this problem, there are machines $i=1,\ldots,m$ and jobs $j=1,\ldots,n$, with $p_{ij} > 0$ denoting the processing time of job $i$ on machine $j$. Each job needs to be assigned to one machine; borrowing the notation from the previous part, we use $y_{ij} \in \{0,1\}$ to indicate whether job $j$ is assigned to machine $i$. Note $\sum_j y_{ij} \ge 1$ for all jobs $j$. Each machine $i$ has an inner norm $\phi_i$ on $\R^n$, and the load on this machine depends on the jobs assigned to it and is given by $load_i(y) := \phi_i(p_{i1} y_{i1}, \ldots, p_{in} y_{in})$. There is also an outer norm $\|\cdot\|$ on $\R^m$ used to aggregate the loads of the machines, so the total cost of the assignment $y$ is given by $\|(load_1(y), \ldots, load_m(y)\|$. The goal is to find the assignment with smallest total cost. 

This problem is also a special case of Covering with Composition of Norms: the constraints $\sum_j y_{ij} \ge 1$ are precisely of covering type, and the objective function can be expressed as the composed norm $\|(f_1(y|_{S_1}),\ldots, f_m(y|_{S_m})\|$, where for each machine $i\in [m]$, $f_i(x) = \phi_i(p_{i1} x_1, \ldots, p_{in} x_n)$ and $S_i = \{(i,1), (i,2), \ldots, (i,n)\}$. 
%{\color{red}Dont we need care since this is no longer symmetric? Do we want to apply it on the vector $px$?}

In the \emph{online} version of the problem, the jobs arrive one by one, and their processing times are revealed upon arrival. As in the case of Facility Location above, while parts of the objective function (namely the processing times $p_{ij}$) are revealed over time, and thus do not conform exactly to \cover. But again, it can be verified that our algorithm from \Cref{thm:cover} can still be used to obtain a competitive fractional solution. When the norms $\phi_1,\ldots,\phi_m$ and $\|\cdot\|$ are monotone symmetric, \Cref{thm:cover} and
\Cref{obs:normScaling} imply that
our algorithm obtains a fractional solution that is $O(\log^2 n \cdot \log^2 m \cdot \log^2 (\max\{m,n\} \cdot \gamma))$ competitive, where $\gamma = \max_i \frac{\max_j \phi_i(p_{ij})}{\min_{j : p_{ij} \neq 0} \phi_i(p_{ij})}$. This can be compared against the $O(\log n)$ approximation for the integral (harder) but offline (easier) version of the problem given recently in~\cite{DLR-SODA23}.

\section{Differentiability of Norms} \label{app:propNorms}

\subsection{Smoothing of $p$-\monotone norms} \label{app:smoothing}

    \begin{lemma}  \label{lemma:diff}
        For every $\e > 0$, every $p$-\monotone norm $\|\cdot\|$ can be $(1+\e)$-approximated by a $p$-\monotone norm $\vertm{\cdot}$ (i.e. $\|x\| \le \vertm{x} \le (1+\e) \|x\|$ for all $x \in \R^d_+$) that is infinitely differentiable everywhere except at the origin. 
    \end{lemma}

    \begin{proof}
     Let $R_1,\ldots,R_d$'s are independent random variables in $[1,1+\e]$ that have pdf $\phi : \R \rightarrow \R$ of class $C^{\infty}$ (infinitely differentiable). The norm $\vertm{x} := \E \|(R_1 x_1, R_2 x_2, \ldots, R_d x_d)\|$ has the desired properties. Clearly for every non-negative vector $x$, $\|x\| \le \vertm{x} \le (1+\e) \|x\|$. Standard arguments show that $\vertm{\cdot}$ is $C^{\infty}$ with the exception of the origin~\cite[Section 3.4]{schneider}. Finally, $\vertm{\cdot}$ is $p$-\monotone, since for each scenario $x \mapsto \|(R_1 x_1, R_2 x_2, \ldots, R_d x_d)\|$ is $p$-\monotone and this is property preserved by taking averages.
    \end{proof}

%########################################################################
%########################################################################

\subsection{Properties of the gradient}

We collect standard properties of general norms, in particular in relation to their gradients and duals. They are all consequences of involution of duality, i.e., $(\|\cdot\|_{\star})_{\star} = \|\cdot\|$ and compactness of norm balls.

\begin{lemma} \label{lemma:gradNorms}
    Every differentiable norm $\|\cdot\|$ in $\R^d$ satisfies the following for every $x \in \R^d \setminus \{0\}$:

    \begin{enumerate}[itemsep=0pt,topsep=0pt]
        \item $\nabla \|x\| = \argmax_{y : \|y\|_{\star} \le 1} \ip{x}{y}$
        
        \item $\|\,\nabla \|x\|\,\|_{\star} = 1$

        \item $\|x\| = \ip{\nabla \|x\|}{x}$

        \item $\nabla \|x\|$ invariant to positive scaling, i.e., $\nabla \|\alpha x\| = \nabla \|x\|$ for all $\alpha > 0$.  
    \end{enumerate}
\end{lemma}

%%%%%%%%%%%%%%%%%%%%%%%%%%%%%%%%%%%%%%%
\section{Missing Proofs}

\subsection{Proof of \Cref{lem:counterEgGeneralNorm}} \label{sec:missingcounterEgGeneralNorm}
We consider a norm obtained by summing $\ell_\infty$ norms on disjoint coordinates. Formally, we partition the $n$ coordinates into $\sqrt{n}$ blocks  of $\sqrt{n}$ coordinates each. 
We use the notation $m:=\sqrt{n}$ and $B_k = \{(i, k) \mid i \in \{1, \ldots, m\}\}$ for $k \in [m]$. We   think of these blocks as columns of a matrix and the sets $(i, 1), \ldots, (i, m)$ for $i \in [m]$ as rows of the matrix.
Now our norm 
\[ \textstyle \|x\| := \sum_{j \in [m]} \|x_{B_k}\|_\infty,\] 
where  $x_{B_k}$ is the $\sqrt{n}$-dimensional vector obtained by taking the coordinates of $x$ in $B_k$. % doesn't have a polylog-approximation that is polylog-supermodular.

Consider any $\alpha$-approximating function $f: \R^n_+ \rightarrow \R_+$ such that $f$ is subadditive and $f^p$ is supermodular for some $p \geq 1$. Reorder the columns of the matrix such that 
\[
f(D_i) \leq f(D_{i-1} \cup \{(i, k)\}) \text{ for all $k \geq i$,}
\]
where $D_i = \{(1,1), \ldots, (i, i)\}$. Then, for $S_i = \{(1,m),\ldots,(i,m)\}$, we have
\begin{align*}
f^p(D_{m}) ~=~ \sum_{i = 1}^{m} \big( f^p(D_i) - f^p(D_{i-1}) \big) ~&\leq~ \sum_{i = 1}^{m} \big(  f^p(D_{i-1} \cup \{(i, m)\}) - f^p(D_{i-1}) \big) \\
&\leq~ \sum_{i = 1}^{m} \big( f^p(D_{m} \cup S_i) - f^p(D_{m} \cup S_{i-1}) \big)\\
&=~ f^p(D_{m} \cup S_{m}) - f^p(D_{m}).
\end{align*}
So, $
2^{1/p} f(D_m) \leq f(D_m \cup S_m) \leq f(D_m) + f(S_m),$ 
which implies
\[
f(D_m) \leq \frac{1}{2^{1/p} - 1} f(S_m) ~\leq~ \frac{p}{\ln 2} f(S_m).
\]
For $f$ to be an $\alpha$-approximation of the norm, we need $f(D_m) \geq m$ and $f(S_m) \leq \alpha$. Therefore $\alpha p \geq (\ln 2) m$.

\subsection{Proof of \Cref{lem:SymmetricToOrlicz}} \label{sec:missingProofs}

%\begin{proof}[Proof of \Cref{lem:SymmetricToOrlicz}]
    Let $c>1$ be some fixed constant and $n$ be a power of $2$.
    Consider the norm $\|\cdot\|$ whose unit ball is given by $\log n + 1$ constraints: 
    \begin{align} \label{eq:OrliczToSymm}
        \sum_{i=1}^{2^j} (x^{\downarrow})_i \leq c^j   \qquad \text{for $j \in \{0,1,\ldots,\log n\}$}.
    \end{align} 
     Now consider the vector $x$ where each of these $\log n+1$ constraints are tight, i.e, starting with $x_1 = 1$, we have for $i \in [2^{j},2^{j+1})$ that      
     $x_i = (c^{j+1} - c^{j})/2^{j} = (\tfrac{c}{2})^{j} (c-1)$. 
    Thus, $\|x\| = 1$. 

    Suppose the norm $\|\cdot\|$ can be approximated by an Orlicz norm $\|\cdot\|_G$ within some factor $\alpha \geq 1$, i.e., $\|u\| \leq \|u\|_G \leq \alpha \|u\|$ for all $u \in \R_+^n$. This implies $\sum_i G(x_i) \leq 1$, or in other words
    \[
        \textstyle G(1)+\sum_{j=1}^{\log n} 2^{j}\cdot G\Big( (\tfrac{c}{2})^{j} (c-1) \Big) \leq 1.
    \]
    This implies that there exists a $k \in \{0,1,\ldots,\log n\}$ such that $G\Big( (\tfrac{c}{2})^{k} (c-1) \Big) \leq \frac{1}{2^k (1+\log n)}$. 

    Now define vector a $y$ such that $y_i = (\tfrac{c}{2})^{k} (c-1)$ for $i \leq 2^k (1+\log n)$ and $y_i =0 $ otherwise. We first observe that $\|y\|_G \leq 1$ since $\sum_i G(y_i) = 1$. Next, we will show that $\|y\|$ is much larger than $1$, and hence $\alpha$ needs to be large. 
    
    To calculate $\|y\|$, consider the constraint given by \eqref{eq:OrliczToSymm} for $2^j =  2^k (1+\log n)$.  For this constraint to be feasible, up to the approximation factor $\alpha$, we need
    \[
        (\tfrac{c}{2})^{k} (c-1) \cdot 2^j \leq \alpha \cdot c^j \qquad \Longleftrightarrow \qquad \tfrac{c-1}{c^{j-k}} \cdot (1+\log n) \leq \alpha .
    \]
    Since $j \geq  k + \log\!\log n$, we get 
    \[
        \alpha  \geq \tfrac{c-1}{c^{\log\!\log n}} \cdot (1+\log n) \geq (c-1)\cdot (\log n)^{1-\log c}.
    \]
    Taking $c= 1+\epsilon$ for some small constant $\epsilon$ implies $\alpha \geq \epsilon \cdot (\log n)^{1-\Theta(\epsilon)}$, which completes the proof.
%\end{proof}

%#############################################################
%#############################################################
%#############################################################
%#############################################################

    \subsection{Proof of \Cref{thm:cover}: Discharging the Assumptions} \label{app:cover}

    We use essentially the construction from \cite{NS-ICALP17} to convert, in an online fashion, any instance of \cover into an equivalent one satisfying \Cref{ass:cover} stated in \Cref{sec:cover}.
    
    More precisely, consider an instance $\mathcal{I}$ of \cover with objective function given by the outer norm $\|\cdot\|$ and inner norms $f_1,\ldots,f_k$, restriction sets $S_1, \ldots, S_k \subseteq [n]$, and online constraints $\ip{A_1}{y} \ge 1, \ip{A_2}{y} \ge 1, \ldots$. We then construct the instance $\bar{\mathcal{I}}$ with modified ground set $\bar{U}$, inner norms $\bar{f}_1,\ldots,\bar{f}_k$, partition sets $\bar{S}_1, \ldots, \bar{S}_k \subseteq \bar{U}$, and online constraints $\ip{\bar{A}_1}{y} \ge 1, \ip{\bar{A}_2}{y} \ge 1, \ldots$ (the outer norm remaining the same) that has the desired property:
    \begin{align*}
        \textrm{The restricting sets $\bar{S}_1,\ldots,\bar{S}_k$ partitions the variable set $\bar{U}$.}
    \end{align*}

    \paragraph{Construction of the instance $\bar{\mathcal{I}}$.} First, we duplicate each variable $y_i$ into a copy $\bar{y}_{i,\ell}$ for each set $S_\ell$ containing $i$. More precisely, define the set of variables $\bar{U}$ to consist of all pairs $(i,\ell)$ with $i \in S_\ell$, and for each $(i, \ell) \in \bar{U}$ introduce the variable $\bar{y}_{i,\ell}$. Define $\bar{S}_{\ell}$ be the set of pairs $(i,\ell)$ (ranging over $i$) in $\bar{U}$, i.e., this is the ``lifting'' $\bar{S}_\ell = \{(i, \ell) : i \in S_\ell\}$ of the set $S_\ell$. 

    Then define the modified inner norms $\bar{f}_\ell : \R^{\bar{S}_\ell} \rightarrow \R$ in the natural way: $$\bar{f}_\ell(\bar{y}|_{\bar{S}_\ell}) := f_\ell((\bar{y}_{i,\ell})_{i \in S_\ell}),~~\forall \bar{y}.$$ Finally, we add new constraints to handle the multiple copies $(i,\ell)$ of the same original coordinate $i$. More precisely, for each constraint $\ip{A_r}{y} \ge 1$, we define the modified constraints $\ip{\bar{A}_r^{\pi}}{\bar{y}} \ge 1$ indexed by all possible ``copy selector'' functions $\pi$ that map $i \mapsto \pi(i)$ so that $i \in S_{\pi(i)}$ as follows: the vector $\bar{A}^\pi_r \in \R^{\bar{U}}_+$ has coordinate $(i,\ell)$ given by $$(\bar{A}^\pi_r)_{(i,\ell)} := \left\{\begin{array}{ll} (A_r)_i, & \textrm{if $\ell = \pi(i)$} \\
        0, & \textrm{else}\end{array} \right.$$

     During the online presentation of the instance, when the  constraint $\ip{A_r}{y} \ge 1$ comes, we present the constraints $\ip{\bar{A}_r^{\pi}}{\bar{y}} \ge 1$, ranging over all copy selectors $\pi$, in any order. This concludes the definition of the instance $\bar{\mathcal{I}}$.

    \paragraph{Properties of $\bar{\mathcal{I}}$.} By definition of the sets $\bar{S}_1,\ldots,\bar{S}_k$, they partition the variable set $\bar{U}$, as desired. 
    
    Moreover, instances $\mathcal{I}$ and $\bar{\mathcal{I}}$ have equivalent solutions. %Let $\cost_{\mathcal{I}}(\cdot)$ $\cost_{\bar{\mathcal{I}}}(\cdot)$ be the cost functions of each of the instances. 
    Given a feasible solution $y$ for the original instance $\mathcal{I}$, then consider the solution given by $\tilde{y}_{i,\ell} := y_i$  for all $\ell$ and $i \in S_\ell$. It is clear that both solutions have the same value on their respective instances. Moreover, $\tilde{y}$ is feasible for $\bar{\mathcal{I}}$, since for every constraint $\bar{A}^{\pi}_r$
    \begin{align*}
        \ip{\bar{A}^{\pi}_r}{\tilde{y}} = \sum_i \sum_{\ell} (\bar{A}^{\pi}_r)_{(i,\ell)} \cdot \tilde{y}_{i,\ell} = \sum_i (A_r)_i \cdot y_i = \ip{A_r}{y} \ge 1.
    \end{align*}

    Conversely, given any feasible solution $\tilde{y}$ for  $\bar{\mathcal{I}}$, the solution $y_i := \min_{\ell : i \in S_\ell} \tilde{y}_{i,\ell}$ is: 1) Feasible for the original instance $\mathcal{I}$: using the copy selector $\pi(i) := \argmin_{\ell : i \in S_\ell} \tilde{y}_{i,\ell}$, we get
    \begin{align*}
    \ip{A_r}{y} = \sum_i (A_r)_i \cdot y_i = \sum_i (A_r)_i \cdot \tilde{y}_{i,\pi(i)} = \sum_{i,\ell} (\bar{A}_r^{\pi})_{(i,\ell)} \cdot \tilde{y}_{i,\ell} = \ip{\bar{A}_r^{\pi}}{\tilde{y}} \ge 1,
    \end{align*}
    and; 2) The cost of $y$ on the instance $\mathcal{I}$ is at most that of $\tilde{y}$ on the instance $\bar{\mathcal{I}}$, since $y_i \le \tilde{y}_{i,\ell}$ for all $\ell$ such that $i \in S_\ell$, which together with the monotonicity of the norm $f_\ell$ implies
    \begin{align*}
       f_\ell(y|_{S_\ell}) \le f_\ell((\tilde{y}_{i,\ell})_{i \in S_\ell}) = \bar{f}_\ell(\tilde{y}|_{\bar{S}_\ell}), 
    \end{align*}
    and the claim follows from the monotonicity of the outer norm $\|\cdot\|$. 

    These observations show that the optimum of both instances $\mathcal{I}$ and $\bar{\mathcal{I}}$ is the same, and given an $\alpha$-approximate solution $\tilde{y}$ for the latter we can construct (in an online fashion) an $\alpha$-approximate solution $y$ for the original instance $\mathcal{I}$. 

    Finally, we note that in the new instance $\bar{\mathcal{I}}$, the outer and inner norms have the same \monotonicity parameters $p',p$ as in the original instance, and the ``width'' parameters $\rho$ and $\gamma$, as well as the ``sparsity'' parameter $d(\bar{\mathcal{I}}) = \max\{\max_{r,\pi} \supp(\bar{A}_r^{\pi})\,,\, \max_\ell |\bar{S}_\ell|\}$ are the same as in the original instance. 
    
    \medskip
    
    In particular, a $O(p'\, p \log^2 d \rho)$-approximation for the modified instance $\bar{\mathcal{I}}$ can be used to give a solution with the same guarantee for the original instance $\mathcal{I}$. This discharges \Cref{ass:cover}, and concludes the proof of \Cref{thm:cover}.

%############################################################################
%############################################################################
%############################################################################

\subsection{Complete Proof of \Cref{thm:stochProbing}} \label{app:stochProbing} 

In each scenario, let $T = |S^*|$ be the (random) number of probes performed by the optimal solution, and let $I_1,\ldots,I_T$ be the sequence of probes it performs (so $\{I_1,\ldots,I_T\} = S^*$). Since the theorem does not depend on the number of items $n$, we can introduce $n$ additional items (i.e. coordinates) with no value $X_{n+1},\ldots,X_{2n} = 0$ and assume that the number of probes $T$ equals exactly $n$, by padding with additional probes $I_{T+1} = n+1, I_{T+2} = n+2,\ldots, I_n = n + (n - T)$. Let $V_t := e_{I_t} X_{I_t}$. Again let $\bar{X}_1,\ldots,\bar{X}_n$ be an independent copy of $X_1,\ldots,X_n$, and let $\bar{I}_1,\ldots,\bar{I}_{\bar{T}}$ be the probes performed by the optimal adaptive strategy based on the values $\bar{X}_t$'s. Let $\bar{V}_t := e_{I_t} \bar{X}_{I_t}$, and recall that 
\begin{align*}
    \adapt = \E \|V_1 + \ldots V_n\| ~~~~\textrm{ and }~~~~ \textrm{hallucinating} = \E \bigg\|\sum_{t \le n} e_{\bar{I}_t} X_{\bar{I}_t} \bigg\| = \E \|\bar{V}_1 + \ldots + \bar{V}_n\|. 
\end{align*}
Thus, proving  \Cref{thm:stochProbing} is equivalent to proving that $\E\|V_1 + \ldots + V_n\| \le O(p) \, \E \|\bar{V}_1 + \ldots + \bar{V}_n\|$. 

%For convenience, we extend the sequence to have exactly $n$ terms by padding with 0 vectors if necessary, namely let $V_t = \ones(T \ge t) \cdot e_{I_t} X_{I_t}$ where the $\ones(T \ge t)$ ``shortcircuits'' the product (i.e. $T < t$ make the product 0 even if the second term is not defined). Thus, our goal is to prove
%%
%\begin{align*}
%\E \|V_1 + \ldots + V_n\| \le O(p) \, \E \|\bar{V}_1 + \ldots + \bar{V}_n\|$\,.
%\end{align*}

Throughout, we use $(\F_t)_t$ to denote the filtration generated by the sequence $(I_t, X_{I_t},\bar{X}_{I_t})_t$ (or equivalently, by the sequence $(V_t,\bar{V}_t)_t$); so $\F_t$ can be thought as the history up to time $t$. We will use the fact that the problem has the optimal substructure property: Consider a prefix $I_1,I_2,\ldots,I_t$ of the probes; since the feasible set $\mathcal{S}$ is downward closed, the remaining probes $I_{t+1}, I_{t+2},\ldots, I_T$ forms a feasible solution, and so can obtain expected value at most $\adapt$, namely $\E[\|V_{t+1} + V_{t+1} + \ldots + V_n\| \mid \F_t ] \le \adapt$. By allowing the prefix size $t$ to depend on the history up to that point, and using the fact that when $t = n$ there are no remaining probes (so no value), we have the following.

%\mnote{I think we are using the fact the optimal strategy can be chosen to be a deterministic function of the past, i.e. $I_t$ is a function of $X_1,\ldots,X_{t-1}$, but it's ok} 

\begin{obs} \label{obs:optSubstructure}
    For any stopping time $\tau$ for adapted to the filtration $(\F_t)_t$, we have 
    \begin{align*}
        \E \|V_{\tau+1} + V_{\tau+2} + \ldots + V_n\| \,\le\, \adapt \cdot \Pr(\tau < n)\,.
    \end{align*}
\end{obs}

To continue the analysis, we define the ``low'' and ``high'' parts of the $j$-th item as $X^L_j := X_j \cdot \ones(\|e_j X_j\| \le \frac{\adapt}{16cp})$ and $X^H_j := X_j \cdot \ones(\|e_j X_j\| > \frac{\adapt}{16cp})$. Also decompose the $V_t$'s in a similar way, namely define $V^L_t := e_{I_t} X^L_{I_t}$ and  $V^H_t := e_{I_t} X^H_{I_t}$, so $V_t = V^L_t + V^H_t$. We show that the ``hallucinating'' non-adaptive strategy compares favorably against the value that $\adapt$ obtains from the low and high parts of the items, respectively $\|V^L_1 + \ldots + V^L_n\|$ and $\|V^H_1 + \ldots + V^H_n\|$.

%##############################################################################
%##############################################################################

\subsubsection{High part}

    Here we show the following.

    \begin{lemma} \label{lemma:highPart}
    $$\E \|V^H_1 + \ldots + V^H_n\| \le O(p) \, \E \|\bar{V}_1 + \ldots + \bar{V}_n\|.$$ In particular, if the high part obtains value $\E \|V^H_1 + \ldots + V^H_n\| \ge \frac{\adapt}{2}$, then the hallucinating strategy obtains value at least $\frac{\adapt}{O(p)}$. 
    \end{lemma}

    The idea is that since as soon as $V^H_t > 0$ (i.e., the item probed at time $t$ is gets high value) the optimal strategy gets value at least $\approx \frac{\adapt}{p}$, we can stop it at this point and not lose more than a factor of $\approx p$; such truncated strategy that only keeps one item has a much nicer structure that we can use to compare against the decoupled strategy $\{\bar{V}_t\}$.

    To make this formal, we start with the following lemma, which is the special property afforded by ``only probing one item''; this follows from the argument used in~\cite{BSZ-RANDOM19}, or the decoupling inequality \cite[Lemma 2.3.3]{deLaPena}; we provide a self-contained proof nonetheless. 

    \begin{lemma} \label{lemma:tangentMax}
        Let $Y_1,\ldots,Y_n$ and $Y'_1, \ldots,Y'_n$ be sequences of random variables (not necessarily independent) adapted to a filtration $(\G_t)_t$ that are tangent, namely for all $t$, conditioned on $\G_{t-1}$ the distributions of $Y_t$ and $Y'_t$ are the same. Then $$\E \max_t Y_t \le 2 \cdot \E \max_t Y'_t.$$
    \end{lemma}

    \begin{proof}
    Let $E_t := E[\cdot \mid \G_t]$ denote conditional expectation w.r.t. $\G_t$, and let $Z_t := \max\{Y_t, Y'_t\}$. We have
    \begin{align}
    \max\{Y_1,\ldots,Y_n\} \le Z_1 + \sum_{t = 2}^n \Big(\max\{Z_1,\ldots,Z_t\} - \max\{Z_1,\ldots,Z_{t-1}\} \Big). \label{eq:decoupleMax}
    \end{align}
    In addition, by tangency, for every $t \ge 2$ and deterministic value $s$ we have $$\E_{t-1} \max\{s, Z_t\} - s \,\le\, \E \Big(\max\{s, Y_t\} - s\Big) + \E \Big(\max\{s, Y'_t\} - s\Big) = 2 \, \E \Big(\max\{s, Y'_t\} - s\Big).$$ Applying this with $s = \max\{Z_1, \ldots, Z_{t-1}\}$ gives 
    \begin{align*}
    &\E_{t-1}  \bigg(\max\{Z_1,\ldots,Z_t\} - \max\{Z_1,\ldots,Z_{t-1}\} \bigg)\\
    &\le 2 \, \E_{t-1} \Big(\max\{Z_1, \ldots, Z_{t-1}, Y'_t\} - \max\{Z_1, \ldots, Z_{t-1}\}\Big) \\
    &\le 2 \, \E_{t-1} \Big(\max\{Y'_1, \ldots, Y'_{t-1}, Y'_t\} - \max\{Y'_1, \ldots, Y'_{t-1}\}\Big),
    \end{align*}
    where the last inequality follows from the fact $Y'_j \le Z_j$. Taking expectations and applying this to~\eqref{eq:decoupleMax} gives
    \begin{align*}
    \E \max\{Y_1,\ldots,Y_n\} &\textstyle \le \E Z_1 + 2 \sum_{t = 2}^n \Big(\max\{Y'_1, \ldots, Y'_{t-1}, Y'_t\} - \max\{Y'_1, \ldots, Y'_{t-1}\}\Big) \\
    &= \E Z_1 + 2\, \E \max\{Y'_1,\ldots,Y'_n\} - 2\, \E Y'_1 \le 2\, \E \max\{Y'_1,\ldots,Y'_n\},
    \end{align*}    
    where the last inequality uses $\E Z_1 \le \E (Y_1 + Y'_1) = 2\, \E Y'_1$. This concludes the proof.
    \end{proof}

    Now let $\tau$ be the first time when $\|V_1^H + \ldots V_{\tau}^H\| > \frac{\adapt}{16cp}$ ($\tau = n$ if no such time exists), which by the definition of the ``high'' variables $V_t^H$ is equivalent to the first time that one of these variables is different from zero. Consider the stopped sequence $\|V_1^H + \ldots + V_{\tau}^H\| = \|V_{\tau}^H\|$. We argue that it has $\frac{1}{O(p)}$-fraction of the value of the non-stopped sequence; in particular, so does the best high value $\max_t \|V^H_t\|$. More precisely, from triangle inequality we have 
    \begin{align}
        \E \|V^H_1 + \ldots + V^H_n\| \le \E \|V^H_1 + \ldots + V^H_{\tau}\| + \E \|V^H_{\tau+ 1} + \ldots + V^H_n\|. \label{eq:probeAlmost}
    \end{align}
    By definition of the stopping time $\tau$, the first term on the right-hand side is at least $\frac{\adapt}{16cp} \cdot \Pr(\tau < n)$, or equivalently, $\adapt \cdot \Pr(\tau < n) \le 16cp \cdot \E \|V^H_1 + \ldots + V^H_{\tau}\|$. Thus, using \Cref{obs:optSubstructure} we see that the second term of \eqref{eq:probeAlmost} is at most $16cp \cdot \E \|V^H_1 + \ldots + V^H_{\tau}\|$. Employing this on \eqref{eq:probeAlmost}  shows 
    \begin{align}
        \E \|V^H_1 + \ldots + V^H_n\| \le (16cp + 1) \,\E \|V^H_1 + \ldots + V^H_{\tau}\|& = (16cp + 1) \,\E \|V^H_{\tau}\| \notag \\
        &\le (16cp + 1) \,\E \max_{t \le n} \|V^H_t\| \,,
    \end{align}
    as desired.

    Now we want to use \Cref{lemma:tangentMax} to upper bound the max on the right-hand side by $\E \max_{t \le n} \|\bar{V}^H_t\|$ (and consequently by the sum $\E \|\bar{V}_1 + \ldots + \bar{V}_n\|$). For that, we claim that the sequences $\|V^H_1\|, \ldots, \|V^H_n\|$ and $\|\bar{V}^H_1\|, \ldots, \|\bar{V}^H_n\|$ are tangent with respect to the filtration $(\F_t)_t$: conditioning on the history $\F_{t-1}$ up to time $t-1$ fixes the next probe $I_t$ but still leaves $X^H_{I_t}$ and $\bar{X}^H_{I_t}$ independent and with the same distribution; thus, conditioned on $\F_{t-1}$, $\|V^H_t\| = \|e_{I_t} X^H_{I_t}\|$ and $\|\bar{V}^H_t\| = \|e_{I_t} \bar{X}^H_{I_t}\|$ have the same distribution. Thus, employing \Cref{lemma:tangentMax} we obtain 
    \begin{align*}
        \E \|V^H_1 + \ldots + V^H_n\| \,\le\, (16cp + 1) \,\E \max_{t \le n} \|\bar{V}^H_t\| \,\le\, (16cp + 1) \,\E \|\bar{V}_1 + \ldots + \bar{V}_n\|\,.
    \end{align*}
    This proves \Cref{lemma:highPart}.
    
%#################################################

    \subsubsection{Low part}

    We now consider the low part of the items and prove the following; let $\adapt^L := \E \|V^L_1 + \ldots + V^L_n\|$ denote the value obtained by the low part of the items. 

     \begin{lemma} \label{lemma:lowPart}
         If $\adapt^L \ge \frac{\adapt}{2}$, then $$\textrm{hallucinating}^H  \,=\, \E \|\bar{V}^H_1 + \ldots + \bar{V}^H_n\| \,\ge\, \frac{\adapt}{O(p)}.$$
     \end{lemma}
    
    Let $\tau$ be the first time when $\|V^L_1 + \ldots + V^L_{\tau}\| \ge 2 \adapt$ (let $\tau = n$ if no such time exists). We show that this truncated reward is $\Omega(\adapt^L)$.

    \begin{claim}
        In every scenario we have the upper bound $\|V^L_1 + \ldots + V^L_{\tau}\| \le 3\, \adapt$, and we have the lower bound in expectation $\E \|V^L_1 + \ldots + V^L_{\tau}\| \ge \frac{1}{2}\, \adapt^L$. 
    \end{claim}

    \begin{proof}
    Since the low part of each item is at most $\frac{\adapt}{16cp} \le \adapt$, we have $\|V^L_1 + \ldots + V^L_{\tau}\| \le \|V^L_1 + \ldots + V^L_{\tau - 1}\| + \|V^L_{\tau}\| \le 2 \adapt + \adapt \le 3 \adapt$, giving the upper bound. 

    For the lower bound, 
    \begin{align*}
    \adapt^L &= \E \|V^L_1 + \ldots + V^L_n\| \le \E \|V^L_1 + \ldots + V^L_{\tau}\| + \E \|V^L_{\tau + 1} + \ldots + V^L_n\|\\
    &= \E \Big( \|V^L_1 + \ldots + V^L_{\tau}\| \cdot \ones(\tau = n)\Big) + \E \Big( \|V^L_1 + \ldots + V^L_{\tau}\| \cdot \ones(\tau < n)\Big) \\
    &\qquad \qquad \qquad \qquad \qquad \qquad \qquad \quad\, + \E \Big(\|V^L_{\tau + 1} + \ldots + V^L_n\| \cdot \ones(\tau < n) \Big).
    \end{align*}
    The second term is at least $2 \adapt \cdot \Pr(\tau < n)$ and, by \Cref{obs:optSubstructure}, the last term is at most $\adapt \cdot \Pr(\tau < n)$, i.e., half the second term. Thus, we have the upper bound 
    \begin{align*}
    \adapt^L &\le \E \Big( \|V^L_1 + \ldots + V^L_{\tau}\| \cdot \ones(\tau = n)\Big) + 2\, \E \Big( \|V^L_1 + \ldots + V^L_{\tau}\| \cdot \ones(\tau < n)\Big) \\
    &\le 2\, \E  \|V^L_1 + \ldots + V^L_{\tau}\|,
    \end{align*}
    proving the second part of the claim. 
    \end{proof}

%    Now, consider the instance of the problem where: 1) We only consider the low values of the items, i.e., the value of item $j$ is distributed according to $X^L_j$; 2) Again we pad by adding another $n$ items of value 0.

    Let $\widehat{V}_1, \widehat{V}_2, \ldots, \widehat{V}_n$ denote the process $V_1^L, V_2^L, \ldots, V_\tau^L$ again padded with additional 0-value items so that we always have $n$ probes. \mnote{Should improve: formally we have to phrase in terms of an instance of the problem to use the result from the body, not in terms of vectors} Let $\widehat{\adapt} = \E \|\widehat{V}_1 + \ldots + \widehat{V}_n\| = \adapt^L$. Thus, under the assumption $\frac{\adapt}{2} \le \adapt^L$, and from the previous claim we have $\frac{\adapt^L}{2} \le \widehat{\adapt}$, so we have:

    \begin{enumerate}[itemsep=1pt]
        \item For every $t$, $\|\widehat{V}_t\| \le \frac{\adapt}{16cp} \le \frac{\widehat{\adapt}}{4cp}$.

        \item $\|\widehat{V}_1 + \ldots + \widehat{V}_n\| \le 3\,\adapt \le 12 \,\widehat{\adapt}$. 

        \item There is the same number of probes in every scenario. 
    \end{enumerate}

    Thus, this instance satisfies the assumptions from \Cref{sec:stochProbing}. The argument in that section then proves that $\E \|\bar{V}^L_1 + \ldots + \bar{V}^L_n\| \ge \frac{\widehat{\adapt}}{O(p)} \ge \frac{\adapt}{O(p)}$.\mnote{Should improve: because of the extra padding, $n$ has been redefined, but it's ok} This proves \Cref{lemma:lowPart}.

	\section{Low-Regret Algorithm for Online Linear Optimization} \label{app:regret}

    Recall the (non-negative) Online Linear Optimization (\OLO) problem~\cite{Hazan-Book16}: A convex set $P \subseteq \R^d_+$ is given upfront, and objective functions $g_1,g_2,\ldots,g_T$ are revealed one-by-one in an online fashion. In each time step $t$, the algorithm needs to produce a point $x_t \in P$ using the information revealed up to this moment; only \emph{after} that, the  adversary reveals a gain vector $g_t \in [0,1]^d$, and the algorithm receives gain $\ip{g_t}{x_t}$. The goal of the algorithm is to maximize its total gain $\sum_{t = 1}^T \ip{g_t}{x_t}$. We are interested in comparing favorably to the gains of the best fixed action in hindsight, namely $\OPT := \max_{x \in P} \sum_{t = 1}^T \ip{g_t}{x}$.

    %Let $P^{\downarrow} := \{x \in \R^d_+ : x \le y \textrm{ for some $y \in P$}\}$ be the downward closure of $P$. 
    Notice that because the gain vectors are non-negative, we can assume without loss of generality that $P$ is downward closed, i.e. we can include in it all points $x \in \R^d_+$ such that $x \le y$ for some $y \in P$: these points $x$ do not change $\OPT$, and if an algorithm uses any of these points one can just replace it by a bigger point $y \in P$ and improve its total gain. By rescaling the coordinates, we can also assume that the norm $\|\cdot\|_{P,\star}$ (dual to the norm $\|\cdot\|_{P}$) satisfies $\|e_i\|_{P,\star} = 1$ for every canonical vector $e_i$.

%    \mnote{The setup is a generalization of that of ``Hedging structured concepts'' \url{http://citeseerx.ist.psu.edu/viewdoc/download?doi=10.1.1.182.346&rep=rep1&type=pdf}. Issue: if approximate the norm even by a constant factor, gets non-sublinear regret}

    %Let $P^\circ := \{x : \ip{x}{y} \le 1 \textrm{ for all $y \in P$}\}$ be the polar of $P$. 
    We show that the $p$-\monotonicity of the dual norm $\|\cdot\|_{P,\star}$ guarantees an algorithm with good gains. We note that the term $\|\ones\|_{P,\star}$ equals $\max_{x \in P} \ip{\ones}{x}$, and so it can be thought as the ``width'' of the set $P$ in the direction of the largest possible gain vector $\ones$; thus, this term is a version of the standard (and necessary) ``diameter'' parameter present in, e.g., Online Gradient Descent.

    \begin{thm} \label{thm:pMono}
        Consider the \OLO problem where the set $P$ is downward closed. Assume that the dual norm $\|\cdot\|_{P,\star}$ is $p$-\monotone, differentiable over $\R^d_+ \setminus \{0\}$, and satisfies $\|e_i\|_{P,\star} = 1$ for all $i \in [d]$. Then for every $\e > 0$ there is a strategy that obtains total gains 
		\begin{align*}
			\sum_t \ip{g_t}{x_t} \,\ge\, e^{-\e} \bigg(\OPT - \frac{p \cdot (\|\ones\|_{P,\star}-1)}{\e}\bigg).
		\end{align*}
	\end{thm}

    As an example, consider the prediction with experts setting, where $P = \{x \in \R^d_+ : \sum_i x_i \le 1\}$. Let $p = \frac{\log d}{\log (1+\e)}$, and $q$ be its H\"older dual, i.e. $\frac{1}{p} + \frac{1}{q} = 1$. We can approximate the set $P$ set by $P' = \{x \in \R^d_+ : \|x\|_q \le 1\}$, which satisfies $P \subseteq P' \subseteq d^{1 - 1/q}\, P = (1+\e) P$, and has dual norm $\|\cdot\|_{P',\star}$ equal to the $\ell_p$ norm. Applying the previous theorem to $P'$, we get a sequence $x'_1,\ldots,x'_T \in T$ with total gain at least $e^{-\e} \,\OPT - \frac{p \cdot (d^{1/p} - 1)}{\e} \gtrsim (1-\e) \OPT - \frac{\log d}{\e}$; the rescaled sequence $x_t := \frac{x'_t}{d^{1-1/q}} = \frac{x'_t}{1+\e}$, which now belongs to the feasible set $P$, has similar total gains. This recovers the standard multiplicative/additive approximation for this experts setting.

    \begin{proof}[Proof of \Cref{thm:pMono}]
	To simplify the notation, let $s_t = g_1 + \ldots + g_t$ be the sum of the gain vectors up to time $t$. Notice that $\OPT = \max_{x \in P} \ip{x}{s_T} = \|s_T\|_{P,\star}$. However, we will work instead with a smoother function $f$ instead of $\|\cdot\|_{P,\star}$.
 
    More precisely, define the function $f(x) := \|x + \frac{p \ones}{\e}\|_{P,\star}$. The strategy for our algorithm is to play, at time $t$, the point $x_t = \nabla f(s_{t-1})$. Recall that $\nabla f(s_{t-1}) = \nabla \|s_{t-1} + \frac{p \ones}{\e}\|_{P,\star} = \argmax_{x \in P} \ip{s_{t-1} + \frac{p \ones}{\e}}{x}$, so this can be thought as a Follow the Perturbed Leader strategy, but with a deterministic shift $+ \frac{p \ones}{\e}$. In particular, notice that $x_t \in P$. 
    
    We now show that this strategy has good value; we will track $f$ instead of $\|\cdot\|_{P,\star}$. By convexity of $f$, 
	\begin{align}
		f(s_T) - f(0) = \sum_t \bigg(f(s_t) - f(s_{t-1})\bigg) \le \sum_t \ip{\nabla f(s_t)}{g_t}. \label{eq:regret1}
	\end{align}
    The next lemma shows that the $p$-\monotonicity of $\|\cdot\|_{P,\star}$ implies a stability of the gradients of $f$, which will be used to replace $s_t$ for $s_{t-1}$ in the right-hand side of the previous inequality (i.e., it gives a ``Be the Leader vs Follow the Leader'' comparison). 
	
	\begin{lemma}
		For every $s \ge 0$ and $g \in [0,1]^d$, we have $$\nabla f(s+g) \le e^{\e} \cdot \nabla f(s).$$
	\end{lemma}
	
	\begin{proof}
        To simplify the notation, let $h(\cdot) := \|\cdot\|_{P,\star}$. Since $\nabla f(x) = \nabla h(x+ \frac{p \ones}{\e})$, it suffices to show $\nabla h(s + g + \frac{p\ones}{\e}) \le e^{\e} \cdot \nabla h(s + \frac{p \ones}{\e})$. 
        
        By our assumption on $g$ we have that $s+ g + \frac{p \ones}{\e} \le (1+\frac{\e}{p})\cdot (s + \frac{p \ones}{\e})$, and  $p$-monotonicity implies
		\begin{align}
			\nabla (h^p)\big(s+g+\tfrac{p \ones}{\e}\big) \le \nabla (h^p)\big((1 + \tfrac{\e}{p}) \cdot (s+\tfrac{p \ones}{\e})\big). \label{eq:regretDom}
		\end{align}
		Moreover, from chain rule
		\begin{align}
			\nabla (h^p)(x) = p \cdot h(x)^{p-1} \cdot \nabla h(x), \label{eq:chainP}
		\end{align}
		and so
		\begin{align*}
			\nabla (h^p)\big((1 + \tfrac{\e}{p}) \cdot (s+\tfrac{p \ones}{\e})\big) &\le p \bigg(1 + \frac{\e}{p}\bigg)^{p-1} h\big(s+\tfrac{p \ones}{\e}\big)^{p-1} \cdot \nabla h\big((1 + \tfrac{\e}{p}) \cdot (s+\tfrac{p \ones}{\e})\big) \\
			&= p \bigg(1 + \frac{\e}{p}\bigg)^{p-1} h\big(s+\tfrac{p \ones}{\e}\big)^{p-1} \cdot  \nabla h\big(s+\tfrac{p \ones}{\e}\big)\\
			&= \bigg(1 + \frac{\e}{p}\bigg)^{p-1} \nabla (h^p)\big(s+\tfrac{p \ones}{\e}\big),
		\end{align*}
		where the first equation the gradient of any norm is invariant with respect to positive scaling. Plugging on \eqref{eq:regretDom} gives
		\begin{align*}
			\nabla(h^p)(s+g+\tfrac{p \ones}{\e}) \,\le\, e^\e \cdot \nabla(h^p)(s+\tfrac{p \ones}{\e}).
		\end{align*}
		Isolating $\nabla h(x)$ on \eqref{eq:chainP} an using this inequality we have 
		\begin{align*}
			\nabla h(s+g+\tfrac{p \ones}{\e}) \le e^\e \cdot \frac{\nabla (h^p)(s+\tfrac{p \ones}{\e})}{p \cdot h(s+g+\tfrac{p \ones}{\e})^{p-1}} \le e^\e \cdot \frac{\nabla (h^p)(s+\tfrac{p \ones}{\e})}{p \cdot h(s+\tfrac{p \ones}{\e})^{p-1}} = e^\e \cdot \nabla h(s+\tfrac{p \ones}{\e}),
		\end{align*}
		where the second inequality is because the norm $h(\cdot) = \|\cdot\|_{P,\star}$ is monotone and $g \ge 0$. This proves the lemma.
	\end{proof}
	
	Based on this lemma, we have $\nabla f(s_t) \le e^{\e} \cdot \nabla f(s_{t-1}) = e^{\e} \cdot x_t$, which applied in \eqref{eq:regret1} gives
	\begin{align}
		\sum_t \ip{g_t}{x_t} \ge e^{-\e} \cdot \sum_t \ip{g_t}{\nabla f(s_t)} \ge e^{-\e} \cdot \big(f(s_T) - f(0)\big)\,. \label{eq:regretFinal}
	\end{align}
        We now lower bound $f(s_T)$. Using the assumption that $\|e_i\|_{P,\star}$ and the monotonicity of this norm, we see that $\frac{(s_T)_i}{\|s_T\|_{P,\star}} \le \frac{(s_T)_i}{\|(s_T)_i e_i\|_{P,\star}} = 1$, and so $\ones \ge \frac{s_T}{\|s_T\|_{P,\star}}$. This gives
        \begin{align*}
        f(s_T) = \bigg\|s_T + \frac{p \ones}{\e}\bigg\|_{P,\star} \ge \bigg\|\bigg(1 + \frac{p}{\e \|s_T\|_{P,\star}}\bigg) s_T \bigg\|_{P,\star} = \|s_T\|_{P,\star} + \frac{p  \|s_T\|_{P,\star}}{\e \|s_T\|_{P,\star}} = \OPT + \frac{p}{\e}.
        \end{align*}
        Since $f(0) = \frac{p}{\e} \|\ones\|_{P,\star}$, plugging these bound on \eqref{eq:regretFinal} gives $\sum_t \ip{g_t}{x_t} \ge e^{-\e} (\OPT - \frac{p (\|\ones\|_{P,\star} - 1)}{\e})$. This proves \Cref{thm:pMono}.
 \end{proof}
	
%	Then setting $\e \approx \sqrt{\frac{p \|\ones\|}{\OPT}}$ should give the theorem, \red{triple check this, specially when $\OPT$ is very small. Maybe in that case need to run another algorithm} (or better, assume $\OPT$ not too small, just say that need to run another algo otherwise, etc.)

\pagebreak

\begin{small}
\bibliographystyle{alpha}
\bibliography{bib,online-lp-short}
\end{small}

\end{document}